\documentclass[
  a4paper,
  runningheads,
  orivec
]{llncs}

\usepackage{graphicx} 
\usepackage[hidelinks]{hyperref}
\hypersetup{
  colorlinks   = true,
  urlcolor     = blue,
  linkcolor    = blue,
  citecolor    = red
}
\usepackage{amsmath,amssymb}
\usepackage{stmaryrd}
\usepackage{wasysym}
\usepackage{mathtools}
\usepackage{subcaption}
\usepackage{thm-restate}
\usepackage{multirow}
\usepackage[nameinlink]{cleveref}
\usepackage{lipsum}

\usepackage[basic]{complexity}
\usepackage{csquotes}
\usepackage{xparse}
\usepackage{xspace}
\usepackage[dvipsnames,table]{xcolor}
\usepackage[ruled,vlined,linesnumbered]{algorithm2e}
\usepackage[inline]{enumitem}
\usepackage[colorinlistoftodos,prependcaption,textsize=tiny]{todonotes}
\usepackage{tablefootnote}
\usepackage{booktabs}
\usepackage{orcidlink}

\usepackage[mathlines]{lineno}
\newcommand*\patchAmsMathEnvironmentForLineno[1]{%
  \expandafter\let\csname old#1\expandafter\endcsname\csname #1\endcsname
  \expandafter\let\csname oldend#1\expandafter\endcsname\csname end#1\endcsname
  \renewenvironment{#1}%
  {\linenomath\csname old#1\endcsname}%
  {\csname oldend#1\endcsname\endlinenomath}}%
\newcommand*\patchBothAmsMathEnvironmentsForLineno[1]{%
  \patchAmsMathEnvironmentForLineno{#1}%
  \patchAmsMathEnvironmentForLineno{#1*}}%
\AtBeginDocument{%
  \patchBothAmsMathEnvironmentsForLineno{equation}%
  \patchBothAmsMathEnvironmentsForLineno{align}%
  \patchBothAmsMathEnvironmentsForLineno{flalign}%
  \patchBothAmsMathEnvironmentsForLineno{alignat}%
  \patchBothAmsMathEnvironmentsForLineno{gather}%
  \patchBothAmsMathEnvironmentsForLineno{multline}%
}

\let\oldqed\qed
\renewcommand{\qed}{~\hfill$\oldqed$}

\colorlet{tablehighlight}{lightgray!40!white}

\setlist{nosep}

\spnewtheorem{construction}{Construction}{\itshape}{\rmfamily}
\crefname{construction}{construction}{constructions}
\Crefname{construction}{Construction}{Constructions}

\newcommand{\Codeterministic}{Codeterministic}
\newcommand{\codeterministic}{codeterministic}
\newcommand{\Codeterminism}{Codeterminism}
\newcommand{\codeterminism}{codeterminism}
\newcommand{\Complete}{Complete}
\newcommand{\complete}{complete}

\newcommand{\procedural}{procedural}
\newcommand{\Procedural}{Procedural}
\newcommand{\proceduralAlphabet}{\procedural\ alphabet}
\newcommand{\ProceduralAlphabet}{\Procedural\ alphabet}
\newcommand{\proceduralTransition}{\procedural\ transition}

\newcommand{\proceduralSymbol}{\procedural\ symbol}

\newcommand{\linkingFunction}{linking function\xspace}

\newcommand{\recursiveLanguage}{recursive language}

\newcommand{\regularLanguage}{regular language}
\newcommand{\regularRun}{regular run}
\newcommand{\RecursiveRun}{Recursive run}
\newcommand{\recursiveRun}{recursive run}

\newcommand{\internal}{\mathit{int}}
\newcommand{\call}{\mathit{call}}
\newcommand{\return}{\mathit{ret}}
\newcommand{\proc}{\mathit{proc}}
\newcommand{\intAlpha}{\Sigma_{\internal}}
\newcommand{\calAlpha}{\Sigma_{\call}}
\newcommand{\retAlpha}{\Sigma_{\return}}
\newcommand{\procAlpha}{\Sigma_{\proc}}
\newcommand{\intTrans}{\delta_{\internal}}
\newcommand{\callTrans}{\delta_{\call}}
\newcommand{\retTrans}{\delta_{\return}}
\newcommand{\pdwAlph}{\reg{\Sigma}}

\newcommand{\AutomataSet}{\Lambda}

\newcommand{\vraStates}[1]{Q_\mathcal{#1}}
\newcommand{\vraInit}[1]{I_\mathcal{#1}}
\newcommand{\vraTrans}[1]{\delta_\mathcal{#1}}
\newcommand{\Lang}[1]{\reg{L}(#1)}
\newcommand{\regLang}[1]{L(#1)}

\newcommand{\runsA}[1]{\Pi({#1})}

\newcommand{\reg}[1]{\widetilde{#1}}
\newcommand{\wm}[1]{\mathit{WM}(#1)}
\newcommand{\depth}[1]{\mathit{depth}(#1)}

\let\oldxrightarrow\xrightarrow
\newcommand\myxrightarrow[2][]{
    \oldxrightarrow[#1]{{\raisebox{-0.2ex}[0pt][0pt]{$\scriptstyle#2$}}}%
}
\let\xrightarrow\myxrightarrow

\AtEndEnvironment{example}{\hfill$\lrcorner$}

\newcommand{\arrow}{[fill={black}  ][line width=0.07]  [draw opacity=0] (8,-3) -- (0,0) -- (8,3) -- cycle }

\title{Visibly Recursive Automata%
\thanks{K. Dubrulle is a {FRIA grantee} of the Belgian \emph{Fonds de la Recherche Scientifique}--FNRS; G. A. P\'erez is supported by the FWO ``SynthEx'' project (G0AH524N).}}
\author{K\'evin Dubrulle\inst{1,2}\orcidlink{0009-0000-0298-1349}, V\'eronique Bruy\`ere\inst{1}\orcidlink{0000-0002-9680-9140},  Guillermo A. P\'erez\inst{2}\orcidlink{0000-0002-1200-4952}, and Ga\"etan Staquet\inst{3}\orcidlink{0000-0001-5795-3265}}
\authorrunning{K. Dubrulle, V. Bruy\`ere, G. A. P\'erez, G. Staquet}

\institute{Université de Mons, Mons, Belgium\\\email{\{kevin.dubrulle,veronique.bruyere\}@umons.ac.be}
\and
Universiteit Antwerpen, Antwerp, Belgium\\\email{guillermo.perez@uantwerpen.be}\and
Nantes Université, École Centrale Nantes, CNRS, LS2N, Nantes, France\\\email{gaetan.staquet@ec-nantes.fr}}

\begin{document}

\maketitle

\begin{abstract}
As an alternative to visibly pushdown automata, we introduce visibly recursive automata (VRAs), composed of a set of classical automata that can call each other. VRAs are a strict extension of so-called systems of procedural automata, a model proposed in 2021 by Frohme and Steffen. We study the complexity of standard language-theoretic operations and classical decision problems for VRAs. Since the class of deterministic VRAs forms a strict subclass in terms of expressiveness, we propose a (weaker) notion that does not restrict expressive power and that we call codeterminism. Codeterminism comes with many desirable algorithmic properties that we demonstrate by using it, e.g., as a stepping stone towards implementing complementation of VRAs.

\keywords{Visibly pushdown languages \and Modular model of computation  \and Automata theoretic properties}
\end{abstract}

\section{Introduction}

Almost all computer systems are programmed by defining multiple functions that
call each other, sometimes recursively. When modeling such systems --- for verification and model checking purposes, for instance --- it is thus important to take into
account the \emph{calls} to functions
(i.e., jump into a different part of the code)
and their \emph{returns}
(i.e., jump back to the position immediately following that before the corresponding call). These behaviors
can be represented by context-free grammars or, equivalently,
pushdown automata (PDAs)~\cite{HU79}.

While \emph{context-free languages} (CFLs), i.e., the family of languages accepted by PDAs, have
model checking tools~\cite{AlurBEGRY05,BouajjaniEM97,NguyenT17,SongT14},
many properties of interest for CFLs are undecidable, for example: checking the equivalence of two PDAs and
universality of a PDA.
Thus, along the years,
some restrictions have been considered to obtain positive decidability results, such as in~\cite{ChenW02,ChiariMPP23,ChiariMP21,EsparzaKS01}.
{Among these restrictions, here, we focus on \emph{visibly pushdown languages} (VPLs), recognized by  \emph{visibly pushdown automata} (VPAs) \cite{alur2004vpl} and \emph{nested word automata} \cite{alur2007marrying,alur2009adding}. We consider only VPAs, as they are more commonly used in the literature.}

VPAs split
the alphabet into three disjoint subsets:
one set of symbols is only used for \emph{calls} whose reading triggers a push on the stack, a second set only for \emph{returns} whose reading pops the top of the stack, and the last set contains the
\emph{internal} symbols with no influence on the stack of calls. 
Thus, it is the type of  symbol that dictates the operation to be applied on the stack.
Thanks to this restriction, 
VPLs are closed under several operations and some problems that are undecidable for CFLs become decidable for VPLs.
For example, VPLs are closed under union,
intersection, and complementation~\cite{alur2004vpl}.
Also, emptiness,
universality,
and language equality and inclusion are all decidable for VPAs.

VPAs have been used in practice to verify XML and JSON documents
against their schemas in a streaming context~\cite{BruyerePS23,KumarMV07,LeGlaunecLM25}.
In particular, in~\cite{BruyerePS23}, we learned (in Angluin's active-learning setup~\cite{Angluin87,Isberner15})
a VPA modeling
a given JSON schema from a sample of good and bad documents. 
We implemented
our algorithm in a prototype tool, showing that an automaton-based
approach is feasible, despite the size of the learned VPA.
However, VPAs suffer from a significant drawback: they tend to be
large and, thus, complex to construct and learn. 
While 
a variation of VPA specialized for JSON schemas is studied in~\cite{LeGlaunecLM25},
we conjecture that it would be more efficient to build an  
automaton defined as a \emph{collection of smaller automata}, 
rather than a single large automaton.
In~\cite{Dubrulle2025}, we studied the validation of JSON documents using such a model
(introduced in the sequel) and we obtained \cite{Dubrulle25Code} much smaller automata and faster validation times. 
Whilst other works treat models that strictly include
VPAs~\cite{AlurBEGRY05,gallier2003dfarecursif,woods1970transition},
we here focus on \enquote{modularizing} VPAs.
Previous work considered two ways for modularization: \emph{\(k\)-module single-entry automata}
(\(k\)-SEVPAs)~\cite{AlurKMV05} and \emph{systems of procedural automata}
(SPAs)~\cite{frohme2021spa}.

The set of call symbols of a \(k\)-SEVPA is partitioned into \(k\) classes, the automaton has a main 
module and \(k\) distinct interconnected submodules, one for each class. The transitions labeled by call and return symbols manage these interconnections, by hard-coding the stack manipulations. While VPAs and \(k\)-SEVPAs are equivalent (for any value of \(k\)) and there exists a unique minimal \(k\)-SEVPA for a given language, this minimal automaton may
have exponentially more states than a VPA accepting the same language.
Nonetheless, this family has some active learning algorithms~\cite{Isberner15,jia2024v}.
Our VPA-based approach from~\cite{BruyerePS23} actually relies on \(1\)-SEVPAs. Due to their sizes, our algorithm for JSON documents is slower than state-of-the-art
JSON validators, as highlighted in~\cite{LeGlaunecLM25}.

The submodules of an SPA are classical finite automata (FAs) that are not interconnected, and
the transitions do not manipulate the stack directly. Instead, each call symbol is associated with a specific automaton.
Whenever the call symbol corresponding to the FA \(\mathcal{A}\) is read by the SPA, a call to this automaton is performed, with a jump to its initial state. Later on, when a sub-word is accepted by $\mathcal{A}$ and followed by a return symbol, the SPA goes back to the state from which it reads the call symbol.
While a stack is maintained to remember the calls, it is handled purely
by the semantics of SPAs, i.e., none of the transitions explicitly push or pop.
In~\cite{frohme2021spa}, an active learning algorithm for SPAs is also presented.
Interestingly, SPAs are strictly less expressive than VPAs.

\subsubsection{Contributions.}
In this paper, we introduce a new kind of modular automata we call
\emph{visibly recursive automata} (VRAs). Similarly to SPAs~\cite{frohme2021spa},
a VRA is composed of multiple FAs. However, in contrast to that work,
we lift the restriction that each call symbol corresponds to a
unique automaton. Instead, we allow multiple FAs to share a common
call symbol. We argue that VRAs form a strict superset of SPAs,
and we show that they are equivalent to VPAs with 
polynomial size translations (see \autoref{th:vpa-vra}). 
    

We claim that, like for~\cite{frohme2021spa},
since we construct smaller FAs, each serving a specific purpose,
VRAs are easier to construct and understand: one can focus
on each part of the system individually.
This is much closer to
how programs are engineered and implemented: each function serves a
specific role and can call other functions to achieve its goal. It also mirrors the way in which JSON schemas are structured,\footnote{See \url{https://json-schema.org/understanding-json-schema/structuring}.} i.e., in a modular way.
Furthermore, we conjecture that this decomposition will allow for more efficient learning than what is possible for VPAs (each part is an FA that can be learned more easily than VPAs~\cite{Angluin87,Isberner15,jia2024v}) and enable a modular and compositional learning algorithm more generally applicable than that for SPAs (an open challenge of active automata learning \cite{Fortz2026ActiveAutomataLearning}).

\begin{table}[b]
    \centering    
    \vspace{-2em}
    \caption{Summary of our complexity results, and comparison with known bounds for VPAs, where \(|\mathcal{A}|\) is the number of states and transitions of either a VRA, or a VPA \(\mathcal{A}\). We highlight 
    where VRAs 
    have better complexity. 
    }%
    \scriptsize
    \begin{tabular}{
        @{}
        c @{\hspace{5pt}}
        l >{\hspace{12pt}}
        c >{\hspace{8pt}}
        c @{}
    }
        \toprule
        &
            &
            \textbf{VRA} &
            \textbf{VPA} \cite{alur2004vpl}
        \\
        \midrule
        \multirow{5}{*}{\rotatebox[origin=c]{90}{\parbox{45pt}{\centering Operations\\(Result size)}}} &
            Concatenation (\(\Lang{\mathcal{A}_1} \cdot \Lang{\mathcal{A}_2}\)) &
            \(\mathcal{O}(|\mathcal{A}_1| + |\mathcal{A}_2|)\) &
            \(\mathcal{O}(|\mathcal{A}_1| + |\mathcal{A}_2|)\)
        \\
        &
            Kleene-$*$ (\(\Lang{\mathcal{A}_1}^*\)) &
            \(\mathcal{O}(|\mathcal{A}_1|)\) &
            \(\mathcal{O}(|\mathcal{A}_1|)\)
        \\
        &
            Union (\(\Lang{\mathcal{A}_1} \cup \Lang{\mathcal{A}_2}\)) &
            \(\mathcal{O}(|\mathcal{A}_1| + |\mathcal{A}_2|)\) &
            \(\mathcal{O}(|\mathcal{A}_1| + |\mathcal{A}_2|)\)
        \\
        &
            Intersection (\(\Lang{\mathcal{A}_1} \cap \Lang{\mathcal{A}_2}\)) &
            \(\mathcal{O}(|\mathcal{A}_1| \cdot |\mathcal{A}_2|)\) &
            \(\mathcal{O}(|\mathcal{A}_1| \cdot |\mathcal{A}_2|)\)
        \\
        &
            \cellcolor{tablehighlight}Complementation (\(\overline{\Lang{\mathcal{A}_1}}\)) &
            \cellcolor{tablehighlight}\(2^{\mathcal{O}(|\mathcal{A}_1|)}\) &
            \cellcolor{tablehighlight}\(2^{\mathcal{O}(|\mathcal{A}_1|^2)}\)
        \\
        \midrule
        \multirow{4}{*}{\rotatebox[origin=c]{90}{\parbox{40pt}{\centering Decision\\problems\\(Runtime)}}} &
            \cellcolor{tablehighlight}Emptiness ($\Lang{\mathcal{A}_1} \stackrel{?}{=} \varnothing$) &
            \cellcolor{tablehighlight}\(\ \mathcal{O}(|\mathcal{A}_1|)\) &
            \cellcolor{tablehighlight}\(\mathcal{O}(|\mathcal{A}_1|^3)\) 
        \\
        &
            \cellcolor{tablehighlight}Universality ($\Lang{\mathcal{A}_1} \stackrel{?}{=} \wm\pdwAlph$) &
            \cellcolor{tablehighlight}\(2^{\mathcal{O}(|\mathcal{A}_1|)}\) &
            \cellcolor{tablehighlight}\(2^{\mathcal{O}(|\mathcal{A}_1|^2)}\)
        \\
        &
            \cellcolor{tablehighlight}Inclusion ($\Lang{\mathcal{A}_1} \stackrel{?}{\subseteq } \Lang{\mathcal{A}_2}$)&
            \cellcolor{tablehighlight}\(\mathcal{O}(|\mathcal{A}_1|)\cdot 2^{\mathcal{O}(|\mathcal{A}_2|)}\) &
            \cellcolor{tablehighlight}\(\mathcal{O}(|\mathcal{A}_1|^3)\cdot2^{\mathcal{O}(|\mathcal{A}_2|^2)}\) 
        \\
        &
            \cellcolor{tablehighlight}Equivalence ($\Lang{\mathcal{A}_1} \stackrel{?}{=} \Lang{\mathcal{A}_2}$) &
            \cellcolor{tablehighlight}\(2^{\mathcal{O}(|\mathcal{A}_1|+|\mathcal{A}_2|)}\) &
            \cellcolor{tablehighlight}\(2^{\mathcal{O}(|\mathcal{A}_1|^2+|\mathcal{A}_2|^2)}\)
        \\
        \bottomrule
    \end{tabular}
    \label{tab:vra-vpa}
    \vspace{-1em}
\end{table}



Our long-term objective is to obtain efficient active learning algorithms for VRAs. In this work, we focus on a first step in that direction. Namely,
we study the complexity of the usual language-theoretic operations for languages accepted by VRAs, as well as the classical decision problems for VRAs. Some of our algorithms leverage the interreduction between VRAs and VPAs, but mostly we provide direct algorithms with better complexity. Concerning the language-theoretic operations, our main result is the complementation closure that requires translating any VRA into a \emph{\codeterministic\footnote{We borrow terminology and draw inspiration from \cite{berstel2002balanced} for this notion.}\ and \complete} one (determinism does not help as deterministic VRAs form a strict subclass). \autoref{tab:vra-vpa}
summarizes our results. For the decision problems (see \autoref{th:decision-problem}), 
%
we highlight that the complexity for VRAs is consistently lower than for VPAs. 
For the operations over the languages (see \autoref{th:closure}), we obtain the same complexities as for VPAs with the exception of complementation, where we again get a lower complexity.




\section{Visibly Recursive Automaton Model}

In this section, we present the \emph{visibly recursive automaton} model and provide a comparison with some other models. Visibly recursive automata are composed of several classical \emph{finite automata} and accept \emph{well-matched words}. 

\subsection{Preliminaries}
\label{subsec:prelim}

\begin{definition}[Finite automaton]\label{def:finite-automaton} A \emph{finite automaton} (FA) is a tuple $\mathcal{A}=\langle\Sigma,Q,I,F,\delta\rangle$ where $\Sigma$ is the \emph{input alphabet}; $Q$, a finite set of \emph{states}; $I\subseteq Q$, a set of \emph{initial} states; $F\subseteq Q$, a set of \emph{final} states; and $\delta \subseteq Q\times \Sigma \times Q$, a set of \emph{transitions}. 
The \emph{size} of an FA $\mathcal{A}$, denoted by $|\mathcal{A}|$, is $|Q|+|\delta|$.
\end{definition}

We denote by $\regLang{\mathcal{A}}$ the \emph{language} of $\mathcal{A}$ composed of all accepted words.
An FA $\mathcal{A}$ is \emph{deterministic} (DFA) if $|I|=1$, and, for all $q\in Q$, $a\in \Sigma$, there is at most one transition $(q,a,p)\in \delta$. It is \emph{complete} if, for all $q\in Q$, $a\in \Sigma$, there exists a transition $(q,a,p)\in \delta$. 
Any FA $\mathcal{A}$ can be transformed into an equivalent complete DFA $\mathcal{B}$ with $|\mathcal{B}|=2^{\mathcal{O}(|\mathcal{A}|)}$ \cite{HU79}. The empty word is denoted $\varepsilon$.  

\begin{definition}[Well-matched word]\label{def:wm-words}
   A \emph{pushdown alphabet} $\pdwAlph = \intAlpha\cup \calAlpha\cup \retAlpha$ is the union of three pairwise disjoint finite alphabets, which are, respectively, the set of \emph{internal}, \emph{call}, and \emph{return} symbols. The set $\wm\pdwAlph$ of \emph{well-matched words} over $\pdwAlph$ is the smallest set satisfying:
    \begin{itemize}
        \item Let $w\in \intAlpha^*$, then $w\in \wm\pdwAlph$. 
        \item Let $w\in \wm\pdwAlph$, $c\in \calAlpha$ and $r\in \retAlpha$, then $c\cdot w\cdot r \in \wm\pdwAlph$.
        \item Let $w_1,w_2\in\wm\pdwAlph$, then $w_1\cdot w_2 \in \wm\pdwAlph$. 
    \end{itemize}
\end{definition}

The \emph{depth} of $w\in \wm\pdwAlph$, denoted by $\depth{w}$, is the deepest level of unmatched call symbols at any point in the word. 
Any $w\in \wm\pdwAlph$ can be decomposed as $w=u_0c_1w_1r_1u_1\dots c_nw_n r_n u_n$ for some $n\in \mathbb{N}$, with $u_i\in \intAlpha^*$, $c_i\in \calAlpha$, $r_i\in \retAlpha$ and $w_i\in \wm\pdwAlph$ such that $\depth{w_i} < \depth{w}$, for all~$i$.  
Note that if $w\in \intAlpha^*$, then $n = 0$ and $\depth{w}=0$. In this paper, we provide proofs using this decomposition or the structural induction of  \autoref{def:wm-words}.

Given $S \subseteq D$, we denote by $\overline{S}$ the complement of $S$ in $D$, i.e., $\overline{S} = D \setminus S$. 
By convention, the intersection over an empty family of subsets of $D$ is equal to $D$: if $(S_i)_{i \in I}$ is a family of subsets of $D$ and $I = \varnothing$, then $\bigcap_{i\in I} S_i = D$.

\subsection{
Visibly Recursive Automata}


A \emph{visibly recursive automaton} (VRA), inspired by the formalisms from \cite{frohme2021spa,gallier2003dfarecursif}, is composed of several FAs 
that can call each other by reading specific call symbols of a pushdown alphabet. 
Each of these FAs is identified by a unique \emph{\proceduralSymbol}. By convention, we use capital letters to denote \proceduralSymbol s. See \autoref{fig:ex-vra} for a first example.

\begin{definition}[\ProceduralAlphabet]\label{def:proc-alpha}
    A \emph{\proceduralAlphabet} $\procAlpha$ w.r.t. $\pdwAlph=\intAlpha\cup \calAlpha \cup \retAlpha$ 
    is a set of \emph{\proceduralSymbol s}. 
    With $\procAlpha$ we associate a \emph{\linkingFunction} $f: \procAlpha \to \calAlpha\times \retAlpha$.
    Let $f_{\call}$ and $f_{\return}$ be the functions such that $f(J)=\langle  f_{\call}(J) , f_{\return}(J)\rangle$, for all $J\in \procAlpha$.
\end{definition}


\begin{definition}[Visibly recursive automaton]\label{def:vra}
    A \emph{visibly recursive automaton (VRA)} is a tuple $\mathcal{A}= \langle\pdwAlph,\procAlpha,\AutomataSet,\mathcal{A}^S\rangle$, where:
    \begin{itemize}
        \item $\pdwAlph=\intAlpha \cup \calAlpha \cup \retAlpha$ is a pushdown alphabet;
        \item $\procAlpha$ is a \proceduralAlphabet\ w.r.t. $\pdwAlph$;
        \item $\AutomataSet = \{ \mathcal{A}^J\mid J\in\procAlpha\}$ is a \emph{set of finite automata} over $\intAlpha\cup \procAlpha$ such that $\mathcal{A}^J=\langle\intAlpha\cup \procAlpha, Q^J,I^J,F^J,\delta^J\rangle$ for each $J$;
        \item $\mathcal{A}^S=\langle\intAlpha\cup\procAlpha, Q^S,I^S,F^S,\delta^S\rangle$ is a \emph{starting automaton}.
    \end{itemize}
    We write $\vraStates{A}=\bigcup_{J\in\procAlpha\cup \{S\}} Q^J$ (resp. $\vraTrans{A}=\bigcup_{J\in\procAlpha\cup \{S\}} \delta^J$) the set of all states (resp. transitions) of a VRA $\mathcal{A}$.
    Its \emph{size}, denoted by $|\mathcal{A}|$, is $|\vraStates{A}|+|\vraTrans{A}|$.
    
\end{definition}

In this definition, we assume that $\mathcal{A}^S \notin \AutomataSet$, $S\notin \procAlpha$, and the sets of states $Q^J$, with $J \in \procAlpha \cup \{S\}$, are pairwise disjoint. A transition in $\delta^J$ on an internal (resp.\ \procedural) symbol is called an \emph{internal} (resp.\ \emph{\procedural}) transition.

A VRA $\mathcal{A}$ accepts words over 
$\pdwAlph$ as follows. 
The semantics of $\mathcal{A}$ use \emph{configurations} $\langle q,\sigma \rangle$ where $q\in \vraStates{A}$ is a state and $\sigma \in \vraStates{A}^*$ is a \emph{stack word} whose symbols are states of the VRA. A \emph{\recursiveRun} of $\mathcal{A}$ on a word $w=a_1\dots a_n \in \pdwAlph^*$ is a sequence $\pi = \langle q_0,\sigma_0 \rangle \xrightarrow{a_1} \langle q_1,\sigma_1 \rangle \xrightarrow{a_2} \dots \xrightarrow{a_n} \langle q_n,\sigma_n \rangle$, where for all $i\in [1,n]$: 
\begin{itemize}
    \item If $a_i \in \intAlpha$, there exists a  transition $q_{i-1} \xrightarrow{a_i} q_i \in \vraTrans{A}$ and $\sigma_i = \sigma_{i-1}$;
    \item If $a_i \in \calAlpha$, there exists a \proceduralSymbol\ $J\in \procAlpha$ such that $f_{\call}(J) = a_i$, $q_i \in I^J$, and there exists $q_{i-1}\xrightarrow{J}p \in \vraTrans{A}$ such that $\sigma_{i} = p \sigma_{i-1}$;\footnote{A symbol is pushed on the left of a stack word.}

    Hence, when reading $a_i \in \calAlpha$, the automaton $\mathcal{A}^J$ such that $f_{\call}(J)= a_i$ is called and there is a jump to an initial state $q_i$ of $\mathcal{A}^J$, while a state $p$ such that $q_{i-1}\xrightarrow{J}p$ is pushed on the stack word. 
    \item If $a_i \in \retAlpha$, there exists a \proceduralSymbol\ $J\in \procAlpha$ such that $f_{\return}(J) = a_i$, $q_{i-1} \in F^J$
    , and $\sigma_{i-1} = q_i  \sigma_{i}$.

    Hence, when reading $a_i \in \retAlpha$, if $q_{i-1}$ is a final state of $\mathcal{A}^J$ and $f_{\return}(J)= a_i$, the call to $\mathcal{A}^J$ is completed and the state $q_i$ is popped from the stack word.
\end{itemize}
See \autoref{ex:vra} below to better understand the semantics. 

We denote by $\runsA{\mathcal{A}}$ the set of all \recursiveRun s of $\mathcal{A}$. 
The \emph{\recursiveLanguage}\footnote{{This refers to a language defined in terms of a VRA and should not be confused with the class of \recursiveLanguage s\ in theory of computation.}} of an automaton $\mathcal{A}^J\in \AutomataSet\cup \{\mathcal{A}^S\}$, denoted by $\Lang{\mathcal{A}^J}$, is the set of words such that there exists an \emph{accepting} \recursiveRun\ on them, i.e., from an initial configuration $\langle q_i,\varepsilon \rangle$, with $q_i\in I^J$, to a final configuration $\langle q_f,\varepsilon \rangle$, with $q_f\in F^J$:
\[\Lang{\mathcal{A}^J} = \left\{w\in\pdwAlph^*\mid \exists q_i \in I^J, q_f \in F^J, \langle q_i,\varepsilon \rangle\xrightarrow{w}\langle q_f,\varepsilon \rangle\in\runsA{\mathcal{A}}\right\}.\]
Notice that \(f_{\call}(J)\) and \(f_{\return}(J)\) do not appear in the definition of $\Lang{\mathcal{A}^J}$. That is, a word of the \recursiveLanguage\ of $J$ can start (resp. end) with a symbol that is not \(f_{\call}(J)\) (resp. \(f_{\return}(J)\)).
The language of a VRA $\mathcal{A}$, denoted by $\Lang{\mathcal{A}}$, is the \recursiveLanguage\ of its starting automaton: $\Lang{\mathcal{A}}=\Lang{\mathcal{A}^S}$. 



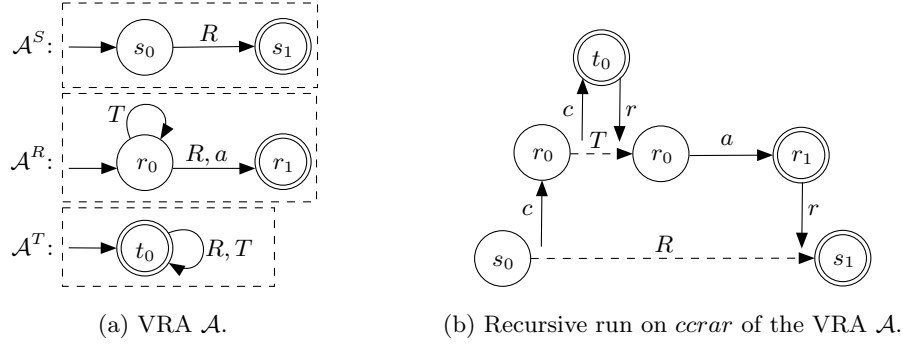
\begin{figure}[t]
    \centering
    \begin{subfigure}{.4\textwidth}
    \centering
    \begin{tikzpicture}[x=0.75pt,y=0.75pt,yscale=-1,xscale=1]

\draw    (72.33,48.95) -- (94.3,49.12) ;
\draw [shift={(96.3,49.13)}, rotate = 180.44]\arrow ;
\draw    (124.3,49.13) -- (163.3,48.66) ;
\draw [shift={(165.3,48.63)}, rotate = 179.3]\arrow ;
\draw [dashed]  (69.2,26.83) -- (196.5,26.83) -- (196.5,69.33) -- (69.2,69.33) -- cycle ;
\draw [dashed]  (69.2,72.33) -- (196,72.33) -- (196,126.83) -- (69.2,126.83) -- cycle ;
\draw [dashed]  (69.2,129.83) -- (175,129.83) -- (175,169.33) -- (69.2,169.33) -- cycle ;
\draw   (96.3,49.13) .. controls (96.3,41.4) and (102.57,35.13) .. (110.3,35.13) .. controls (118.03,35.13) and (124.3,41.4) .. (124.3,49.13) .. controls (124.3,56.87) and (118.03,63.13) .. (110.3,63.13) .. controls (102.57,63.13) and (96.3,56.87) .. (96.3,49.13) -- cycle ;
\draw   (165.3,48.63) .. controls (165.3,40.9) and (171.57,34.63) .. (179.3,34.63) .. controls (187.03,34.63) and (193.3,40.9) .. (193.3,48.63) .. controls (193.3,56.37) and (187.03,62.63) .. (179.3,62.63) .. controls (171.57,62.63) and (165.3,56.37) .. (165.3,48.63) -- cycle ;
\draw   (167.49,48.63) .. controls (167.49,42.11) and (172.78,36.82) .. (179.3,36.82) .. controls (185.82,36.82) and (191.11,42.11) .. (191.11,48.63) .. controls (191.11,55.16) and (185.82,60.45) .. (179.3,60.45) .. controls (172.78,60.45) and (167.49,55.16) .. (167.49,48.63) -- cycle ;
\draw    (72.33,149.95) -- (94.3,150.12) ;
\draw [shift={(96.3,150.13)}, rotate = 180.44]\arrow ;
\draw    (72.33,110) -- (94.3,110) ;
\draw [shift={(96.6,110)}, rotate = 180.44]\arrow ;
\draw   (96.3,107.13) .. controls (96.3,99.4) and (102.57,93.13) .. (110.3,93.13) .. controls (118.03,93.13) and (124.3,99.4) .. (124.3,107.13) .. controls (124.3,114.87) and (118.03,121.13) .. (110.3,121.13) .. controls (102.57,121.13) and (96.3,114.87) .. (96.3,107.13) -- cycle ;
\draw    (103,95) .. controls (93.25,71.11) and (129.13,73.37) .. (118.88,93.89) ;
\draw [shift={(118,95.5)}, rotate = 300.58]\arrow ;
\draw    (124.3,110) -- (163.3,110) ;
\draw [shift={(165.5,110)}, rotate = 179.3]\arrow ;
\draw   (165.3,106.63) .. controls (165.3,98.9) and (171.57,92.63) .. (179.3,92.63) .. controls (187.03,92.63) and (193.3,98.9) .. (193.3,106.63) .. controls (193.3,114.37) and (187.03,120.63) .. (179.3,120.63) .. controls (171.57,120.63) and (165.3,114.37) .. (165.3,106.63) -- cycle ;
\draw   (167.49,106.63) .. controls (167.49,100.11) and (172.78,94.82) .. (179.3,94.82) .. controls (185.82,94.82) and (191.11,100.11) .. (191.11,106.63) .. controls (191.11,113.16) and (185.82,118.45) .. (179.3,118.45) .. controls (172.78,118.45) and (167.49,113.16) .. (167.49,106.63) -- cycle ;
\draw   (96.3,150.13) .. controls (96.3,142.4) and (102.57,136.13) .. (110.3,136.13) .. controls (118.03,136.13) and (124.3,142.4) .. (124.3,150.13) .. controls (124.3,157.87) and (118.03,164.13) .. (110.3,164.13) .. controls (102.57,164.13) and (96.3,157.87) .. (96.3,150.13) -- cycle ;
\draw    (121.92,142.56) .. controls (145.82,132.84) and (144.47,168.72) .. (124.01,158.44) ;
\draw [shift={(122.2,157.56)}, rotate = 30.66]\arrow ;
\draw   (98.49,150.13) .. controls (98.49,143.61) and (103.78,138.32) .. (110.3,138.32) .. controls (116.82,138.32) and (122.11,143.61) .. (122.11,150.13) .. controls (122.11,156.66) and (116.82,161.95) .. (110.3,161.95) .. controls (103.78,161.95) and (98.49,156.66) .. (98.49,150.13) -- cycle ;

\draw (43,39.67) node [anchor=north west][inner sep=0.75pt]   [align=left] {$\displaystyle \mathcal{A}^{S}$:};
\draw (42,97.67) node [anchor=north west][inner sep=0.75pt]   [align=left] {$\displaystyle \mathcal{A}^{R}$:};
\draw (136.1,36.67) node [anchor=north west][inner sep=0.75pt]   [align=left] {$\displaystyle R$};
\draw (103.4,46) node [anchor=north west][inner sep=0.75pt]  [font=\small] [align=left] {$\displaystyle s_{0}$};
\draw (172.4,46) node [anchor=north west][inner sep=0.75pt]  [font=\small] [align=left] {$\displaystyle s_{1}$};
\draw (129.1,96.67) node [anchor=north west][inner sep=0.75pt]   [align=left] {$\displaystyle R,a$};
\draw (104.9,104) node [anchor=north west][inner sep=0.75pt]  [font=\small] [align=left] {$\displaystyle r_{0}$};
\draw (172.9,104) node [anchor=north west][inner sep=0.75pt]  [font=\small] [align=left] {$\displaystyle r_{1}$};
\draw (42.5,141.17) node [anchor=north west][inner sep=0.75pt]   [align=left] {$\displaystyle \mathcal{A}^{T}$:};
\draw (104.4,146.67) node [anchor=north west][inner sep=0.75pt]  [font=\small] [align=left] {$\displaystyle t_{0}$};
\draw (139.1,144.67) node [anchor=north west][inner sep=0.75pt]   [align=left] {$\displaystyle R,T$};
\draw (90.6,77) node [anchor=north west][inner sep=0.75pt]   [align=left] {$\displaystyle T$};

\end{tikzpicture}
    \vspace{-2em}
    \caption{VRA $\mathcal{A}$.}
    \label{fig:ex-vra}
    \end{subfigure}
    \hfill
    \begin{subfigure}{.5\textwidth}
    \centering
    \begin{tikzpicture}[x=0.75pt,y=0.75pt,yscale=-1,xscale=1]

\draw [dashed]   (455,178.3) -- (594.8,178.14) ;
\draw [shift={(596.8,178.13)}, rotate = 179.93] \arrow  ;
\draw    (590.67,139.97) -- (590.33,171.32) ;
\draw [shift={(590.31,173.32)}, rotate = 270.61] \arrow  ;
\draw   [dashed] (474.79,125.63) -- (504,125.63) ;
\draw [shift={(506,125.63)}, rotate = 180] \arrow  ;
\draw    (534.3,126.51) -- (574,126.31) ;
\draw [shift={(576,126.3)}, rotate = 179.72] \arrow  ;
\draw   (596.8,178.13) .. controls (596.8,170.4) and (603.07,164.13) .. (610.8,164.13) .. controls (618.53,164.13) and (624.8,170.4) .. (624.8,178.13) .. controls (624.8,185.87) and (618.53,192.13) .. (610.8,192.13) .. controls (603.07,192.13) and (596.8,185.87) .. (596.8,178.13) -- cycle ;
\draw   (598.99,178.13) .. controls (598.99,171.61) and (604.28,166.32) .. (610.8,166.32) .. controls (617.32,166.32) and (622.61,171.61) .. (622.61,178.13) .. controls (622.61,184.66) and (617.32,189.95) .. (610.8,189.95) .. controls (604.28,189.95) and (598.99,184.66) .. (598.99,178.13) -- cycle ;
\draw   (576,126.3) .. controls (576,118.57) and (582.27,112.3) .. (590,112.3) .. controls (597.73,112.3) and (604,118.57) .. (604,126.3) .. controls (604,134.03) and (597.73,140.3) .. (590,140.3) .. controls (582.27,140.3) and (576,134.03) .. (576,126.3) -- cycle ;
\draw   (578.19,126.3) .. controls (578.19,119.78) and (583.48,114.49) .. (590,114.49) .. controls (596.52,114.49) and (601.81,119.78) .. (601.81,126.3) .. controls (601.81,132.82) and (596.52,138.11) .. (590,138.11) .. controls (583.48,138.11) and (578.19,132.82) .. (578.19,126.3) -- cycle ;
\draw   (506,125.63) .. controls (506,117.9) and (512.27,111.63) .. (520,111.63) .. controls (527.73,111.63) and (534,117.9) .. (534,125.63) .. controls (534,133.37) and (527.73,139.63) .. (520,139.63) .. controls (512.27,139.63) and (506,133.37) .. (506,125.63) -- cycle ;
\draw   (446.67,124.97) .. controls (446.67,117.24) and (452.93,110.97) .. (460.67,110.97) .. controls (468.4,110.97) and (474.67,117.24) .. (474.67,124.97) .. controls (474.67,132.7) and (468.4,138.97) .. (460.67,138.97) .. controls (452.93,138.97) and (446.67,132.7) .. (446.67,124.97) -- cycle ;
\draw   (476.33,77.3) .. controls (476.33,69.57) and (482.6,63.3) .. (490.33,63.3) .. controls (498.06,63.3) and (504.33,69.57) .. (504.33,77.3) .. controls (504.33,85.03) and (498.06,91.3) .. (490.33,91.3) .. controls (482.6,91.3) and (476.33,85.03) .. (476.33,77.3) -- cycle ;
\draw   (478.52,77.3) .. controls (478.52,70.78) and (483.81,65.49) .. (490.33,65.49) .. controls (496.86,65.49) and (502.14,70.78) .. (502.14,77.3) .. controls (502.14,83.82) and (496.86,89.11) .. (490.33,89.11) .. controls (483.81,89.11) and (478.52,83.82) .. (478.52,77.3) -- cycle ;
\draw   (427,178.3) .. controls (427,170.57) and (433.27,164.3) .. (441,164.3) .. controls (448.73,164.3) and (455,170.57) .. (455,178.3) .. controls (455,186.03) and (448.73,192.3) .. (441,192.3) .. controls (433.27,192.3) and (427,186.03) .. (427,178.3) -- cycle ;
\draw    (460.31,172.32) -- (460.64,140.97) ;
\draw [shift={(460.67,138.97)}, rotate = 90.61] \arrow  ;
\draw    (499.67,87.3) -- (499.33,117.66) ;
\draw [shift={(499.31,119.66)}, rotate = 270.63] \arrow  ;
\draw    (480.64,118.99) -- (480.98,89.63) ;
\draw [shift={(481,87.63)}, rotate = 90.65] \arrow  ;

\draw (583.9,123) node [anchor=north west][inner sep=0.75pt]  [font=\small] [align=left] {$\displaystyle r_{1}$};
\draw (548.43,114.55) node [anchor=north west][inner sep=0.75pt]   [align=left] {$\displaystyle a$};
\draw (592,150) node [anchor=north west][inner sep=0.75pt]   [align=left] {$\displaystyle r$};
\draw (515.6,165.71) node [anchor=north west][inner sep=0.75pt]   [align=left] {$\displaystyle R$};
\draw (483.1,113.21) node [anchor=north west][inner sep=0.75pt]   [align=left] {$\displaystyle T$};
\draw (604.57,175) node [anchor=north west][inner sep=0.75pt]  [font=\small] [align=left] {$\displaystyle s_{1}$};
\draw (513.9,122.33) node [anchor=north west][inner sep=0.75pt]  [font=\small] [align=left] {$\displaystyle r_{0}$};
\draw (454.57,121.67) node [anchor=north west][inner sep=0.75pt]  [font=\small] [align=left] {$\displaystyle r_{0}$};
\draw (484.23,73) node [anchor=north west][inner sep=0.75pt]  [font=\small] [align=left] {$\displaystyle t_{0}$};
\draw (434.57,175.67) node [anchor=north west][inner sep=0.75pt]  [font=\small] [align=left] {$\displaystyle s_{0}$};
\draw (449.27,150) node [anchor=north west][inner sep=0.75pt]   [align=left] {$\displaystyle c$};
\draw (501,100) node [anchor=north west][inner sep=0.75pt]   [align=left] {$\displaystyle r$};
\draw (469.6,100) node [anchor=north west][inner sep=0.75pt]   [align=left] {$\displaystyle c$};

\end{tikzpicture}
    \vspace{-2em}
    \caption{\RecursiveRun\ on $ccrar$ of the VRA $\mathcal{A}$.}
    \label{fig:ex-vra-run}
    \end{subfigure}
\caption{Example of a VRA and a \recursiveRun\ of it.}
\vspace{-1em}
\end{figure}
    
\begin{example}  \label{ex:vra}
    \autoref{fig:ex-vra} shows an example of a VRA $\mathcal{A}=\langle\pdwAlph,\procAlpha,\AutomataSet,\mathcal{A}^S\rangle$, with the pushdown alphabet $\pdwAlph = \{a\} \cup \{c\} \cup \{r\}$, the procedural alphabet $\procAlpha = \{R,T\}$, and the \linkingFunction\ $f$ such that $f(R)=f(T)=\langle c,r\rangle$. The VRA is composed of three DFAs $\mathcal{A}^S$, $\mathcal{A}^R$ and $\mathcal{A}^T$, where $\mathcal{A}^S$ is the starting one.

    Let $w=ccrar\in \wm\pdwAlph$. The following \recursiveRun\ on $w$ witnesses that $w\in \Lang{\mathcal{A}}$ (it is also illustrated in \autoref{fig:ex-vra-run} with the automata calls):
    \[\langle s_0,\varepsilon\rangle\xrightarrow{c}\langle r_0,s_1\rangle\xrightarrow{c}\langle t_0,r_0 s_1\rangle\xrightarrow{r}\langle r_0,s_1\rangle\xrightarrow{a}\langle r_1,s_1\rangle\xrightarrow{r}\langle s_1,\varepsilon\rangle\in\runsA{\mathcal{A}}.\]
    We explain the first three transitions of the \recursiveRun: 
    \begin{itemize}
        \item $\langle s_0,\varepsilon\rangle\xrightarrow{c}\langle r_0,s_1\rangle$ is possible since $s_0\xrightarrow{R}s_1\in \delta^S$ and  $f_{\call}(R)=c$. We call the FA $\mathcal{A}^R$, go to $r_0 \in I^R$, and push $s_1$ on top of the stack word $\varepsilon$.
        \item $\langle r_0,s_1\rangle\xrightarrow{c}\langle t_0,r_0  s_1\rangle$ is also possible, but with a call to the FA $\mathcal{A}^T$. 
        \item $\langle t_0,r_0 s_1\rangle\xrightarrow{r}\langle r_0,s_1\rangle$ is possible since $t_0\in F^T$ and $f_{\return}(T)=r$. The call to $\mathcal{A}^T$ is completed. We pop $r_0$ from the stack word and go to this state.
    \end{itemize}
    As this run starts in  $s_0 \in I^S$ and ends in $s_1 \in F^S$, it follows that $w \in \Lang{\mathcal{A}}$.
\end{example}

Given a VRA $\mathcal{A}$, each of its FAs $\mathcal{A}^J\in \AutomataSet$, can be seen as accepting either the recursive language $\Lang{\mathcal{A}^J} \subseteq \pdwAlph^*$, or the language $\regLang{\mathcal{A}^J} \subseteq (\procAlpha \cup \intAlpha)^*$. To avoid any confusion, 
a run of $\mathcal{A}^J$ on a word over $\intAlpha\cup\procAlpha$ is called a \emph{\regularRun}, and the language
$\regLang{\mathcal{A}^J}$ is called its \emph{\regularLanguage}.
Note that $\Lang{\mathcal{A}^J}\subseteq\wm\pdwAlph$. Indeed, a \recursiveRun\ on $w \in \Lang{\mathcal{A}}$ begins and ends with an empty stack word, and we cannot pop a symbol from an empty stack word.

In order to better understand the VRA model, we state \autoref{prop:vra-run-to-nfa-trans} below, which provides a recursive definition of the semantics of VRAs: 
to follow a \proceduralTransition\ $q\xrightarrow{J}p$, a VRA must read a word $cwr$ such that $f(J) = \langle c,r \rangle$ and $w$ is accepted by $\mathcal{A}^J$. \autoref{prop:vra-run-to-nfa-trans} is illustrated in \autoref{fig:sem-vra} (see also the example of \autoref{fig:ex-vra-run}) and proved in \autoref{ax:vra-prop}. Given $c\in \calAlpha$ and $r\in \retAlpha$, we write $\procAlpha^{\langle c,r\rangle}$ as the set of \proceduralSymbol s linked by $f$ to $\langle c,r\rangle$: $\procAlpha^{\langle c,r\rangle}= \{J\in \procAlpha \mid f(J)= \langle c,r\rangle\}$.

\begin{restatable}{proposition}{regRunRecRun}\label{prop:vra-run-to-nfa-trans}
    Given a VRA $\mathcal{A}$, let $cwr \in \calAlpha\cdot \wm\pdwAlph\cdot \retAlpha$ and $p,q\in \vraStates{A}$:
    \[\langle q,\varepsilon \rangle\xrightarrow{cwr}\langle p,\varepsilon \rangle\in \runsA{\mathcal{A}} \iff \exists J\in\procAlpha^{\langle c,r\rangle}:q\xrightarrow{J}p\in\vraTrans{A}\wedge w\in \Lang{\mathcal{A}^J}.\]
\end{restatable}

\begin{figure}[t]
    \centering
    \begin{tikzpicture}[x=0.75pt,y=0.75pt,yscale=-1,xscale=1]

\draw [dashed]   (277.7,94.38) -- (394.7,94.38) ;
\draw [shift={(396.7,94.38)}, rotate = 180] \arrow  ;
\draw    (285.6,90.07) -- (285.7,51.38) ;
\draw [shift={(285.7,49.38)}, rotate = 90.14] \arrow  ;
\draw    (390.7,49.38) -- (390.6,88.07) ;
\draw [shift={(390.6,90.07)}, rotate = 270.14] \arrow  ;
\draw    (299.7,35.38) .. controls (327.82,23.04) and (345.67,46.8) .. (374.9,36.08) ;
\draw [shift={(376.7,35.38)}, rotate = 157.72] \arrow  ;
\draw   (376.7,35.38) .. controls (376.7,27.65) and (382.97,21.38) .. (390.7,21.38) .. controls (398.43,21.38) and (404.7,27.65) .. (404.7,35.38) .. controls (404.7,43.11) and (398.43,49.38) .. (390.7,49.38) .. controls (382.97,49.38) and (376.7,43.11) .. (376.7,35.38) -- cycle ;
\draw   (271.7,35.38) .. controls (271.7,27.65) and (277.97,21.38) .. (285.7,21.38) .. controls (293.43,21.38) and (299.7,27.65) .. (299.7,35.38) .. controls (299.7,43.11) and (293.43,49.38) .. (285.7,49.38) .. controls (277.97,49.38) and (271.7,43.11) .. (271.7,35.38) -- cycle ;
\draw   (249.7,94.38) .. controls (249.7,86.65) and (255.97,80.38) .. (263.7,80.38) .. controls (271.43,80.38) and (277.7,86.65) .. (277.7,94.38) .. controls (277.7,102.11) and (271.43,108.38) .. (263.7,108.38) .. controls (255.97,108.38) and (249.7,102.11) .. (249.7,94.38) -- cycle ;
\draw   (396.7,94.38) .. controls (396.7,86.65) and (402.97,80.38) .. (410.7,80.38) .. controls (418.43,80.38) and (424.7,86.65) .. (424.7,94.38) .. controls (424.7,102.11) and (418.43,108.38) .. (410.7,108.38) .. controls (402.97,108.38) and (396.7,102.11) .. (396.7,94.38) -- cycle ;
\draw   (378.89,35.38) .. controls (378.89,28.85) and (384.18,23.57) .. (390.7,23.57) .. controls (397.22,23.57) and (402.51,28.85) .. (402.51,35.38) .. controls (402.51,41.9) and (397.22,47.19) .. (390.7,47.19) .. controls (384.18,47.19) and (378.89,41.9) .. (378.89,35.38) -- cycle ;

\draw (258.9,91.17) node [anchor=north west][inner sep=0.75pt]  [font=\small] [align=left] {$\displaystyle q$};
\draw (328.77,80.01) node [anchor=north west][inner sep=0.75pt]   [align=left] {$\displaystyle J$};
\draw (405.9,91.67) node [anchor=north west][inner sep=0.75pt]  [font=\small] [align=left] {$\displaystyle p$};
\draw (279.4,32.67) node [anchor=north west][inner sep=0.75pt]  [font=\small] [align=left] {$\displaystyle q_{i}$};
\draw (383.9,32.67) node [anchor=north west][inner sep=0.75pt]  [font=\small] [align=left] {$\displaystyle q_{f}$};
\draw (306.27,14.51) node [anchor=north west][inner sep=0.75pt]   [align=left] {$\displaystyle w\in L(\mathcal{A}^{J})$};
\draw (273.77,63.51) node [anchor=north west][inner sep=0.75pt]   [align=left] {$\displaystyle c$};
\draw (393.27,62.51) node [anchor=north west][inner sep=0.75pt]   [align=left] {$\displaystyle r$};

\end{tikzpicture}
    \vspace{-2em}
    \caption{Illustration of the semantics of a VRA on $cwr \in \calAlpha \cdot \wm{\pdwAlph} \cdot \retAlpha$, with $q_i \in I^J$, $q_f\in F^J$ and $f(J)=\langle c,r \rangle$.}
    \label{fig:sem-vra}
\vspace{-1em}
\end{figure}
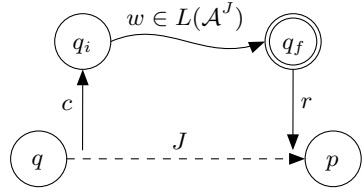

This proposition provides a link between the \recursiveLanguage\ of an automaton and its \regularLanguage. Consider an accepting \recursiveRun\ of $\mathcal{A}^J\in \AutomataSet\cup \{\mathcal{A}^S\}$ on a well-matched $w=u_0c_1w_1r_1\ldots c_nw_nr_nu_n\in \Lang{\mathcal{A}^J}$, with $n\in \mathbb{N}$, $u_i\in \intAlpha^*$, $c_i \in\calAlpha$, $r_i\in \retAlpha$ and $w_i\in \wm\pdwAlph$. We can decompose the \recursiveRun\ into  $\langle q,\varepsilon \rangle \xrightarrow{u_0} \langle q_1,\varepsilon \rangle \xrightarrow{c_1w_1r_1} \langle p_1,\varepsilon \rangle \xrightarrow{u_1} \dots \xrightarrow{u_{n-1}} \langle q_n,\varepsilon \rangle \xrightarrow{c_nw_nr_n} \langle p_n,\varepsilon \rangle \xrightarrow{u_n} \langle p,\varepsilon \rangle$, with $q\in I^J$, $p\in F^J$ and $q_i,p_i \in Q^J$ for all $i\in [1,n]$. By \autoref{prop:vra-run-to-nfa-trans}, we can replace each \recursiveRun\ $\langle q_i,\varepsilon\rangle \xrightarrow{c_iw_ir_i}\langle p_i, \varepsilon \rangle$ by a \regularRun\ $q_i\xrightarrow{J_i}p_i\in \vraTrans{A}$, with $J_i\in \procAlpha^{\langle c_i,r_i\rangle}$ such that $w_i\in \Lang{\mathcal{A}^{J_i}}$. This results in an accepting \regularRun\ on $u_0J_1\ldots J_nu_n\in \regLang{\mathcal{A}^J}$. Note that the converse also holds: from the word $u_0J_1\ldots J_nu_n$, we can replace each $J_i$ by a word $c_iw_i'r_i$, with $w_i'$ any word in $\Lang{\mathcal{A}^{J_i}}$, to obtain a word in the \recursiveLanguage\ of $\mathcal{A}^J$.

\subsection{Comparison with Other Models}

A VRA $\mathcal{A}=\langle\pdwAlph,\procAlpha, \AutomataSet, \mathcal{A}^S\rangle$ is \emph{deterministic} if all its automata in $\AutomataSet\cup \{\mathcal{A}^S\}$ are DFAs and, for all $q\in  \vraStates{A}$, if there exist two transitions $(q,J,p),(q,J',p') \in \vraTrans{A}$ with distinct $J,J'\in \procAlpha$, then $f_{\call}(J) \neq f_{\call}(J')$. The VRA of \autoref{fig:ex-vra} is not deterministic because  $r_0\xrightarrow{T}r_0,r_0\xrightarrow{R}r_1\in \delta^R$ and $f_{\call}(R)=f_{\call}(T)$. 

In a deterministic VRA, for all configurations $\langle q,\sigma \rangle\in \vraStates{A}\times \vraStates{A}^*$ and symbols $a\in \pdwAlph$, there exists at most one \recursiveRun\ on 
$a$ from $\langle q,\sigma \rangle$. This is clear when $a\in\intAlpha$, since all FAs are DFAs, and when $a\in\retAlpha$, by the semantics of VRAs. When $a\in \calAlpha$, since there exists at most one \proceduralSymbol\ $J\in\procAlpha$ such that $f_{\call}(J)=a$ and $(q,J,p)\in \vraTrans{A}$, the only reachable configuration is $\langle q_i,p \sigma\rangle$, with $q_i \in I^J$ the unique initial state of $\mathcal{A}^J$. 
The next proposition states that deterministic VRAs are less expressive than VRAs.
\begin{restatable}{proposition}{propDetVra}\label{prop:det-vra-are-not-general-vra}
    There exists no deterministic VRA accepting the \recursiveLanguage\ 
    accepted by the VRA depicted in \autoref{fig:ex-vra}.  
\end{restatable}
\begin{proof}[Sketch]
    The main idea is that a deterministic VRA cannot simulate the nondeterministic transitions $r_0\xrightarrow{T}r_0, r_0\xrightarrow{R}r_1 \in \delta^R$ without altering the \recursiveLanguage\ of $\mathcal{A}^R$. A detailed proof is given in \autoref{ax:det-vra-and-vra}.
\qed\end{proof}

A particular class of VRAs, called systems of procedural automata, is studied in \cite{frohme2021spa}. It consists of VRAs such that for all distinct $J,J'\in \procAlpha$, $f_{\call}(J)\neq f_{\call}(J')$. This class forms a strict subclass of the deterministic VRAs 
(see \autoref{ax:spa-vra}). 

\emph{Visibly pushdown automata} (VPAs) form a subclass of pushdown automata~\cite{alur2004vpl}.  The next theorem states that VRAs and VPAs are equivalent models. In \autoref{ax:vra-vpa}, we recall the formal definition of VPA and prove the theorem.

\begin{restatable}[Equivalence of VRAs and VPAs]{theorem}{vravpa}\label{th:vpa-vra}
    Let $L\subseteq \wm\pdwAlph$. There exists a VRA $\mathcal{A}$ accepting $L$ iff there exists a VPA $\mathcal{B}$ accepting $L$. 
    Moreover, there exists a logspace-computable construction for
    $\mathcal{B}$ with $|\mathcal{B}|=\mathcal{O}(|\mathcal{A}|)$ (resp.\ for $\mathcal{A}$  with $|\mathcal{A}|=\mathcal{O}(|\mathcal{B}|^4)$).
\end{restatable}

\section{\Codeterministic\ and \Complete\ VRAs}\label{sec:codeterministic-complete-vra}



\autoref{prop:det-vra-are-not-general-vra} states that not all VRAs have an equivalent deterministic VRA. We introduce in this section the notions of \emph{\codeterministic} VRA and \emph{\complete} VRA, and prove that any VRA can be transformed into a \codeterministic\ and \complete\ one. This property is notably useful to show that the class of VRAs is closed under complement (see \autoref{th:closure} below). 
Our concept of \codeterminism\  is inspired by the concept of \emph{\codeterministic\ grammars} introduced in~\cite{berstel2002balanced}.

\begin{definition}[\Codeterministic\ VRA]\label{def:codeterministic-vra}
    A VRA $\mathcal{A}$ is  \emph{\codeterministic} if all automata linked to the same call/return symbols have pairwise disjoint languages:
    \[\forall c\in\calAlpha, \forall r\in\retAlpha, \forall J,J'\in\procAlpha^{\langle c,r\rangle}:J\neq J' \Rightarrow \Lang{\mathcal{A}^J}\cap \Lang{\mathcal{A}^{J'}}=\varnothing.\]
\end{definition}


The notion of \complete\ VRA requires two conditions. 
The first asks all the FAs of the VRA to be complete. The second is a universality condition on the \recursiveLanguage s.
These conditions guarantee that there always exists a \recursiveRun\ on any well-matched word, regardless of the starting configuration.
  
\begin{definition}[\Complete\ VRA]\label{def:complete-vra}
   A VRA $\mathcal{A}$ is  \emph{\complete} if 
   \begin{itemize}
   \item for all $q\in \vraStates{A}$ and $a\in\intAlpha\cup\procAlpha$, there exists $(q,a,p)\in\vraTrans{A}$; \item for all $c\in\calAlpha$ and $r\in\retAlpha$: $\bigcup_{J\in\procAlpha^{\langle c,r\rangle}}\Lang{\mathcal{A}^J}=\wm\pdwAlph$.
   \end{itemize}
\end{definition}

\noindent
We remark that the second equation in the definition above holds for the partition of $\procAlpha$ into the $\procAlpha^{\langle c,r \rangle}$.
        
        
        

The conditions on \recursiveLanguage s for a VRA to be \codeterministic\ and \complete\ can be replaced by a condition at the level of \regularLanguage s:

\begin{restatable}{proposition}{codetCompRegLang}\label{prop:comp-codet-on-reg-lang}
     Let $\mathcal{A}$ be a VRA with all its automata being complete FAs. If, for all $ c\in\calAlpha$, $r\in\retAlpha$, the \regularLanguage s $\regLang{\mathcal{A}^J}$, with $J\in \procAlpha^{\langle c,r \rangle}$, form a partition of $(\intAlpha\cup \procAlpha)^*$, then $\mathcal{A}$ is \codeterministic\ and \complete. 
\end{restatable}
\begin{proof}[Sketch]
     By \autoref{prop:vra-run-to-nfa-trans}, since all regular languages $\regLang{\mathcal{A}^J}$,  $J\in \procAlpha^{\langle c,r \rangle}$, are pairwise disjoint, we can see that the \recursiveLanguage s are pairwise disjoint too, i.e., $\mathcal{A}$ is \codeterministic. Additionally, with \autoref{prop:vra-run-to-nfa-trans} again, as the union of all $\regLang{\mathcal{A}^J}$'s is equal to $(\intAlpha\cup \procAlpha)^*$, we can check that the union of the \recursiveLanguage s is $\wm\pdwAlph$. Since all FAs of $\mathcal{A}$ are complete by hypothesis, it follows that $\mathcal{A}$ is \complete. The formal proof is given in \autoref{ax:converse-codet-on-reg-lang}.
    \qed 
\end{proof}

We now show that given any VRA, we can construct an equivalent \codeterministic\ \complete\ VRA with an exponential size in the size of the input VRA.
\begin{restatable}[Power of Codeterministic Complete VRAs]{theorem}{codetcompvra}\label{th:codet-complete-vra}
    Given a VRA $\mathcal{A}$, one can construct an equivalent \codeterministic\ \complete\ VRA $\mathcal{B}$ such that $|\mathcal{B}| = 2^{\mathcal{O}(|\mathcal{A}|)}$.
    Moreover, the automata that compose the VRA $\mathcal{B}$ are all DFAs.
\end{restatable}


\begin{proof}[Sketch] The complete proof is given in \autoref{ax:detail-th-codet}. Given a VRA $\mathcal{A}=\langle\pdwAlph,\procAlpha,\AutomataSet,\mathcal{A}^S\rangle$, we want to construct an equivalent VRA $\mathcal{B}=\langle\pdwAlph,\procAlpha',\AutomataSet',\mathcal{B}^{S}\rangle$ that is \codeterministic\ and \complete. Merging \Cref{def:codeterministic-vra,def:complete-vra}, $\mathcal{B}$ must respect the following conditions: all its automata must be complete FAs and, for all $\langle c,r\rangle\in \calAlpha\times \retAlpha$, the \recursiveLanguage s of all $\mathcal{B^J}$, with $\mathcal{J}\in \procAlpha'^{\langle c,r\rangle}$, must be a partition of $\wm\pdwAlph$. The main idea is the following. We define $\procAlpha'^{\langle c,r\rangle} = 2^{\procAlpha^{\langle c,r\rangle}}$, leading to the procedural alphabet of $\mathcal{B}$ equal to $\procAlpha' = \bigcup_{\langle c,r\rangle \in \calAlpha \times \retAlpha} \procAlpha'^{\langle c,r\rangle}$. Then, for each $\langle c,r \rangle$, we want to obtain, for all $\mathcal{J}\in \procAlpha'^{\langle c,r\rangle}$:
\begin{equation}\label{eq:lang-B^J}
    \Lang{\mathcal{B}^\mathcal{J}} =  \bigcap_{J \in \mathcal{J}} \Lang{\mathcal{A}^J} ~\setminus \bigcup_{J\in \overline{\mathcal{J}}} \Lang{\mathcal{A}^{J}}. 
\end{equation} 
Recall (see \autoref{subsec:prelim}) that $\overline{\mathcal{J}}=\procAlpha^{\langle c,r\rangle} \setminus \mathcal{J}$ and, when $\mathcal{J} = \varnothing$, $\bigcap_{J \in \mathcal{J}} \Lang{\mathcal{A}^J}=\wm\pdwAlph$. In this way, the \recursiveLanguage s of all $\mathcal{B^J}$, $\mathcal{J}\in \procAlpha'^{\langle c,r\rangle}$, form a partition of $\wm\pdwAlph$ (see \autoref{fig:venn-codet-complete}).
Note that for each $\langle c,r\rangle$, the set $\varnothing$ belongs to $\procAlpha'^{\langle c,r\rangle}$, each time corresponding to a distinct automaton.

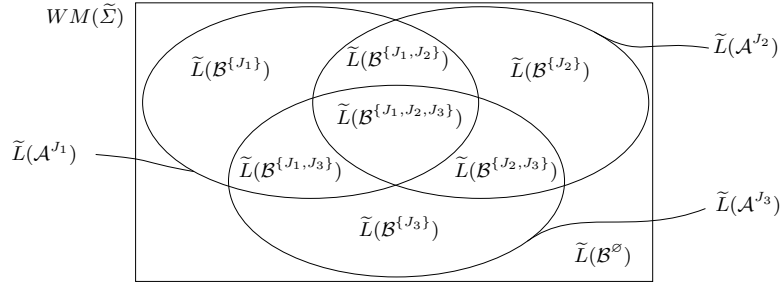
\begin{figure}[t]
    \centering
    \resizebox{.85\textwidth}{!}{
\begin{tikzpicture}[x=0.75pt,y=0.75pt,yscale=-1,xscale=1]

\draw    (189,6.13) -- (496,6.13) -- (496,172.8) -- (189,172.8) -- cycle ;
\draw    (294.16,65.8) .. controls (294.16,34.13) and (338.82,8.46) .. (393.91,8.46) .. controls (449,8.46) and (493.67,34.13) .. (493.67,65.8) .. controls (493.67,97.46) and (449,123.13) .. (393.91,123.13) .. controls (338.82,123.13) and (294.16,97.46) .. (294.16,65.8) -- cycle ;
\draw    (193.16,65.8) .. controls (193.16,34.13) and (237.82,8.46) .. (292.91,8.46) .. controls (348,8.46) and (392.67,34.13) .. (392.67,65.8) .. controls (392.67,97.46) and (348,123.13) .. (292.91,123.13) .. controls (237.82,123.13) and (193.16,97.46) .. (193.16,65.8) -- cycle ;
\draw    (244.16,112.8) .. controls (244.16,81.13) and (288.82,55.46) .. (343.91,55.46) .. controls (399,55.46) and (443.67,81.13) .. (443.67,112.8) .. controls (443.67,144.46) and (399,170.13) .. (343.91,170.13) .. controls (288.82,170.13) and (244.16,144.46) .. (244.16,112.8) -- cycle ;

\draw     (165.3,96.85) .. controls (203.6,100.25) and (188.4,100.85) .. (223.67,107.08) ;
\draw     (461.4,23.55) .. controls (490.33,40.21) and (491.33,27.55) .. (529.33,33.21) ;
\draw     (421.67,148.7) .. controls (450.67,127.21) and (484,146.21) .. (527.33,129.21) ;

\draw (409.99,37.45) node [anchor=north west][inner sep=0.75pt]  [font=\small] [align=left] {$\displaystyle \Lang{\mathcal{B}^{\{J_{2}\}}}$};
\draw (220.15,37.12) node [anchor=north west][inner sep=0.75pt]  [font=\small] [align=left] {$\displaystyle \Lang{\mathcal{B}^{\{J_{1}\}}}$};
\draw (320.31,132.51) node [anchor=north west][inner sep=0.75pt]  [font=\small] [align=left] {$\displaystyle \Lang{\mathcal{B}^{\{J_{3}\}}}$};
\draw (312.51,30.82) node [anchor=north west][inner sep=0.75pt]  [font=\small] [align=left] {$\displaystyle \Lang{\mathcal{B}^{\{J_{1} ,J_{2}\}}}$};
\draw (377,97) node [anchor=north west][inner sep=0.75pt]  [font=\small] [align=left] {$\displaystyle \Lang{\mathcal{B}^{\{J_{2} ,J_{3}\}}}$};
\draw (249,97) node [anchor=north west][inner sep=0.75pt]  [font=\small] [align=left] {$\displaystyle \Lang{\mathcal{B}^{\{J_{1} ,J_{3}\}}}$};
\draw (306.98,65.31) node [anchor=north west][inner sep=0.75pt]  [font=\small] [align=left] {$\displaystyle \Lang{\mathcal{B}^{\{J_{1} ,J_{2} ,J_{3}\}}}$};
\draw (448.19,148) node [anchor=north west][inner sep=0.75pt]  [font=\small] [align=left] {$\displaystyle \Lang{\mathcal{B}^{\varnothing}}$};
\draw (531.98,117.67) node [anchor=north west][inner sep=0.75pt]  [font=\small] [align=left] {$\displaystyle \Lang{\mathcal{A}^{J_{3}}}$};
\draw (134.72,5) node [anchor=north west][inner sep=0.75pt]  [font=\small] [align=left] {$\displaystyle WM( \pdwAlph )$};
\draw (113.98,87.4) node [anchor=north west][inner sep=0.75pt]  [font=\small] [align=left] {$\displaystyle \Lang{\mathcal{A}^{J_{1}}}$};
\draw (530.64,23) node [anchor=north west][inner sep=0.75pt]  [font=\small] [align=left] {$\displaystyle \Lang{\mathcal{A}^{J_{2}}}$};
\end{tikzpicture}
}
    \vspace{-2em}
    \caption{Set of \recursiveLanguage s of three automata $\mathcal{A}^{J_1}$, $\mathcal{A}^{J_2}$ and $\mathcal{A}^{J_3}$, and the corresponding set of languages of the automata $\mathcal{B^J}$, for all $\mathcal{J}\subseteq \{J_1,J_2,J_3\}$.}
    \label{fig:venn-codet-complete}
\vspace{-1em}
\end{figure}
    
    Before detailing the construction of each $\mathcal{B}^{\mathcal{J}}$, we transform each $\mathcal{A}^J\in \AutomataSet\cup \{\mathcal{A}^S\}$ into a complete DFA $\mathcal{A}'^J$ over the alphabet $\intAlpha\cup \procAlpha'$ as follows. 
    For all $J\in \procAlpha$, we replace each \proceduralTransition\ $q\xrightarrow{J}p\in\delta_\mathcal{A}$ by the transitions $q\xrightarrow{\mathcal{J}} p$, for all $\mathcal{J} \in \procAlpha'$ such that $\mathcal{J} \ni J$, and we then apply the subset construction to get a complete DFA~\cite{HU79}. This first step later helps obtain~\eqref{eq:lang-B^J}.
    With this construction, the \regularLanguage\ of $\mathcal{A}'^J$ is equal to the one of $\mathcal{A}^J$, up to the replacement of the \proceduralSymbol s appearing in the accepted words:

    \begin{restatable}{property}{langAprime}\label{pro:lang-A'^K}
        Let $J\in \procAlpha$, $n\in \mathbb{N}$, $u_i\in \intAlpha^*$ and $\mathcal{J}_i\in \procAlpha'$ for all $i$:
        \[ u_0\mathcal{J}_1\dots \mathcal{J}_nu_n \in \regLang{\mathcal{A}'^J} \iff \forall i\in [1,n], \exists J_i\in \mathcal{J}_i : u_0J_1\dots J_nu_n \in \regLang{\mathcal{A}^J}.\]
    \end{restatable}
    \begin{proof}[of \autoref{pro:lang-A'^K} - Sketch] 
        This  follows from the replacement of the \proceduralTransition s $q\xrightarrow{J}p$ by $q\xrightarrow{\mathcal{J}} p$, for all $\mathcal{J} \in \procAlpha'$ such that $J\in \mathcal{J}$. 
        More details are given in \autoref{ax:detail-th-codet}. 
          \hfill   $\lrcorner$ 
    \end{proof}

    Since all $\mathcal{A}'^J$, $J \in \procAlpha \cup \{S\}$,  are complete DFAs, they are closed under Boolean operations with well-known constructions~\cite{HU79,sipser1996introduction}. 
    For each $\langle c,r\rangle$, we can thus construct an automaton $\mathcal{B^J}$, $\mathcal{J} \in \procAlpha'^{\langle c,r\rangle}$,  such that its \regularLanguage\ respects a form similar to \eqref{eq:lang-B^J}:
    \begin{equation}\label{eq:reg-lang-B^J}
        \regLang{\mathcal{B^J}}=\bigcap_{J\in \mathcal{J}} \regLang{\mathcal{A}'^J} \setminus \bigcup_{J \in \overline{\mathcal{J}}} \regLang{\mathcal{A}'^{J}}=\bigcap_{J\in \mathcal{J}} \regLang{\mathcal{A}'^J} \cap \bigcap_{J \in \overline{\mathcal{J}}} \overline{\regLang{\mathcal{A}'^{J}}}.
    \end{equation}
    The formal construction of $\mathcal{B}^{\mathcal{J}}$ is described in \autoref{ax:detail-th-codet}, with an appropriate Cartesian product of the FAs $\mathcal{A}'^J$, with $J\in \procAlpha^{\langle c,r\rangle}$, in a way that $\mathcal{B^J}$ accepts the \regularLanguage\ of \eqref{eq:reg-lang-B^J}.

    Finally, 
    we construct the required \codeterministic\ and \complete\ VRA $\mathcal{B}=\langle\pdwAlph,\procAlpha',\AutomataSet',\mathcal{B}^S\rangle$ such that 
    $\AutomataSet'=\{\mathcal{B}^{\mathcal{J}}\mid \mathcal{J}\in \procAlpha'\}$ where each $\mathcal{B}^{\mathcal{J}}$ is obtained as described before, and $\mathcal{B}^S=\mathcal{A}'^S$.

    Let us prove that $\mathcal{B}$ is \codeterministic\ and \complete. For all $\langle c,r\rangle\in \calAlpha\times \retAlpha$, according to \eqref{eq:reg-lang-B^J}, the \regularLanguage s of all  DFAs $\mathcal{B^J}$, with $\mathcal{J}\in \procAlpha'^{\langle c,r\rangle}$, form a partition of $(\intAlpha\cup \procAlpha')^*$. By \autoref{prop:comp-codet-on-reg-lang}, since all automata are complete DFAs, it follows that $\mathcal{B}$ is \codeterministic\ and \complete. 

    We now prove that $\mathcal{B}$ accepts the same language as $\mathcal{A}$. We first prove the correctness of the \recursiveLanguage s $\Lang{\mathcal{B^J}}$ as exposed in \eqref{eq:lang-B^J}, which is a consequence of $\mathcal{B}$ being \codeterministic\ and \complete\ and the next property (see also \autoref{fig:venn-codet-complete}).
    
    \begin{restatable}{property}{langBJ}\label{pro:lang-B^J}
            For all $J \in \procAlpha$, $\Lang{\mathcal{A}^J} = \bigcup_{\mathcal{J} \in \procAlpha', \mathcal{J} \ni J} \Lang{\mathcal{B}^\mathcal{J}}$.
    \end{restatable}
    \begin{proof}[of \autoref{pro:lang-B^J} - Sketch] 
        This  follows from $\eqref{eq:reg-lang-B^J}$, \autoref{prop:vra-run-to-nfa-trans} and \autoref{pro:lang-A'^K}. More details of the proof are given in \autoref{ax:detail-th-codet}. \hfill   $\lrcorner$ 
    \end{proof}

    Finally, to show that $\mathcal{A}$ and $\mathcal{B}$ are equivalent,   we must prove that for all $w \in \wm\pdwAlph$, $w \in \Lang{\mathcal{A}^{S}} \Leftrightarrow w \in \Lang{\mathcal{B}^S}$.
    Suppose that $w = u_0c_1w_1r_1 \ldots c_nw_nr_nu_n \in \wm\pdwAlph$ with $n\in \mathbb{N}$, $u_i\in \intAlpha^*$, $c_i\in \calAlpha$, $r_i\in \retAlpha$, $w_i \in \wm\pdwAlph$ for all~$i$: 
    \begin{itemize}
        \item[$\Rightarrow$]
        If $w \in \Lang{\mathcal{A}^S}$, by \autoref{prop:vra-run-to-nfa-trans}, we have $u_0{J}_1\dots {J}_n u_n\in\regLang{\mathcal{A}^S}$ for some ${J}_i\in \procAlpha^{\langle c_i,r_i\rangle}$ such that $w_i\in \Lang{\mathcal{A}^{{J}_i}}$, for all $i\in[1,n]$. By \autoref{pro:lang-B^J}, for all $i$, as $w_i \in \Lang{\mathcal{A}^{J_i}}$, there exists $\mathcal{J}_i \ni J_i$ such that $w_i \in \Lang{\mathcal{B}^{\mathcal{J}_i}}$. Then by \autoref{pro:lang-A'^K}, we have that $u_0\mathcal{J}_1 \ldots \mathcal{J}_nu_n \in \regLang{\mathcal{A}'^{S}}$. By \autoref{prop:vra-run-to-nfa-trans}, it follows that $w\in \Lang{\mathcal{A}'^{S}}= \Lang{\mathcal{B}^S}$.
        
        \item [$\Leftarrow$]
        The other implication is proved similarly.
    \end{itemize}

\medskip
To complete the proof, it remains to study the size of $\mathcal{B}$. The {number of states}  of $\mathcal{B}^S$ is equal to $2^{|Q^S|}$. 
For each $\langle c,r\rangle\in \calAlpha\times \retAlpha$, there are $2^{|\procAlpha^{\langle c,r\rangle}|}$ automata $\mathcal{B}^\mathcal{J}$, each with a {number of states}  $\prod_{J\in \procAlpha ^{\langle c,r\rangle}} 2 ^{|Q^J|} = 2^{\sum|Q^J|}$. Hence,
\begin{center}
    $|\vraStates{B}| = 2 ^{|Q^S|} + \sum_{\langle c,r\rangle \in \calAlpha \times \retAlpha}2^{|\procAlpha^{\langle c,r\rangle}|}\cdot 2^{\sum_{J \in \procAlpha^{\langle c,r\rangle}}|Q^J|} = 2^{\mathcal{O}({|\vraStates{A}|})}$.
\end{center} 
Since the number of transitions $|\vraTrans{B}|$ of $\mathcal{B}$ is in $\mathcal{O}(|\vraStates{B}|^2\cdot |\procAlpha'|)= 2^{\mathcal{O}({|\mathcal{A}|})}$ ($|\intAlpha|$ is supposed constant and $|\procAlpha'|\leq |\vraStates{B}|$), we conclude that $|\mathcal{B}|= 2^{\mathcal{O}(|\mathcal{A}|)}$.
\qed 



\end{proof}

It is natural to study the  problem of whether a given VRA is \codeterministic\ or \complete. We show in \autoref{ax:codet-comp-dp} that the first problem is PTIME-complete while the second is EXPTIME-complete. 

\section{Closure Properties and Decision Problems}
\label{sec:ClosureProp}

\subsection{Closure Properties of VRAs}

As VRAs and VPAs are interreducible (\autoref{th:vpa-vra}), VRAs inherit the same closure properties as VPAs, which are closed under concatenation, Kleene-$*$, and Boolean operations~\cite{alur2004vpl}. 
Although correct, translating VRAs into VPAs in order to perform these closure operations 
leads to a polynomially larger automaton.
We therefore provide direct constructions over VRAs, yielding automata
with sizes as exposed in \autoref{tab:vra-vpa}.

\begin{restatable}[Closure properties]{theorem}{closure}\label{th:closure}
Let $\mathcal{A}_1$ and $\mathcal{A}_2$ be VRAs with $\Lang{\mathcal{A}_1} = L_1$ and $\Lang{\mathcal{A}_2} = L_2$. 
One can construct a VRA accepting 
$L_1 \cdot L_2$, $L_1^*$, $L_1 \cup L_2$, $L_1 \cap L_2$, and $\overline{L_1}$ 
with respective sizes in 
$\mathcal{O}(|\mathcal{A}_1|+|\mathcal{A}_2|)$,
$\mathcal{O}(|\mathcal{A}_1|)$,
$\mathcal{O}(|\mathcal{A}_1|+|\mathcal{A}_2|)$,
$\mathcal{O}(|\mathcal{A}_1|\cdot |\mathcal{A}_2|)$, and
$2^{\mathcal{O}(|\mathcal{A}_1|)}$.
\end{restatable}

\begin{proof}[Sketch] 
    We only provide intuition of the constructions. Formal explanations for each operation are given in \autoref{ax:closure-properties}. Let $\mathcal{A}_1=\langle\pdwAlph,\Sigma_{\proc 1},\AutomataSet_1,\mathcal{A}_1^S\rangle$ and $\mathcal{A}_2=\langle\pdwAlph,\Sigma_{\proc 2},\AutomataSet_2,\mathcal{A}_2^S\rangle$. We assume, without loss of generality, that their input alphabets are the same, and that their \proceduralAlphabet s are disjoint.

    In case of concatenation, Kleene-$*$ and union operations, the constructions are easy, as they only involve the starting automaton. For instance, to obtain a VRA $\mathcal{B}$ accepting $\Lang{\mathcal{B}}= L_1 \cdot L_2$, we simply copy all automata of $\AutomataSet_1$ and $\AutomataSet_2$ and define the starting automaton $\mathcal{B}^S$ such that it accepts the \regularLanguage\ $\regLang{\mathcal{B}^S}= \regLang{\mathcal{A}_1^S} \cdot \regLang{\mathcal{A}_2^S}$, using the concatenation construction for FAs \cite{HU79}. 

    For the intersection, 
    we need to compute the intersection of each pair of \recursiveLanguage s $\Lang{\mathcal{A}^{J_1}_1}$ and $\Lang{\mathcal{A}^{J_2}_2}$ of $\mathcal{A}_1$ and $\mathcal{A}_2$.   
    Intuitively, we define the new \proceduralSymbol s $\langle J_1,J_2 \rangle\in \Sigma_{\proc 1}\times \Sigma_{\proc 2}$ and we replace all transitions $q_1\xrightarrow{J_1}p_1 \in \vraTrans{A_\text{1}}$ and $q_2\xrightarrow{J_2}p_2 \in \vraTrans{A_\text{2}}$, respectively by $q_1\xrightarrow{\langle J_1,J_2 \rangle}p_1$ and $q_2\xrightarrow{\langle J_1,J_2 \rangle}p_2$. We then construct $\mathcal{B}^{\langle J_1,J_2 \rangle}$ equal to the Cartesian product of $\mathcal{A}^{J_1}_1$ and $\mathcal{A}^{J_2}_2$, such that $\Lang{\mathcal{B}^{\langle J_1,J_2 \rangle}}=\Lang{\mathcal{A}^{J_1}_1}\cap \Lang{\mathcal{A}^{J_2}_2}$. As we must construct a Cartesian product for all pairs of procedural symbols, it follows that $|\mathcal{B}|=\mathcal{O}(|\mathcal{A}_1|\cdot |\mathcal{A}_2|)$. 
    
    Lastly, for the complementation, if $\mathcal{A}_1$ is \codeterministic\ and \complete, with all its automata being complete DFAs, then we construct $\mathcal{B}=\langle\pdwAlph,  \Sigma_{\proc 1}, \AutomataSet_1, \mathcal{B}^S\rangle$, with $\mathcal{B}^{S}$ accepting the \regularLanguage\ $\regLang{\mathcal{B}^S}=\overline{\regLang{\mathcal{A}_1^S}}$ (i.e., final states of $\mathcal{B}^S$ are the non final states of $\mathcal{A}_1^S$ \cite{HU79}). If not, we first apply the construction of \autoref{th:codet-complete-vra}, and then apply the previous construction. Thus, $|\mathcal{B}|=2^{\mathcal{O}(|\mathcal{A}_1|)}$.    
    \qed
    
\end{proof}

\subsection{Decision Problems for VRAs}

We here study the complexity of the emptiness, universality, inclusion, and equivalence decision problems. 
They belong to the same complexity class as for VPAs~\cite{alur2004vpl,lange2011p}, since \autoref{th:vpa-vra} states that VRAs and VPAs are equivalent under a logspace reduction.
However, using direct algorithms without translations into equivalent VPAs yields lower upper-bounds, as summarized in \autoref{tab:vra-vpa}.

\begin{restatable}[Decision problems for VRAs]{theorem}{emptinessdp}\label{th:decision-problem}
    Let $\mathcal{A}_1, \mathcal{A}_2$ be  two VRAs. The emptiness decision problem is PTIME-complete, with an upper-bound time complexity in $\mathcal{O}(|\mathcal{A}_1|)$. The universality, inclusion, and equivalence decision problems are EXPTIME-complete, with an upper-bound time complexity respectively in $2^{\mathcal{O}(|\mathcal{A}_1|)}$,  $\mathcal{O}(|\mathcal{A}_1|) \cdot 2^{\mathcal{O}(|\mathcal{A}_2|)}$, and $2^{\mathcal{O}(|\mathcal{A}_1| + |\mathcal{A}_2|)}$.
\end{restatable}

\begin{proof}[Sketch]
    The complexity class was discussed above. We here propose an algorithm that solves the emptiness decision problem for VRAs. The other decision problems are solved with classical methods (see \autoref{ax:emptiness-dp}).
    
    The main idea to solve the emptiness problem is to progressively compute the set of automata in $\AutomataSet \cup \{\mathcal{A}^S\}$ whose languages are not empty, using a reachability algorithm. Initially, the  algorithm starts from all initial states and is limited to the internal transitions. When a final state of an automaton $\mathcal{A}^J \in \AutomataSet$ is reached, this automaton is marked as having a nonempty \recursiveLanguage, and the algorithm is then allowed to take \proceduralTransition s reading the symbol $J$.

    The algorithm updates a set $\mathit{Reach}_i \subseteq \vraStates{A}$, for $i\in \mathbb{N}$, which contains the states marked as reachable from an initial state. It also uses a set $\mathcal{J}_i \subseteq \procAlpha$ which contains symbols $J\in \procAlpha$ such that $\mathit{Reach}_i$ contains a final state of $\mathcal{A}^J$:
    \begin{itemize}
        \item \textbf{Initialization:} $\mathit{Reach}_0 =  \bigcup_{J\in\procAlpha\cup\{S\}} I^J$;
        \item \textbf{Main loop:}  Let $\mathcal{J}_i = \{J \in \procAlpha \mid F^J \cap \mathit{Reach}_i \neq \varnothing \}$, then $\mathit{Reach}_{i+1} = \mathit{Reach}_i \cup \{p \in \vraStates{A}\mid \exists q\in \mathit{Reach}_i,a\in \intAlpha \cup \mathcal{J}_i : q\xrightarrow{a}p\in \vraTrans{A} \}$; 
        \item \textbf{Output:} When $\textit{Reach}_{i+1}=\textit{Reach}_i$, $\Lang{\mathcal{A}} = \varnothing$ iff $F^S\cap \textit{Reach}_i = \varnothing$. 
    \end{itemize}

    \noindent The correctness and complexity of the algorithm are detailed in \autoref{ax:emptiness-dp}.
    \qed
\end{proof}

\section{Conclusion}

We studied an extension of a class of modular automata proposed in~\cite{frohme2021spa}: visibly recursive automata, in which the modules are FAs. We showed that they are equivalent to VPAs, and provided complexity results about the classical language operations and automata decision problems. In line with this paper, we intend to
study the problem of deciding 
the determinization problem for VRAs.

Our future main goal is the design of a learning algorithm for VRAs (in the Angluin's sense~\cite{Angluin87}), 
to obtain readable 
automata instead of VPAs. This paper is the first step in that direction. The second step is to study the existence of a canonical VRA 
as Angluin's algorithm requires such a canonical model. We believe that \codeterministic\ and \complete\ VRAs are a promising avenue for this task. 
Once the VRA learning algorithm is designed, we intend to continue our search for an effective and scalable inference algorithm for JSON schemas to enable efficient
validation of JSON documents as initiated in \cite{Dubrulle2025}.

\bibliographystyle{splncs04}
\bibliography{bibliography}

\appendix

\clearpage 
\section{Basic Properties of VRAs}\phantomsection\label{ax:vra-prop}

In this appendix, we first highlight a basic property of VRAs that can help to better understand the behavior of a stack word when reading a well-matched word. We then provide a  proof of \autoref{prop:vra-run-to-nfa-trans}.

\begin{lemma} \label{lem:stack-equivalence-vra}
    Let $\mathcal{A}= \langle\pdwAlph,\procAlpha,\AutomataSet,\mathcal{A}^S\rangle$ be a VRA, $w\in \wm{\pdwAlph}$ and $q_1,q_2 \in \vraStates{A}$. 
    If $\langle q_1,\sigma_1 \rangle \xrightarrow{w}\langle q_2,\sigma_2 \rangle$ is a recursive run of $\mathcal{A}$, then $\sigma_2 = \sigma_1$ and $q_1,q_2 \in Q^J$ for some $J \in \procAlpha \cup \{S\}$.
    Moreover, $\langle q_1,\sigma \rangle \xrightarrow{w}\langle q_2,\sigma \rangle$ is also a recursive run of $\mathcal{A}$, for any $\sigma \in \vraStates{A}^*$.
\end{lemma}

\begin{proof}
    This  proof is by structural induction of well-matched words (see \autoref{def:wm-words}).
    \begin{itemize}
        \item $w\in \intAlpha^*$: This is trivial since a VRA behaves like an FA on internal symbols. 
        \item $w=cw'r$ ($c\in\calAlpha$, $r\in\retAlpha$, $w'\in\wm{\pdwAlph}$): The \recursiveRun\ on $w$ can be divided into:
        \[\langle q_1,\sigma_1 \rangle\xrightarrow{c}\langle q_c,\sigma_c \rangle\xrightarrow{w'}\langle q_r,\sigma_r \rangle\xrightarrow{r}\langle q_2,\sigma_2 \rangle.\]
        By structural induction, we have that $\sigma_c = \sigma_r$. Moreover, by definition of call and return transitions: 
        \begin{itemize}
            \item For $\langle q_r,\sigma_r \rangle\xrightarrow{r}\langle q_2,\sigma_2 \rangle$ to exist, $\sigma_r$ must satisfy $\sigma_r = q_2\cdot \sigma_2$;
            \item For $\langle q_1,\sigma_1 \rangle \xrightarrow{c} \langle q_c,q_2\cdot\sigma_2\rangle $ to exist, the symbol stacked by the call transition is $q_2$, and there must exist a \proceduralTransition\ $q_1\xrightarrow{K}q_2\in \delta^J\subseteq \vraTrans{A}$, for some $J\in\procAlpha\cup \{S\}$. Therefore, $\sigma_1 = \sigma_2$ and $q_1, q_2 \in Q^J$. 
        \end{itemize}
        \par We also get the existence of the \recursiveRun\ $\langle q_1,\sigma\rangle \xrightarrow{cw'r}\langle q_2,\sigma \rangle$, whatever $\sigma \in \vraStates{A}^*$, because what we have explained is independent of $\sigma_1$ and $\sigma_2$. 
        \item $w=w_1\cdot w_2$ ($w_1,w_2\in\wm\pdwAlph$): This is trivial by structural induction. \qed
    \end{itemize}
\end{proof}

We now prove \autoref{prop:vra-run-to-nfa-trans}.
\regRunRecRun*
\begin{proof}
    We prove both implications separately.
    \begin{itemize}
        \item[$\Rightarrow$] Thanks to \autoref{lem:stack-equivalence-vra}, the \recursiveRun\ on $cwr$ can be decomposed as: 
        \[ \langle q,\varepsilon \rangle \xrightarrow{c} \langle q_1,\sigma \rangle \xrightarrow{w} \langle q_2,\sigma \rangle \xrightarrow{r} \langle p,\varepsilon \rangle.\]
        with $q_1,q_2\in Q^J$ for some $J\in\procAlpha$, and $\sigma \in \vraStates{A}^*$ a stack word. Moreover, by semantics of VRAs, 
        we must have $\sigma=p$, $q_1 \in I^J$, $q_2 \in F^J$, $q\xrightarrow{J}p \in \vraTrans{A}$, and $f(J)=\langle c,r\rangle$. Finally, since $q_1 \in I^J$, $q_2 \in F^J$, and $\langle q_1,\varepsilon \rangle\xrightarrow{w} \langle q_2,\varepsilon \rangle \in \runsA{\mathcal{A}}$ by \autoref{lem:stack-equivalence-vra}, we have that $w\in \Lang{\mathcal{A}^J}$. 
        
        \item[$\Leftarrow$] For all $w\in \Lang{\mathcal{A}^J}$, there exists $\langle q_i,\varepsilon \rangle \xrightarrow{w} \langle q_f,\varepsilon \rangle\in \runsA{\mathcal{A}}$ with $q_i \in I^J$ and $q_f\in F^J$. By \autoref{lem:stack-equivalence-vra}, $\langle q_i,\sigma \rangle \xrightarrow{w} \langle q_f,\sigma \rangle\in \runsA{\mathcal{A}}$  for all  $\sigma\in \vraStates{A}^*$. By hypothesis, $q\xrightarrow{J}p\in\delta$ and $f(J)=\langle c,r\rangle$, thus, following the semantics of VRAs, $\langle q,\varepsilon \rangle \xrightarrow{c} \langle q_i,p\rangle \in \runsA{\mathcal{A}}$, and $\langle q_f,p\rangle \xrightarrow{r} \langle p,\varepsilon \rangle\in \runsA{\mathcal{A}}$ (since $q_i\in I^J$ and $q_f\in F^J$). Combining the \recursiveRun s leads to $\langle q,\varepsilon \rangle \xrightarrow{c} \langle q_i,p\rangle \xrightarrow{w} \langle q_f,p\rangle\xrightarrow{r} \langle p,\varepsilon \rangle\in \runsA{\mathcal{A}}$. \qed
    \end{itemize}
\end{proof}

\section{Proof of \autoref{prop:det-vra-are-not-general-vra}}\phantomsection \label{ax:det-vra-and-vra}
    \propDetVra* 

    \begin{proof}
        By contradiction, suppose that there exists a deterministic VRA $\mathcal{B}=\langle\pdwAlph,\procAlpha',\AutomataSet',\mathcal{B}^{S}\rangle$, with $\pdwAlph = \{a\} \cup \{c\} \cup \{r\}$, accepting $\Lang{\mathcal{B}}=\Lang{\mathcal{A}}$. Let $\mathcal{B}^S=\langle\intAlpha\cup\procAlpha', Q^S,I^S,F^S,\delta^S\rangle$. 
        
        It is easy to see that $ccarr,ccrar\in \Lang{\mathcal{A}}$. The idea is to show that, if $\mathcal{B}$ accepts both these words, then it must accept a word not in $\Lang{\mathcal{A}}$.
        In the VRA $\mathcal{B}$, according to the semantics of VRAs, the accepting \recursiveRun s of $ccarr$ and $ccrar$ should be: 
        \[
        \begin{aligned}
            &\langle q_0,\varepsilon \rangle \xrightarrow{c} 
        \langle q_1,q_5\rangle\xrightarrow{c} 
        \langle q_2,q_4\cdot q_5\rangle\xrightarrow{a}
        \langle q_3,q_4\cdot q_5\rangle \xrightarrow{r}
        \langle q_4,q_5\rangle\xrightarrow{r}
        \langle q_5,\varepsilon\rangle\in \runsA{\mathcal{B}}, \text{ and} \\
        &\langle q_0',\varepsilon \rangle \xrightarrow{c} 
        \langle q_1',q_5'\rangle\xrightarrow{c} 
        \langle q_2',q_3'\cdot q_5'\rangle\xrightarrow{r}
        \langle q_3',q_5'\rangle \xrightarrow{a}
        \langle q_4',q_5'\rangle\xrightarrow{r}
        \langle q_5',\varepsilon\rangle\in \runsA{\mathcal{B}}, 
        \end{aligned}
        \]
        for some $q_i,q_i' \in \vraStates{B}$,  $i\in [0,5]$, such that $q_0,q_0' \in I^S$ and $q_5,q_5' \in F^S$.

        As $\mathcal{B}$ is deterministic by hypothesis, $\mathcal{B}^S$ is a DFA and has a unique initial state, i.e., $q_0=q_0'$. Additionally, from the configuration $\langle q_0 ,\varepsilon \rangle$, the \recursiveRun\ on $cc$ must be unique, therefore $\langle q_2,q_4\cdot q_5\rangle = \langle q_2',q_3'\cdot q_5'\rangle$. This implies $q_4=q_3'$ and $q_5=q_5'$. Thus, we can deduce the existence of the following  \recursiveRun: 
        \[\langle q_0,\varepsilon \rangle \xrightarrow{cc} 
        \langle q_2,q_4\cdot q_5\rangle\xrightarrow{ar}
        \langle q_4,q_5\rangle \hspace{-.2em}= \hspace{-.2em}\langle q'_3,q'_5\rangle\xrightarrow{ar} 
        \langle q_5',\varepsilon\rangle \in 
        \runsA{\mathcal{B}}.\]
        Since this \recursiveRun\ is accepting, it follows that $ccarar \in \Lang{\mathcal{B}}$, which is a contradiction since $ccarar \notin \Lang{\mathcal{A}}$. \qed

    \end{proof}

    

    

\section{Systems of Procedural Automata}\phantomsection \label{ax:spa-vra}
Frohme and Steffen introduced in \cite{frohme2021spa} the model of \emph{system of procedural automata} (SPA). Up to a renaming of symbols, an SPA is equivalent to a VRA $\mathcal{A}=\langle \pdwAlph, \procAlpha, \AutomataSet, \mathcal{A}^S\rangle$ where all automata are DFAs, and each call symbol corresponds to a unique DFA, i.e., for all distinct $J,J'\in \procAlpha$, $f_{\call} (J) \neq f_{\call}(J')$. 
We show that the class of SPAs forms a strict subclass of the deterministic VRAs.
\begin{proposition}
Each SPA is a deterministic VRA, but there exist deterministic VRAs with no equivalent SPA. 
\end{proposition}

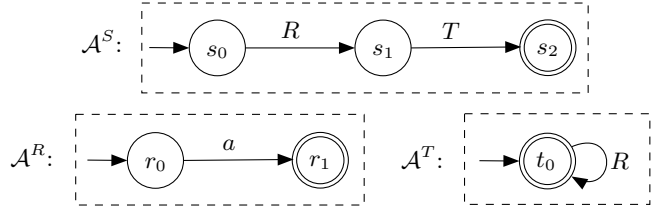
\begin{figure}[t]
    \centering
    \begin{tikzpicture}[x=0.75pt,y=0.75pt,yscale=-1,xscale=1]

\draw    (166.93,47.55) -- (184.9,47.72) ;
\draw [shift={(186.9,47.73)}, rotate = 180.52] \arrow  ;
\draw [dashed]  (163.2,25.15) -- (386,25.15) -- (386,70.35) -- (163.2,70.35) -- cycle ;
\draw  [dashed] (130.2,81.15) -- (274.2,81.15) -- (274.2,126.75) -- (130.2,126.75) -- cycle ;
\draw  (297.61,47.73) -- (350.61,47.25) ;
\draw [shift={(352.61,47.23)}, rotate = 179.48] \arrow  ;
\draw   (186.9,47.73) .. controls (186.9,40) and (193.17,33.73) .. (200.9,33.73) .. controls (208.63,33.73) and (214.9,40) .. (214.9,47.73) .. controls (214.9,55.47) and (208.63,61.73) .. (200.9,61.73) .. controls (193.17,61.73) and (186.9,55.47) .. (186.9,47.73) -- cycle ;
\draw   (351.9,47.63) .. controls (351.9,39.9) and (358.17,33.63) .. (365.9,33.63) .. controls (373.63,33.63) and (379.9,39.9) .. (379.9,47.63) .. controls (379.9,55.37) and (373.63,61.63) .. (365.9,61.63) .. controls (358.17,61.63) and (351.9,55.37) .. (351.9,47.63) -- cycle ;
\draw   (354.09,47.63) .. controls (354.09,41.11) and (359.38,35.82) .. (365.9,35.82) .. controls (372.42,35.82) and (377.71,41.11) .. (377.71,47.63) .. controls (377.71,54.16) and (372.42,59.45) .. (365.9,59.45) .. controls (359.38,59.45) and (354.09,54.16) .. (354.09,47.63) -- cycle ;
\draw    (214.9,47.73) -- (267.61,47.73) ;
\draw [shift={(269.61,47.73)}, rotate = 180] \arrow  ;
\draw   (269.61,47.73) .. controls (269.61,40) and (275.88,33.73) .. (283.61,33.73) .. controls (291.34,33.73) and (297.61,40) .. (297.61,47.73) .. controls (297.61,55.47) and (291.34,61.73) .. (283.61,61.73) .. controls (275.88,61.73) and (269.61,55.47) .. (269.61,47.73) -- cycle ;
\draw    (184.01,104.33) -- (237.01,103.85) ;
\draw [shift={(239.01,103.83)}, rotate = 179.48] \arrow  ;
\draw   (238.3,104.23) .. controls (238.3,96.5) and (244.57,90.23) .. (252.3,90.23) .. controls (260.03,90.23) and (266.3,96.5) .. (266.3,104.23) .. controls (266.3,111.97) and (260.03,118.23) .. (252.3,118.23) .. controls (244.57,118.23) and (238.3,111.97) .. (238.3,104.23) -- cycle ;
\draw   (240.49,104.23) .. controls (240.49,97.71) and (245.78,92.42) .. (252.3,92.42) .. controls (258.82,92.42) and (264.11,97.71) .. (264.11,104.23) .. controls (264.11,110.76) and (258.82,116.05) .. (252.3,116.05) .. controls (245.78,116.05) and (240.49,110.76) .. (240.49,104.23) -- cycle ;
\draw   (156.01,104.33) .. controls (156.01,96.6) and (162.28,90.33) .. (170.01,90.33) .. controls (177.74,90.33) and (184.01,96.6) .. (184.01,104.33) .. controls (184.01,112.07) and (177.74,118.33) .. (170.01,118.33) .. controls (162.28,118.33) and (156.01,112.07) .. (156.01,104.33) -- cycle ;
\draw    (136.05,104.15) -- (154.01,104.32) ;
\draw [shift={(156.01,104.33)}, rotate = 180.52] \arrow  ;
\draw   (351.7,104.63) .. controls (351.7,96.9) and (357.97,90.63) .. (365.7,90.63) .. controls (373.43,90.63) and (379.7,96.9) .. (379.7,104.63) .. controls (379.7,112.37) and (373.43,118.63) .. (365.7,118.63) .. controls (357.97,118.63) and (351.7,112.37) .. (351.7,104.63) -- cycle ;
\draw   (353.89,104.63) .. controls (353.89,98.11) and (359.18,92.82) .. (365.7,92.82) .. controls (372.22,92.82) and (377.51,98.11) .. (377.51,104.63) .. controls (377.51,111.16) and (372.22,116.45) .. (365.7,116.45) .. controls (359.18,116.45) and (353.89,111.16) .. (353.89,104.63) -- cycle ;
\draw    (331.73,104.45) -- (349.7,104.62) ;
\draw [shift={(351.7,104.63)}, rotate = 180.52] \arrow  ;
\draw    (377.93,97.38) .. controls (401.69,87.32) and (399.91,123.21) .. (379.26,113.23) ;
\draw [shift={(377.63,112.37)}, rotate = 29.82] \arrow  ;
\draw  [dashed] (324.29,80.83) -- (416.8,80.83) -- (416.8,126.43) -- (324.29,126.43) -- cycle ;

\draw (131.8,37.47) node [anchor=north west][inner sep=0.75pt]   [align=left] {$\displaystyle \mathcal{A}^{S}$:};
\draw (96.4,94.67) node [anchor=north west][inner sep=0.75pt]   [align=left] {$\displaystyle \mathcal{A}^{R}$:};
\draw (291.4,95.67) node [anchor=north west][inner sep=0.75pt]   [align=left] {$\displaystyle \mathcal{A}^{T}$:};
\draw (395.8,99.37) node [anchor=north west][inner sep=0.75pt]   [align=left] {$\displaystyle R$};
\draw (311.8,33.97) node [anchor=north west][inner sep=0.75pt]   [align=left] {$\displaystyle T$};
\draw (194,45.27) node [anchor=north west][inner sep=0.75pt]  [font=\small] [align=left] {$\displaystyle s_{0}$};
\draw (359,44.17) node [anchor=north west][inner sep=0.75pt]  [font=\small] [align=left] {$\displaystyle s_{2}$};
\draw (231.4,33.35) node [anchor=north west][inner sep=0.75pt]   [align=left] {$\displaystyle R$};
\draw (277,45.07) node [anchor=north west][inner sep=0.75pt]  [font=\small] [align=left] {$\displaystyle s_{1}$};
\draw (202.2,92.57) node [anchor=north west][inner sep=0.75pt]   [align=left] {$\displaystyle a$};
\draw (245.4,100.77) node [anchor=north west][inner sep=0.75pt]  [font=\small] [align=left] {$\displaystyle r_{1}$};
\draw (163.4,101.67) node [anchor=north west][inner sep=0.75pt]  [font=\small] [align=left] {$\displaystyle r_{0}$};
\draw (358.8,99.17) node [anchor=north west][inner sep=0.75pt]  [font=\small] [align=left] {$\displaystyle t_{0}$};

\end{tikzpicture}
    \vspace{-2em}
    \caption{A deterministic VRA with no equivalent SPA.}
    \label{fig:det-vra-ex}
        \vspace{-1em}
\end{figure}

\begin{proof}
    Recall that a deterministic VRA $\mathcal{A}= \langle\pdwAlph,\procAlpha,\AutomataSet,\mathcal{A}^S\rangle$ is a VRA such that, all its automata in $\AutomataSet \cup \{\mathcal{A}^S\}$ are DFAs, and for all states $q\in \vraStates{A}$ and distinct $J,J'\in \procAlpha$, if there exist $(q,J,p),(q,J',p')\in\vraTrans{A}$, then $f_{\call}(J)\neq f_{\call}(J')$. Therefore, by definition, an SPA is a deterministic VRA.

    We now prove 
    that the deterministic VRA $\mathcal{A}=\langle \pdwAlph, \procAlpha,\{\mathcal{A}^R,\mathcal{A}^T\},\mathcal{A}^S \rangle$,
    with $\pdwAlph = \{a\} \cup \{c\} \cup \{r\}$ and $\procAlpha = \{R,T\}$,
    depicted in \autoref{fig:det-vra-ex}, has no equivalent SPA using the same pushdown alphabet $\pdwAlph$.

    By contradiction, suppose that there exists an SPA $\mathcal{B}=\langle \pdwAlph, \procAlpha', \AutomataSet', \mathcal{B}^S\rangle$ accepting $\Lang{\mathcal{B}}=\Lang{\mathcal{A}}$. 
    It is easy to see that $carcr\in \Lang{\mathcal{A}}=\Lang{\mathcal{B}}$. By \autoref{prop:vra-run-to-nfa-trans}, we have $JJ'\in \regLang{\mathcal{B}^S}$ for some $J,J'\in \procAlpha'$ such that $f(J)=f(J')=\langle c,r\rangle$, $a\in\Lang{\mathcal{B}^J}$, and $\varepsilon \in \Lang{\mathcal{B}^{J'}}$. Since $f_{\call}(J)=f_{\call}(J')=c$ and $\mathcal{B}$ is an SPA, we deduce that $J=J'$. As $a \in \Lang{\mathcal{B}^J}$ and $JJ\in \regLang{\mathcal{B}^S}$, by \autoref{prop:vra-run-to-nfa-trans}, we deduce that $carcar\in \Lang{\mathcal{B}}$, which is a contradiction because $carcar\notin \Lang{\mathcal{A}}$.\qed
\end{proof}

\section{Equivalence with Visibly Pushdown Automata}\phantomsection\label{ax:vra-vpa}

\emph{Visibly pushdown automata} (VPAs) form a subclass of pushdown automata~\cite{alur2004vpl}. In this section, we show that VRAs and VPAs are  equivalent models. We first recall the definition of VPAs and their semantics.

\begin{definition}[Visibly pushdown automaton] \label{def:vpa}
A \emph{visibly pushdown automaton} is a tuple $\mathcal{A}=\langle\pdwAlph,\Gamma,Q,I,F,\delta\rangle$ where 
\begin{itemize}
\item $\pdwAlph = \intAlpha \cup \calAlpha \cup \retAlpha$ is a pushdown alphabet; \item $\Gamma$ is a \emph{stack alphabet}; 
\item $Q$ is a finite set of states; 
\item $I\subseteq Q$ is a set of initial states;
\item $F\subseteq Q$ is a set of final states; and 
\item $\delta$ is a set of transitions of the form $\delta = \intTrans \cup \callTrans \cup \retTrans$ where:
\begin{itemize}
    \item $\intTrans \subseteq Q \times \intAlpha \times Q$ is the set of \emph{internal} transitions;
    \item $\callTrans \subseteq Q \times \calAlpha \times Q \times \Gamma$ is the set of \emph{call} transitions;
    \item $\retTrans \subseteq Q \times \retAlpha \times \Gamma \times Q$ is the set of \emph{return} transitions.
\end{itemize}
\end{itemize}
The size of a VPA $\mathcal{A}$, denoted by $|\mathcal{A}|$, is $|Q|+|\delta|$. 
\end{definition}

Let $\mathcal{A}=\langle\pdwAlph,\Gamma,Q,I,F,\delta\rangle$ be a VPA. Transitions $(q,a,q')\in \intTrans$, $(q,a,q',\gamma)\in\callTrans$, and $(q,a,\gamma,q')\in\retTrans$ are respectively written $q\xrightarrow{a}q' \in \intTrans$, $q\xrightarrow{a/\gamma}q'\in \callTrans$, and $q\xrightarrow{a[\gamma]}q' \in \retTrans$. 
As for VRAs, the semantics of VPAs use configurations $\langle q,\sigma\rangle \in Q\times \Gamma^*$.
A \emph{stacked run} of $\mathcal{A}$ on $w=a_1\dots a_n \in \pdwAlph^*$ is a sequence $\langle q_0,\sigma_0 \rangle\xrightarrow{a_1}\langle q_1,\sigma_1 \rangle\xrightarrow{a_2} \dots \xrightarrow{a_n}\langle q_n,\sigma_n \rangle$, where for all $i \in [1,n]$ 
the following is respected: 
\begin{itemize}
    \item If $a_i \in \intAlpha$, there is a transition $q_{i-1}\xrightarrow{a_i}q_i \in \intTrans$ and $\sigma_i = \sigma_{i-1}$; 
    \item If $a_i \in \calAlpha$, there is a  transition $q_{i-1}\xrightarrow{a_i/\gamma}q_i \in \callTrans$ verifying $\sigma_i =\gamma \sigma_{i-1}$; 
    \item If $a_i \in \retAlpha$, there is a  transition $q_{i-1}\xrightarrow{a_i[\gamma]}q_i \in \retTrans$ verifying $\sigma_{i-1} =\gamma \sigma_{i}$.
\end{itemize}
The stacked run on $w$ can alternatively be written $\langle q_0,\sigma_0 \rangle\xrightarrow{w}\langle q_n,\sigma_n \rangle$. We denote by $\runsA{\mathcal{A}}$ the set of stacked runs of $\mathcal{A}$.
A word $w$ is accepted by $\mathcal{A}$ if there is an accepting  stacked run on $w$, i.e., that begins in an initial configuration $\langle q_o,\varepsilon \rangle$ with $q_o\in I$ and ends in a final configuration $\langle q_f,\varepsilon \rangle$ with $q_f\in F$. The language $\Lang{\mathcal{A}}$ of~$\mathcal{A}$ is the set of all words accepted by~$\mathcal{A}$:
\[ \Lang{\mathcal{A}} = \left\{ w \in \pdwAlph^* \mid \exists q_i \in I, q_f\in F, \langle q_i,\varepsilon \rangle\xrightarrow{w}\langle q_f,\varepsilon \rangle \in \runsA{\mathcal{A}}\right\}.\]
With this definition of $\Lang{\mathcal{A}}$, it is easy to verify that $\Lang{\mathcal{A}} \subseteq \wm\pdwAlph$.\footnote{The original definition of VPAs~\cite{alur2004vpl} also allows accepting ill-matched words. In this
paper, we focus on languages of well-matched words.}


A VPA $\mathcal{A}$ is \emph{deterministic} if $|I|=1$ and, for all $q\in Q$ and $a\in \intAlpha$ (resp.\ $c\in \calAlpha$, {$(r,\gamma)\in \retAlpha \times \Gamma$}), there exists at most one transition $q\xrightarrow{a}q'\in \intTrans$ (resp.\ $q \xrightarrow{c/\gamma}q'\in \callTrans$, $q\xrightarrow{r[\gamma]}q'\in \retTrans$). Any VPA $\mathcal{A}$ is equivalent to a deterministic VPA $\mathcal{B}$ with $|\mathcal{B}| = 2^{\mathcal{O}(|\mathcal{A}|^2)}$ \cite{alur2004vpl}, and this bound is tight \cite{martynova2024exact}.

\begin{figure}[t]
    \centering
    \begin{tikzpicture}[x=0.75pt,y=0.75pt,yscale=-1,xscale=1]

\draw   (270.06,72.38) .. controls (270.06,64.65) and (276.33,58.38) .. (284.06,58.38) .. controls (291.79,58.38) and (298.06,64.65) .. (298.06,72.38) .. controls (298.06,80.11) and (291.79,86.38) .. (284.06,86.38) .. controls (276.33,86.38) and (270.06,80.11) .. (270.06,72.38) -- cycle ;
\draw   (272.25,72.38) .. controls (272.25,65.85) and (277.53,60.57) .. (284.06,60.57) .. controls (290.58,60.57) and (295.87,65.85) .. (295.87,72.38) .. controls (295.87,78.9) and (290.58,84.19) .. (284.06,84.19) .. controls (277.53,84.19) and (272.25,78.9) .. (272.25,72.38) -- cycle ;
\draw    (248.09,72.2) -- (268.06,72.36) ;
\draw [shift={(270.06,72.38)}, rotate = 180.48] \arrow ;
\draw    (295.97,64.47) -- (370.4,64.75) ;
\draw [shift={(372.4,64.75)}, rotate = 180.21] \arrow ;
\draw    (370.5,79) -- (298.4,79.73) ;
\draw [shift={(296.4,79.75)}, rotate = 359.42] \arrow ;
\draw    (376.5,61) .. controls (366.75,37.11) and (402.63,40.32) .. (392.38,60.88) ;
\draw [shift={(391.5,62.5)}, rotate = 300.58] \arrow ;
\draw   (369.06,73.41) .. controls (369.06,65.68) and (375.33,59.41) .. (383.06,59.41) .. controls (390.79,59.41) and (397.06,65.68) .. (397.06,73.41) .. controls (397.06,81.14) and (390.79,87.41) .. (383.06,87.41) .. controls (375.33,87.41) and (369.06,81.14) .. (369.06,73.41) -- cycle ;
\draw    (277,60) .. controls (267.25,36.11) and (303.13,38.37) .. (292.88,58.89) ;
\draw [shift={(292,60.5)}, rotate = 300.58] \arrow ;

\draw (276.4,69.17) node [anchor=north west][inner sep=0.75pt]  [font=\small] [align=left] {$\displaystyle q_{0}$};
\draw (272.77,26.81) node [anchor=north west][inner sep=0.75pt]   [align=left] {$\displaystyle r[ \gamma ]$};
\draw (322.27,81) node [anchor=north west][inner sep=0.75pt]   [align=left] {$\displaystyle r[ \gamma ]$};
\draw (328.77,53.51) node [anchor=north west][inner sep=0.75pt]   [align=left] {$\displaystyle a$};
\draw (374.27,27.5) node [anchor=north west][inner sep=0.75pt]   [align=left] {$\displaystyle c/\gamma $};
\draw (376.4,69.67) node [anchor=north west][inner sep=0.75pt]  [font=\small] [align=left] {$\displaystyle q_{1}$};

\end{tikzpicture}
    \vspace{-2em}
    \caption{A two-state deterministic VPA.}
    \label{fig:ex-vpa}
    \vspace{-1em}
\end{figure}
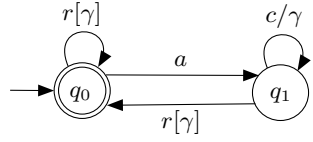

\begin{example}
    \autoref{fig:ex-vpa} depicts a deterministic VPA $\mathcal{A}=(\pdwAlph,\Gamma, Q, I, F, \delta)$, with $\pdwAlph = \{a\} \cup \{c\} \cup \{r\}$, $\Gamma = \{\gamma\}$, and $I=F=\{q_0\}$.
    As an example, the well-matched word  $w=accrar$ is accepted by $\mathcal{A}$, with the accepting stacked run:
    \[\langle q_0,\varepsilon \rangle \xrightarrow{a} \langle q_1,\varepsilon\rangle \xrightarrow{c} \langle q_1,\gamma\rangle \xrightarrow{c} \langle q_1,\gamma\gamma\rangle \xrightarrow{r} \langle q_0,\gamma\rangle \xrightarrow{a} \langle q_1,\gamma\rangle \xrightarrow{r} \langle q_0,\varepsilon \rangle.\] 
    
    \vspace{-1.5em}
\end{example}

Let us now state that VRAs and VPAs have the same expressive power.

\vravpa*

To prove \autoref{th:vpa-vra}, we provide the construction of a VPA (resp. VRA) equivalent to a given VRA (resp. VPA) and a proof of its correctness. We then give an illustration of the constructions.

\subsection{VPA Equivalent to a Given VRA}

Taking inspiration from \autoref{fig:sem-vra}, we can easily transform any VRA $\mathcal{A}$ into an equivalent VPA $\mathcal{B}$: for each procedural transition on $J\in \procAlpha$, we must add a call transition to all initial states of $\mathcal{A}^J$ and a return transition from all final states of $\mathcal{A}^J$. In that way, the transitions in the VPA $\mathcal{B}$ match the semantics of VRAs.

However, doing so results in a quadratic number of call and return transitions.
To obtain a linear number of transitions, we must ensure that the number of initial and final states of each automaton is constant. Since each $\mathcal{A}^J$ is an NFA, we can easily construct an equivalent NFA with one initial state and at most two final states. 

\begin{construction}\label{const:nfa-to-constant-nfa}
    Let $\mathcal{A}=\langle\Sigma, Q,I,F,\delta \rangle$ be an NFA. We construct an equivalent NFA $\mathcal{A}' = \langle\Sigma, Q',I',F',\delta' \rangle$ with:
    \begin{itemize}
        \item $Q'= Q \cup \{q_i, q_f\}$; 
        \item $I' = \{q_i\}$;
        \item $F'= \{q_i,q_f\}$ if $ I\cap F \neq  \varnothing$, otherwise $F'=\{q_f\}$;
        \item $\delta' = \delta \cup \delta_i \cup \delta_f$, with:
        \begin{itemize}
            \item $\forall q \in I: q\xrightarrow{a} p \in \delta \iff q_i \xrightarrow{a} p \in \delta_i$;
            \item $\forall p \in F: q\xrightarrow{a}p \in \delta \cup \delta_i \iff q\xrightarrow{a} q_f \in \delta_f$.
        \end{itemize}
    \end{itemize}
\end{construction}

\Cref{const:nfa-to-constant-nfa} is a classical construction in automata theory. We omit the proof that $L(\mathcal{A})=L(\mathcal{A}')$, as this is a classical result. Notice that $|\mathcal{A}'|= \mathcal{O}(|\mathcal{A}|)$. 

\begin{construction}\label{const:vra-to-vpa}
Let $\mathcal{A}=\langle\pdwAlph,\procAlpha,\AutomataSet,\mathcal{A}^S\rangle$ be a VRA. 
We construct the equivalent VPA $\mathcal{B}=\langle\pdwAlph,\Gamma,\vraStates{B},I,F,\vraTrans{B}\rangle$, where: $\Gamma = \vraStates{A}$, $\vraStates{B}=\vraStates{A}$, $I=I^S$, $F=F^S$, and $\vraTrans{B}=\delta_{\internal}\cup\delta_{\call}\cup\delta_{\return}$ with: 
\begin{itemize}
    \item For all $a\in \intAlpha$: $q\xrightarrow{a}p \in \vraTrans{A} \Leftrightarrow q\xrightarrow{a}p\in\delta_{\internal}$;
    \item For all $J\in \procAlpha$: $q\xrightarrow{J}p \in \vraTrans{A} \Leftrightarrow \left\{ \begin{array}{l}
         \forall q_i \in I^J : q \xrightarrow{f_c(J)[p]}q_i\in\delta_{\call}\text{, and,}  \\
         \forall q_f \in F^J : q_f \xrightarrow{f_r(J)[p]}p \in\delta_{\return}
    \end{array}\right.$.
\end{itemize}
\end{construction}

Let us discuss the size of $\mathcal{B}$. It is clear that it has as many states as $\mathcal{A}$ and as many internal transitions as $\mathcal{A}$. For the call and return transitions, for a \proceduralSymbol\ $J\in \procAlpha$, we have as many call and return transitions as the number of initial and final states of $\mathcal{A}^J$ times the number of \proceduralTransition s on $J$. Since the number of initial and final states of $\mathcal{A}^J$ can be constant using \Cref{const:nfa-to-constant-nfa}, the number of call and return transitions is linear in the number of \proceduralTransition s. We conclude that $\mathcal{B}$ has a size of $|\mathcal{B}|=\mathcal{O}(|\mathcal{A}|)$.

Notice that each part of the output can be computed from a constant-size portion of the input, so the reduction never needs to store more than a constant number of pointers at a time. As each pointer requires only $\mathcal{O}(\log n)$ bits, the entire construction is computable in logarithmic space (see \cite{sipser1996introduction} for more details on logspace reduction).
Let us prove that this construction is correct.

\begin{property}\label{pro:vra-run-to-vpa-run}
    Let $\mathcal{A}$ be a VRA and $\mathcal{B}$ be the VPA obtained with \Cref{const:vra-to-vpa}. For all $q,p\in \vraStates{A}=\vraStates{B}$ and $w\in \wm\pdwAlph$:
    \[ \langle q, \varepsilon \rangle \xrightarrow{w} \langle p,\varepsilon \rangle \in \runsA{\mathcal{A}} \iff \langle q, \varepsilon \rangle \xrightarrow{w} \langle p,\varepsilon \rangle \in \runsA{\mathcal{B}}.\]
\end{property}
\begin{proof}
    We proceed by structural induction on well-matched words. 
    \begin{itemize}
        \item $w\in \intAlpha^*$: This is trivial since the set of internal transitions of the VRA and the VPA are the same.
        \item $w = cw'r$ (with $c\in \calAlpha$, $r\in \retAlpha$ and $w'\in \wm\pdwAlph$):
        \begin{itemize}
            \item[$\Rightarrow$] By \Cref{prop:vra-run-to-nfa-trans}, there exists a transition $q\xrightarrow{J}p\in \vraTrans{A}$ such that $J\in \procAlpha^{\langle c,r\rangle}$, $w'\in \Lang{\mathcal{A}^J}$. Since $w'\in \Lang{\mathcal{A}^J}$, there exists a $q_i \in I^J$ and $q_f \in F^J$ such that $\langle q_i, \varepsilon \rangle \xrightarrow{w'} \langle q_f,\varepsilon \rangle \in \runsA{\mathcal{A}}$. By structural induction, such a stacked run also exists in $\mathcal{B}$. Additionally, since $q_i \in I^J$, $q_f \in F^J$ and $q\xrightarrow{J}p\in \vraTrans{A}$, by \Cref{const:vra-to-vpa}, there exists the transitions $q\xrightarrow{c/p}q_i \in \delta_{\call}$ and $q_f\xrightarrow{r[p]}p \in \delta_{\return}$. Hence, there exists the stacked run: $$\langle q, \varepsilon \rangle \xrightarrow{c} \langle q_i,p\rangle \xrightarrow{w'}\langle q_f,p \rangle \xrightarrow{r} \langle p,\varepsilon \rangle \in \runsA{\mathcal{B}}.$$
            \item[$\Leftarrow$] The other implication is proved similarly.
        \end{itemize}
        \item $w= w_1\cdot w_2$ ($w_1,w_2\in \wm\pdwAlph$): This is true by structural induction.  \hfill   $\lrcorner$ 
    \end{itemize}
\end{proof}
Since the set of initial (resp.\ final) states of $\mathcal{B}$ is $I^S$ (resp.\ $F^S$), thanks to \autoref{pro:vra-run-to-vpa-run}, we have $w\in \Lang{\mathcal{A}}$ iff $w\in\Lang{\mathcal{B}}$, i.e., $\Lang{\mathcal{A}}=\Lang{\mathcal{B}}$, which concludes the correctness of \autoref{const:vra-to-vpa}.

\subsection{VRA Equivalent to a Given VPA}

Let us first provide the intuition needed to build a VRA $\mathcal{A}$ equivalent to a given VPA $\mathcal{B}$. Consider a  stacked run of $\mathcal{B}$ equal to $\langle q,\varepsilon \rangle\xrightarrow{c}\langle q',\gamma \rangle\xrightarrow{w}\langle p',\gamma \rangle \xrightarrow{r} \langle p,\varepsilon \rangle$. 
For an equivalent \recursiveRun\ to exist in a VRA, there should exist a transition $q \xrightarrow{J}p$ for some \proceduralSymbol\ $J\in\procAlpha$ such that $f(J) = \langle c,r\rangle$,  $q'\in I^J$, and $p'\in F^J$ (as exposed in \autoref{fig:sem-vra} and \autoref{prop:vra-run-to-nfa-trans}). 
{Hence, we identify a new procedural symbol $J$ with each such quadruplet $\langle q',p',c,r \rangle$, $q',p'\in Q$, $c\in \calAlpha$, and $r\in \retAlpha$, and we make a copy $\mathcal{A}^J$ of $\mathcal{B}$ as an FA with the same states and internal transitions as $\mathcal{B}$, $I^J=\{q'\}$ and $F^J=\{p'\}$, and $f(J)=\langle c,r\rangle$.} 
Then, in all such copies, we replace the call transitions $q\xrightarrow{c/ \gamma} q'\in \vraTrans{B}$ and return transitions $p'\xrightarrow{r[\gamma]}p\in \vraTrans{B}$ of the VPA $\mathcal{B}$ 
by the \proceduralTransition\  $q\xrightarrow{J}p$, with $J=\langle q',p',c,r\rangle$. Let us describe formally the construction.


\begin{construction}\label{const:vpa-to-vra}
    Let $\mathcal{B}=\langle\pdwAlph, \Gamma,\vraStates{B},I,F,\vraTrans{B}\rangle$ be a VPA, with $\pdwAlph = \intAlpha \cup \calAlpha \cup \retAlpha$ and $\vraTrans{B}=\delta_{\internal}\cup\delta_{\call}\cup \delta_{\return}$. Then, we construct the equivalent VRA $\mathcal{A}=\langle\pdwAlph,\procAlpha,\AutomataSet,\mathcal{A}^{S}\rangle$ in the following way:
\begin{itemize}
        \item $\procAlpha=\vraStates{B}\times \vraStates{B}\times \calAlpha \times \retAlpha$, and the linking function is defined as $f(\langle q,p,c,r \rangle)=\langle c,r\rangle$ for all $\langle q,p,c,r \rangle\in \procAlpha$;
        \item Let $\mathcal{A}^J=\langle\intAlpha\cup\procAlpha,Q^J,I^J,F^J,\delta^J\rangle\in \AutomataSet$ be an FA of $\mathcal{A}$, with $J=\langle q,p,c,r \rangle\in \procAlpha$. The set of states $Q^J$ is a copy of $\vraStates{B}$. To express this copy, we create a bijection $g^J:\vraStates{B}\to Q^J$. Therefore, we define $Q^J=\{g^J(q_1)\mid q_1\in \vraStates{B}\}$, $I^J=\{g^J(q)\}$, $F^J=\{g^J(p)\}$, and $\delta^J$ is constructed as follows:
        \begin{itemize}
            \item $\forall q_1\xrightarrow{a}p_1\in \delta_{\internal} : g^J(q_1)\xrightarrow{a}g^J(p_1) \in \delta^J$;
            \item  $\forall q_1\xrightarrow{c'/\gamma}q_2\in \delta_{\call}, \forall p_2\xrightarrow{r'[\gamma]}p_1 \in \delta_{\return}: g^J(q_1) \xrightarrow{J'} g^J(p_1) \in \delta^J$, with $J' = \langle q_2,p_2,c',r' \rangle$.  
        \end{itemize}
        \item $\mathcal{A}^{S}$ is defined similarly to all $\mathcal{A}^J\in\AutomataSet$, but with $I^S$ a copy of $I$ and $F^S$ a copy of $F$.
    \end{itemize}
    
\end{construction}

Let us discuss the size of $\mathcal{A}$. It is clear that each FA $\mathcal{A}^J$, with $J\in \procAlpha\cup\{S\}$, has a number of states $|Q^J|=|\vraStates{B}|$ and a number of transitions $|\delta^J|=\mathcal{O}(|\delta_\internal|+|\delta_\call|\cdot |\delta_\return|)$, meaning that $|\mathcal{A}^J|=\mathcal{O}(|\mathcal{B}|^2)$. Since there are $|\procAlpha|+1$ such automata, and $|\procAlpha|=|\vraStates{B}|^2\cdot |\calAlpha|\cdot |\retAlpha|$, we conclude that the size of $\mathcal{A}$ is $|\mathcal{A}|=\mathcal{O}(|\mathcal{B}^4|)$ (we consider the size of the input alphabet fixed).

Again, this construction can be computed in logarithmic space, since it always translates a constant number of components of the input into a component of the output (note that, given $J\in \procAlpha$, constructing the bijection $g^J:\vraStates{B}\to Q^J$ can be easily computed in logarithmic space). 

Let us explain why this construction is correct and sound.

\begin{property}\label{pro:vpa-run-to-vra-run}
    For all $w\in\wm\pdwAlph$, $q,p\in \vraStates{B}$, $J\in \procAlpha\cup \{S\}$: 
    \[\langle q,\varepsilon \rangle\xrightarrow{w}\langle p,\varepsilon \rangle\in \runsA{\mathcal{B}}\iff  \langle g^J(q),\varepsilon \rangle\xrightarrow{w}\langle g^J(p),\varepsilon \rangle\in \runsA{\mathcal{A}}.\]
\end{property}
\begin{proof} First, note that, since the sets of transitions of all automata $\mathcal{A}^J \in \AutomataSet\cup \{\mathcal{A}^{S}\}$ are defined in the same way, if a \recursiveRun\ exists in one of them, it exists in all of them. We prove \autoref{pro:vpa-run-to-vra-run} by structural induction on $w\in\wm\pdwAlph$:
\begin{itemize}
    \item $w\in \intAlpha^*$: the property holds as internal transitions in $\mathcal{A}^J$ are the same as those of $\delta_{\internal}$.
    \item $w=cw'r$ (with $c\in\calAlpha$, $r\in\retAlpha$, $w'\in\wm\pdwAlph$):
    \begin{itemize}
        \item[$\Rightarrow$] The stacked run is decomposed into $\langle q,\varepsilon \rangle \xrightarrow{c} \langle q',\gamma \rangle \xrightarrow{w'} \langle p',\gamma \rangle \xrightarrow{r}\langle p,\varepsilon \rangle \in \runsA{\mathcal{B}}$ with $q',p'\in \vraStates{B}$ and $\gamma \in \Gamma$. This means that $q\xrightarrow{c / \gamma} q' \in \delta_{\call}$ and $p'\xrightarrow{r[\gamma]} p \in \delta_{\return}$. By construction, it follows that $g^J(q) \xrightarrow{K}g^J(p)\in\delta^J$, with $K = \langle q',p',c,r \rangle$. 
        As $w' \in \wm\pdwAlph$, by structural induction, we have $\langle g^{K}(q'),\varepsilon \rangle \xrightarrow{w'} \langle g^{K}(p'),\varepsilon \rangle \in \runsA{\mathcal{A}}$, meaning that $w'\in \Lang{\mathcal{A}^K}$ because $g^{K}(q')\in I^{K}$ and $g^{K}(p')\in F^{K}$. By \autoref{prop:vra-run-to-nfa-trans}, as $f(K)=\langle c,r\rangle$, we have  the \recursiveRun\ $\langle g^J(q),\varepsilon \rangle\xrightarrow{cw'r}\langle g^J(p),\varepsilon \rangle\in \runsA{\mathcal{A}}$.

        \item[$\Leftarrow$] The other implication is proved  similarly.

    \end{itemize}
    \item $w=w_1\cdot w_2$ ($w_1,w_2\in\wm\pdwAlph$): This is true by structural induction. \hfill   $\lrcorner$ 
\end{itemize}
\end{proof}

Thanks to \autoref{pro:vpa-run-to-vra-run}, since $I^S$ and $F^S$ are copies of respectively $I$ and $F$, we get that $ \Lang{\mathcal{A}}= \Lang{\mathcal{A}^{S}}=\Lang{\mathcal{B}}$.

\subsection{Illustration of the Constructions}
\begin{figure}[t]
    \centering
    \begin{subfigure}{.48\textwidth}
        \centering
        \begin{tikzpicture}[x=0.75pt,y=0.75pt,yscale=-1,xscale=1]

\draw   (116.3,69.13) .. controls (116.3,61.4) and (122.57,55.13) .. (130.3,55.13) .. controls (138.03,55.13) and (144.3,61.4) .. (144.3,69.13) .. controls (144.3,76.87) and (138.03,83.13) .. (130.3,83.13) .. controls (122.57,83.13) and (116.3,76.87) .. (116.3,69.13) -- cycle ;
\draw   (185.3,68.63) .. controls (185.3,60.9) and (191.57,54.63) .. (199.3,54.63) .. controls (207.03,54.63) and (213.3,60.9) .. (213.3,68.63) .. controls (213.3,76.37) and (207.03,82.63) .. (199.3,82.63) .. controls (191.57,82.63) and (185.3,76.37) .. (185.3,68.63) -- cycle ;
\draw   (187.49,68.63) .. controls (187.49,62.11) and (192.78,56.82) .. (199.3,56.82) .. controls (205.82,56.82) and (211.11,62.11) .. (211.11,68.63) .. controls (211.11,75.16) and (205.82,80.45) .. (199.3,80.45) .. controls (192.78,80.45) and (187.49,75.16) .. (187.49,68.63) -- cycle ;
\draw   (116.3,127.13) .. controls (116.3,119.4) and (122.57,113.13) .. (130.3,113.13) .. controls (138.03,113.13) and (144.3,119.4) .. (144.3,127.13) .. controls (144.3,134.87) and (138.03,141.13) .. (130.3,141.13) .. controls (122.57,141.13) and (116.3,134.87) .. (116.3,127.13) -- cycle ;
\draw   (185.3,126.63) .. controls (185.3,118.9) and (191.57,112.63) .. (199.3,112.63) .. controls (207.03,112.63) and (213.3,118.9) .. (213.3,126.63) .. controls (213.3,134.37) and (207.03,140.63) .. (199.3,140.63) .. controls (191.57,140.63) and (185.3,134.37) .. (185.3,126.63) -- cycle ;
\draw   (116.3,190.13) .. controls (116.3,182.4) and (122.57,176.13) .. (130.3,176.13) .. controls (138.03,176.13) and (144.3,182.4) .. (144.3,190.13) .. controls (144.3,197.87) and (138.03,204.13) .. (130.3,204.13) .. controls (122.57,204.13) and (116.3,197.87) .. (116.3,190.13) -- cycle ;

\draw    (138.75,201.51) .. controls (148.64,225.35) and (112.75,225.19) .. (122.88,204.71) ;
\draw [shift={(124,202.5)}, rotate = 120.26]\arrow  ;
\draw    (210.5,118.53) .. controls (234.28,108.76) and (233.95,144.65) .. (213.52,134.42) ;
\draw [shift={(211.5,133.54)}, rotate = 30.54]\arrow  ;
\draw    (118.79,135.48) .. controls (95,145.46) and (95.96,109.56) .. (116.53,119.61) ;
\draw [shift={(118.15,120.48)}, rotate = 210.03]\arrow  ;

\draw    (199.3,112.63) -- (199.3,84.63) ;
\draw [shift={(199.3,82.63)}, rotate = 90]\arrow  ;
\draw    (199.3,140.63) -- (199.24,189.75) -- (146.3,190.12) ;
\draw [shift={(144.3,190.13)}, rotate = 359.6]\arrow  ;
\draw    (136.7,139.73) -- (136.7,175.73) ;
\draw [shift={(136.7,177.73)}, rotate = 270]\arrow  ;
\draw    (124.7,177.2) -- (124.98,142.49) ;
\draw [shift={(125,140)}, rotate = 90.47]\arrow  ;
\draw    (92.33,68.95) -- (114.3,69.12) ;
\draw [shift={(116.3,69.13)}, rotate = 180.44]\arrow  ;
\draw    (130.3,83.13) -- (130.3,111.13) ;
\draw [shift={(130.3,113.13)}, rotate = 270]\arrow  ;
\draw    (144,130.13) -- (183.3,129.66) ;
\draw [shift={(185.8,129.63)}, rotate = 179.3]\arrow  ;

\draw (123.4,66.67) node [anchor=north west][inner sep=0.75pt]  [font=\small] [align=left] {$\displaystyle s_{0}$};
\draw (157.5,118.87) node [anchor=north west][inner sep=0.75pt]   [align=left] {$\displaystyle a$};
\draw (192.4,65.17) node [anchor=north west][inner sep=0.75pt]  [font=\small] [align=left] {$\displaystyle s_{1}$};
\draw (192.9,123.17) node [anchor=north west][inner sep=0.75pt]  [font=\small] [align=left] {$\displaystyle r_{1}$};
\draw (124.4,186.67) node [anchor=north west][inner sep=0.75pt]  [font=\small] [align=left] {$\displaystyle t_{0}$};
\draw (124.9,124.67) node [anchor=north west][inner sep=0.75pt]  [font=\small] [align=left] {$\displaystyle r_{0}$};
\draw (75.3,120) node [anchor=north west][inner sep=0.75pt]   [align=left] {$\displaystyle c/r_{1}$};
\draw (103.1,87.37) node [anchor=north west][inner sep=0.75pt]   [align=left] {$\displaystyle c/s_{1}$};
\draw (200.1,90.97) node [anchor=north west][inner sep=0.75pt]   [align=left] {$\displaystyle r/[ s_{1}]$};
\draw (140.75,204.51) node [anchor=north west][inner sep=0.75pt]   [align=left] {$\displaystyle c/t_{0} ,r[ t_{0}]$};
\draw (200.2,154.45) node [anchor=north west][inner sep=0.75pt]   [align=left] {$\displaystyle r[ t_{0}]$};
\draw (137.7,146.07) node [anchor=north west][inner sep=0.75pt]   [align=left] {$\displaystyle c/r_{0}$};
\draw (95.8,147.65) node [anchor=north west][inner sep=0.75pt]   [align=left] {$\displaystyle  \begin{array}{{r}}
r[ r_{0}]\\
c/t_{0}
\end{array}$};
\draw (229.78,120) node [anchor=north west][inner sep=0.75pt]   [align=left] {$\displaystyle r[ r_{1}]$};

\end{tikzpicture}
        \vspace{-2em}
        \caption{VPA equivalent to the VRA depicted in \autoref{fig:ex-vra}.}
        \label{fig:vra-to-vpa}
    \end{subfigure}
    \hfill
    \begin{subfigure}{.48\textwidth}
    \centering
        \begin{tikzpicture}[x=0.75pt,y=0.75pt,yscale=-1,xscale=1]

\draw    (91.45,80.2) -- (110.41,80.36) ;
\draw [shift={(112.41,80.38)}, rotate = 180.5] \arrow   ;
\draw    (139.27,74.77) -- (206.2,74.46) ;
\draw [shift={(208.2,74.45)}, rotate = 179.73] \arrow   ;
\draw    (207.77,84.67) -- (142.2,83.97) ;
\draw [shift={(140.2,83.95)}, rotate = 0.61] \arrow   ;
\draw  [dashed] (87.6,59.05) -- (238.6,59.05) -- (238.6,118.25) -- (87.6,118.25) -- cycle ;
\draw   (112.5,80.23) .. controls (112.5,72.5) and (118.77,66.23) .. (126.5,66.23) .. controls (134.23,66.23) and (140.5,72.5) .. (140.5,80.23) .. controls (140.5,87.97) and (134.23,94.23) .. (126.5,94.23) .. controls (118.77,94.23) and (112.5,87.97) .. (112.5,80.23) -- cycle ;
\draw   (114.69,80.23) .. controls (114.69,73.71) and (119.98,68.42) .. (126.5,68.42) .. controls (133.02,68.42) and (138.31,73.71) .. (138.31,80.23) .. controls (138.31,86.76) and (133.02,92.05) .. (126.5,92.05) .. controls (119.98,92.05) and (114.69,86.76) .. (114.69,80.23) -- cycle ;
\draw   (206.9,80.23) .. controls (206.9,72.5) and (213.17,66.23) .. (220.9,66.23) .. controls (228.63,66.23) and (234.9,72.5) .. (234.9,80.23) .. controls (234.9,87.97) and (228.63,94.23) .. (220.9,94.23) .. controls (213.17,94.23) and (206.9,87.97) .. (206.9,80.23) -- cycle ;
\draw    (236.53,142.45) -- (218.5,142.62) ;
\draw [shift={(216.5,142.63)}, rotate = 359.48] \arrow   ;
\draw    (120.87,137.17) -- (187.8,136.86) ;
\draw [shift={(189.8,136.85)}, rotate = 179.73] \arrow   ;
\draw    (189.37,147.07) -- (123.8,146.37) ;
\draw [shift={(121.8,146.35)}, rotate = 0.61] \arrow   ;
\draw  [dashed] (87.6,121.25) -- (238.6,121.25) -- (238.6,180.45) -- (87.6,180.45) -- cycle ;
\draw   (94.1,142.63) .. controls (94.1,134.9) and (100.37,128.63) .. (108.1,128.63) .. controls (115.83,128.63) and (122.1,134.9) .. (122.1,142.63) .. controls (122.1,150.37) and (115.83,156.63) .. (108.1,156.63) .. controls (100.37,156.63) and (94.1,150.37) .. (94.1,142.63) -- cycle ;
\draw   (190.69,142.63) .. controls (190.69,136.11) and (195.98,130.82) .. (202.5,130.82) .. controls (209.02,130.82) and (214.31,136.11) .. (214.31,142.63) .. controls (214.31,149.16) and (209.02,154.45) .. (202.5,154.45) .. controls (195.98,154.45) and (190.69,149.16) .. (190.69,142.63) -- cycle ;
\draw   (188.5,142.63) .. controls (188.5,134.9) and (194.77,128.63) .. (202.5,128.63) .. controls (210.23,128.63) and (216.5,134.9) .. (216.5,142.63) .. controls (216.5,150.37) and (210.23,156.63) .. (202.5,156.63) .. controls (194.77,156.63) and (188.5,150.37) .. (188.5,142.63) -- cycle ;
\draw    (236.33,204.45) -- (218.3,204.62) ;
\draw [shift={(216.3,204.63)}, rotate = 359.48] \arrow   ;
\draw    (120.67,199.17) -- (187.6,198.86) ;
\draw [shift={(189.6,198.85)}, rotate = 179.73] \arrow   ;
\draw    (189.17,209.07) -- (123.6,208.37) ;
\draw [shift={(121.6,208.35)}, rotate = 0.61] \arrow   ;
\draw  [dashed] (87.6,183.45) -- (238.6,183.45) -- (238.6,242.65) -- (87.6,242.65) -- cycle ;
\draw   (93.9,204.63) .. controls (93.9,196.9) and (100.17,190.63) .. (107.9,190.63) .. controls (115.63,190.63) and (121.9,196.9) .. (121.9,204.63) .. controls (121.9,212.37) and (115.63,218.63) .. (107.9,218.63) .. controls (100.17,218.63) and (93.9,212.37) .. (93.9,204.63) -- cycle ;
\draw   (96.09,204.63) .. controls (96.09,198.11) and (101.38,192.82) .. (107.9,192.82) .. controls (114.42,192.82) and (119.71,198.11) .. (119.71,204.63) .. controls (119.71,211.16) and (114.42,216.45) .. (107.9,216.45) .. controls (101.38,216.45) and (96.09,211.16) .. (96.09,204.63) -- cycle ;
\draw   (188.3,204.63) .. controls (188.3,196.9) and (194.57,190.63) .. (202.3,190.63) .. controls (210.03,190.63) and (216.3,196.9) .. (216.3,204.63) .. controls (216.3,212.37) and (210.03,218.63) .. (202.3,218.63) .. controls (194.57,218.63) and (188.3,212.37) .. (188.3,204.63) -- cycle ;

\draw (120.9,77.67) node [anchor=north west][inner sep=0.75pt]  [font=\small] [align=left] {$\displaystyle q_{0}$};
\draw (145,86.1) node [anchor=north west][inner sep=0.75pt]   [align=left] {$\displaystyle \langle q_{1} ,q_{1} ,c,r \rangle$\\$\displaystyle \langle q_{1} ,q_{0} ,c,r \rangle$};
\draw (168.27,63.11) node [anchor=north west][inner sep=0.75pt]   [align=left] {$\displaystyle a$};
\draw (214.9,76.67) node [anchor=north west][inner sep=0.75pt]  [font=\small] [align=left] {$\displaystyle q_{1}$};
\draw (65.1,71.13) node [anchor=north west][inner sep=0.75pt]   [align=left] {$\displaystyle \mathcal{B}^{S}$:};
\draw (24,132.23) node [anchor=north west][inner sep=0.75pt]   [align=left] {$\displaystyle \mathcal{B}^{\langle q_{1} ,q_{1} ,c,r \rangle}$:};
\draw (24.1,191.83) node [anchor=north west][inner sep=0.75pt]   [align=left] {$\displaystyle \mathcal{B}^{\langle q_{1} ,q_{0} ,c,r \rangle}$:};
\draw (101.5,136.07) node [anchor=north west][inner sep=0.75pt]  [font=\small] [align=left] {$\displaystyle q'_{0}$};
\draw (127,148.5) node [anchor=north west][inner sep=0.75pt]   [align=left] {$\displaystyle \langle q_{1} ,q_{1} ,c,r \rangle$\\$\displaystyle \langle q_{1} ,q_{0} ,c,r \rangle$};
\draw (149.87,125.51) node [anchor=north west][inner sep=0.75pt]   [align=left] {$\displaystyle a$};
\draw (196.5,135.07) node [anchor=north west][inner sep=0.75pt]  [font=\small] [align=left] {$\displaystyle q'_{1}$};
\draw (100.3,198.07) node [anchor=north west][inner sep=0.75pt]  [font=\small] [align=left] {$\displaystyle q''_{0}$};
\draw (127,210.5) node [anchor=north west][inner sep=0.75pt]   [align=left] {$\displaystyle \langle q_{1} ,q_{1} ,c,r \rangle$\\$\displaystyle \langle q_{1} ,q_{0} ,c,r \rangle$};
\draw (149.67,187.51) node [anchor=north west][inner sep=0.75pt]   [align=left] {$\displaystyle a$};
\draw (194.3,197.07) node [anchor=north west][inner sep=0.75pt]  [font=\small] [align=left] {$\displaystyle q''_{1}$};

\end{tikzpicture}
        \vspace{-2em}
        \caption{VRA equivalent to the VPA depicted in \autoref{fig:ex-vpa}.}
        \label{fig:vpa-to-vra}
    \end{subfigure}
    \caption{Illustration of \autoref{th:vpa-vra}.}
    \label{fig:vra-vpa}
        \vspace{-1em}
\end{figure}
    
    \autoref{fig:vra-vpa} illustrates \autoref{th:vpa-vra}. 
\autoref{fig:vra-to-vpa} illustrates \autoref{const:vra-to-vpa} and provides a VPA equivalent to the VRA from \autoref{fig:ex-vra}. For readability, we have omitted some return transitions that were not relevant to accept the target language.
\autoref{fig:vpa-to-vra} illustrates \autoref{const:vpa-to-vra} and provides a VRA equivalent to the VPA from \autoref{fig:ex-vpa}. Again, for readability, we have omitted some automata that were not relevant. Note that both the VPA and the VRA obtained this way are non-deterministic. 

\section{Proof of \autoref{prop:comp-codet-on-reg-lang}}\phantomsection\label{ax:converse-codet-on-reg-lang}

\codetCompRegLang*
\begin{proof}
    By hypothesis, all the \regularLanguage s $\regLang{\mathcal{A}^J}$,  $J\in \procAlpha^{\langle c,r \rangle}$, are pairwise disjoint, and their union is equal to $(\intAlpha\cup \procAlpha)^*$. We first prove  that the \recursiveLanguage s are pairwise disjoint too, i.e., $\mathcal{A}$ is \codeterministic. Then, as the union of the \regularLanguage s is $(\intAlpha\cup \procAlpha)^*$, we prove that the union of the \recursiveLanguage s is $\wm\pdwAlph$. 

    \begin{property}\label{pro:codet-on-reg-lang}
     Let $\mathcal{A}$ be a VRA. 
     If for all $ c\in\calAlpha$, $r\in\retAlpha$, and distinct $J,J'\in\procAlpha^{\langle c,r\rangle}$, $\regLang{\mathcal{A}^{J}} \cap \regLang{\mathcal{A}^{J'}} = \varnothing$, then $\mathcal{A}$ is \codeterministic. 
    \end{property}
    \begin{proof}[of \autoref{pro:codet-on-reg-lang}]
    Assume by contradiction that there exist some well-matched words accepted by two automata sharing their call and return symbols. We choose among these words a word $w\in \wm\pdwAlph$ with minimal depth, i.e., $w \in  \Lang{\mathcal{A}^J}\cap \Lang{\mathcal{A}^{J'}}$ for some distinct $J,J'\in\procAlpha$ with $f(J)=f(J')$. 
    We note $w=u_0c_1w_1r_1 \dots c_nw_nr_nu_n$ with $n\in \mathbb{N}$, $u_i\in\intAlpha^*$, $c_i\in \calAlpha$, $r_i\in \retAlpha$,  $w_i\in \wm\pdwAlph$ and $depth(w_i)<depth(w)$ for all $i$. By \autoref{prop:vra-run-to-nfa-trans}, since $w\in \Lang{\mathcal{A}^J}$ (resp.\ $w \in \Lang{\mathcal{A}^{J'}}$), we have  $u_0K_1\dots K_nu_n\in \regLang{\mathcal{A}^J}$ (resp.\ $u_0K_1'\dots K_n'u_n\in \regLang{\mathcal{A}^{J'}}$), with $K_i \in\procAlpha^{\langle c_i,r_i\rangle}$ (resp.\ $ K_i' \in\procAlpha^{\langle c_i,r_i\rangle}$) such that $w_i\in \Lang{\mathcal{A}^{K_i}}$ (resp.\ $w_i \in\Lang{\mathcal{A}^{K_i'}}$). For all $i\in[1,n]$, since $w_i\in \Lang{\mathcal{A}^{K_i}}\cap \Lang{\mathcal{A}^{K_i'}}$ and $depth(w_i)<depth(w)$, by minimality of depth of $w$, we get that $K_i=K_i'$ for all $i$. It follows that $u_0K_1 \dots K_nu_n=u_0K_1' \dots K_n'u_n \in \regLang{\mathcal{A}^J}\cap \regLang{\mathcal{A}^{J'}}$, which is a contradiction. \hfill   $\lrcorner$
\end{proof}
    \begin{property}\label{pro:complete-on-reg-lang}
    Let $\mathcal{A}$ be a VRA. 
    If for all $ c\in\calAlpha$ and $r\in\retAlpha$, we have $\bigcup_{J\in\procAlpha^{\langle c,r\rangle}} \regLang{\mathcal{A}^J} = (\intAlpha \cup \procAlpha)^*$, then $\bigcup_{J\in\procAlpha^{\langle c,r\rangle}} \Lang{\mathcal{A}^J} =\wm\pdwAlph$.
    \end{property}
    \begin{proof}[of \autoref{pro:complete-on-reg-lang}]
        We prove that, for all $ c\in \calAlpha$, $r \in  \retAlpha$ and $w\in \wm\pdwAlph$, we have $w \in \Lang{\mathcal{A}^J}$ for some $J\in \procAlpha^{\langle c,r\rangle}$. The proof is by induction on the depth of $w$.
    \begin{itemize}
        \item $\depth{w}=0$ ($w\in \intAlpha^*$): By hypothesis, $w\in \regLang{\mathcal{A}^J}$ for some $J\in \procAlpha^{\langle c,r\rangle}$. It follows that $w\in \Lang{\mathcal{A}^J}$ by \autoref{prop:vra-run-to-nfa-trans}.
        \item $\depth{w}>0$ ($w=u_0c_1w_1r_1\dots c_nw_nr_nu_n$, with $n\in \mathbb{N}_0$, $u_i\in\intAlpha^*$, $c_i\in\calAlpha$, $r_i\in \retAlpha$ and $w_i\in \wm\pdwAlph$): 
        By induction, as $\depth{w_i}<\depth{w}$ for all $i\in[1,n]$, there exists some $K_i \in \procAlpha^{\langle c_i,r_i\rangle}$ such that $w_i \in \Lang{\mathcal{A}^{K_i}}$. 
        By hypothesis again, we have $u_0K_1\dots K_nu_n \in \regLang{\mathcal{A}^J}$ for some $J\in \procAlpha^{\langle c,r\rangle}$. 
        By \autoref{prop:vra-run-to-nfa-trans}, we conclude that $w\in \Lang{\mathcal{A}^J}$. \hfill   $\lrcorner$ 
        \end{itemize}
    \end{proof}

    Since, by hypothesis, all the \regularLanguage s $\regLang{\mathcal{A}^J}$,  $J\in \procAlpha^{\langle c,r \rangle}$, are pairwise disjoint, by \autoref{pro:codet-on-reg-lang}, it follows that $\mathcal{A}$ is \codeterministic. Additionally, since, by hypothesis, the union of the \regularLanguage s is $(\intAlpha\cup \procAlpha)^*$, by \autoref{pro:complete-on-reg-lang}, the union of the \recursiveLanguage s is $\wm\pdwAlph$. Since all the FAs of $\mathcal{A}$ are complete, it follows that $\mathcal{A}$ is \complete. \qed
\end{proof}

\section{Proof of \autoref{th:codet-complete-vra}}\phantomsection\label{ax:detail-th-codet}

In this section, we provide a more complete proof of \autoref{th:codet-complete-vra} and  more details about the constructions. 

\codetcompvra*


\begin{proof}
    
Given $\mathcal{A}=\langle\pdwAlph,\procAlpha,\AutomataSet,\mathcal{A}^S\rangle$, we want to construct an equivalent VRA $\mathcal{B}=\langle\pdwAlph,\procAlpha',\AutomataSet',\mathcal{B}^{S}\rangle$ that is \codeterministic\ and \complete. Merging \Cref{def:codeterministic-vra,def:complete-vra}, $\mathcal{B}$ must respect the following conditions: each of its automata must be complete as FAs and, for all $\langle c,r\rangle\in \calAlpha\times \retAlpha$, the \recursiveLanguage s of all $\mathcal{B^J}$, with $\mathcal{J}\in \procAlpha'^{\langle c,r\rangle}$, must be a partition of $\wm\pdwAlph$. 

The main idea is the following. We define $\procAlpha'^{\langle c,r\rangle} = 2^{\procAlpha^{\langle c,r\rangle}}$, leading to the procedural alphabet of $\mathcal{B}$ equal to $\procAlpha' = \bigcup_{\langle c,r\rangle \in \calAlpha \times \retAlpha} \procAlpha'^{\langle c,r\rangle}$. Then, for each $\langle c,r \rangle$, we want to obtain, for all $\mathcal{J}\in\procAlpha'^{\langle c,r\rangle}$, Equality~\eqref{eq:lang-B^J} which we recall here:
\begin{equation}\tag{\ref{eq:lang-B^J}}
    \Lang{\mathcal{B}^\mathcal{J}} =  \bigcap_{J \in \mathcal{J}} \Lang{\mathcal{A}^J} ~\setminus \bigcup_{J\in \overline{\mathcal{J}}} \Lang{\mathcal{A}^{J}}. 
\end{equation} 
Recall that, when $\mathcal{J} = \varnothing$, $\bigcap_{J \in \mathcal{J}} \Lang{\mathcal{A}^J}=\wm\pdwAlph$. In this way, the \recursiveLanguage s of all $\mathcal{B^J}$, $\mathcal{J}\in \procAlpha'^{\langle c,r\rangle}$, form a partition of $\wm\pdwAlph$ as depicted in  \autoref{fig:venn-codet-complete}.
Note that for each $\langle c,r\rangle$, the set $\varnothing$ belongs to $\procAlpha'^{\langle c,r\rangle}$, each time corresponding to a distinct automaton.

    Before detailing the construction of each $\mathcal{B}^{\mathcal{J}}$, we transform all $\mathcal{A}^J\in \AutomataSet\cup \{\mathcal{A}^S\}$ into complete DFAs over the alphabet $\intAlpha\cup \procAlpha'$ as follows. 
    We replace each \proceduralTransition\ $q\xrightarrow{J}p\in\delta_\mathcal{A}$ by the transitions $q\xrightarrow{\mathcal{J}} p$, for all $\mathcal{J} \in \procAlpha'$ such that $\mathcal{J} \ni J$, and we then apply the subset construction to get complete DFAs~\cite{HU79}. This first step helps later to obtain~\eqref{eq:lang-B^J}.
    

    \begin{construction}\label{const:complete-dfa-A^J'} 
    For each FA $\mathcal{A}^J=\langle\intAlpha\cup \procAlpha, Q^J,I^J,F^J,\delta^J\rangle\in \AutomataSet\cup \{\mathcal{A}^S\}$, we construct the complete DFA $\mathcal{A}'^{J}=\langle\intAlpha\cup \procAlpha', Q'^J,I'^J,F'^J,\delta'^J\rangle$ with:
    \begin{itemize}
        \item $Q'^J= 2^{Q^J}$; $I'^J= \{I^J\}$; $F'^J=\{P\subseteq Q^J \mid P\cap F^J\neq \varnothing\}$;
        \item $\delta'^J$ is the  transition function defined by: Let $P \subseteq Q^J$,
        \begin{itemize}
            \item for $a\in \intAlpha$: $P\xrightarrow{a} \{p\in Q^J \mid \exists q \in P, q\xrightarrow{a}p \in \delta^J\}$; 
            \item for $\mathcal{J}\in \procAlpha'$: $P\xrightarrow{\mathcal{J}} \{p\in Q^J \mid \exists q \in P, \exists K\in \mathcal{J}, q\xrightarrow{K}p \in \delta^J\}$. 
        \end{itemize}
    \end{itemize}
    \end{construction}
    With this construction, we can state that the \regularLanguage\ of $\mathcal{A}'^J$ is equal to the \regularLanguage\ of $\mathcal{A}^J$, up to the replacement of the \proceduralSymbol s appearing in the accepted words.

    \langAprime*
    \begin{proof}[of \autoref{pro:lang-A'^K}] 
        As applying the subset construction to any FA leads to a DFA accepting the same regular language~\cite{HU79}, it is enough to prove \autoref{pro:lang-A'^K} when we replace the \proceduralTransition s $q\xrightarrow{J}p\in\delta_\mathcal{A}$ by the transitions $q\xrightarrow{\mathcal{J}} p$, for all $\mathcal{J} \subseteq \procAlpha'$ such that $J\in \mathcal{J}$.
        

        \begin{itemize}
            \item[$\Rightarrow$] In the accepting \regularRun\ on $u_0\mathcal{J}_1\dots \mathcal{J}_nu_n$, each transition $q_i\xrightarrow{\mathcal{J}_i}p_i$ comes from a transition $q_i\xrightarrow{J_i}p_i$ for some $J_i$ such that $J_i\in \mathcal{J}_i$. Therefore, we deduce the accepting \regularRun\ on $u_0J_1\dots J_nu_n$.
            \item[$\Leftarrow$] In the accepting \regularRun\ on $u_0J_1\dots J_nu_n$, each transition $q_i\xrightarrow{{J}_i}p_i$ implies the existence of the transitions $q_i\xrightarrow{\mathcal{J}_i}p_i$ for all $\mathcal{J}_i$ containing $J_i$. Therefore, for all $i\in[1,n]$ and $\mathcal{J}_i\ni J_i$, we deduce the accepting regular run on $u_0\mathcal{J}_1\dots \mathcal{J}_nu_n$. \hfill   $\lrcorner$ 
        \end{itemize}

    \end{proof}

    Since all $\mathcal{A}'^J$, $J \in \procAlpha \cup \{S\}$,  are complete DFAs, they are closed under Boolean operations with well-known constructions~\cite{HU79}. 
    For each $\langle c,r\rangle$, we can thus construct an automaton $\mathcal{B^J}$, $\mathcal{J} \in \procAlpha'^{\langle c,r\rangle}$,  such that its \regularLanguage\ respects~\eqref{eq:reg-lang-B^J} that we recall here. Notice that it has a form similar to \eqref{eq:lang-B^J}.
    \begin{equation}\tag{\ref{eq:reg-lang-B^J}}
        \regLang{\mathcal{B^J}}=\bigcap_{J\in \mathcal{J}} \regLang{\mathcal{A}'^J} \setminus \bigcup_{J \in \overline{\mathcal{J}}} \regLang{\mathcal{A}'^{J}}=\bigcap_{J\in \mathcal{J}} \regLang{\mathcal{A}'^J} \cap \bigcap_{J \in \overline{\mathcal{J}}} \overline{\regLang{\mathcal{A}'^{J}}}.
    \end{equation}
    The formal construction of $\mathcal{B}^{\mathcal{J}}$ is described in the following, with  appropriate Cartesian product and definition of the final states~\cite{HU79,sipser1996introduction}, in a way that $\mathcal{B^J}$ accepts the \regularLanguage\ of \eqref{eq:reg-lang-B^J}. 
    
    \begin{construction}\label{const:B^J} Let $\langle c,r\rangle\in \calAlpha\times \retAlpha$, $\mathcal{A}'^{J}=\langle\intAlpha\cup \procAlpha', Q'^{J},I'^{J},F'^{J},\delta'^{J}\rangle$ obtained with \autoref{const:complete-dfa-A^J'} for all $J\in\procAlpha^{\langle c,r\rangle}$. Given $\mathcal{J}\in \procAlpha'^{\langle c,r\rangle}$, we construct $\mathcal{B}^{\mathcal{J}}=\langle\intAlpha\cup \procAlpha', Q^{\mathcal{J}},I^{\mathcal{J}},F^{\mathcal{J}},\delta^{\mathcal{J}}\rangle$ with:
    \begin{itemize}
        \item $Q^{\mathcal{J}}=\bigtimes_{J\in \procAlpha^{\langle c,r\rangle}} Q'^{J}$;
        \item $I^{\mathcal{J}}= \bigtimes_{J\in \procAlpha^{\langle c,r\rangle}} I'^{J}$;
        \item $\delta^{\mathcal{J}}=\bigtimes_{J\in \procAlpha^{\langle c,r\rangle}} \delta'^{J}$;
        \item  $F^{\mathcal{J}} = \bigtimes_{J\in \procAlpha^{\langle c,r\rangle}} G^{J}$, where $G^J=F'^J$ when $J\in \mathcal{J}$, and $G^J= \overline{F'^J}$ otherwise. 
    \end{itemize}
    \end{construction}
    Since \autoref{const:B^J} is obtained through Cartesian products of FAs, we suppose that it is unnecessary to prove that $\mathcal{B^J}$ accepts the \regularLanguage\ described in \eqref{eq:reg-lang-B^J}, and we refer to \cite{HU79,sipser1996introduction} for more details.

    Finally, 
    we construct the required VRA $\mathcal{B}=\langle\pdwAlph,\procAlpha',\AutomataSet',\mathcal{B}^S\rangle$ such that 
    $\AutomataSet'=\{\mathcal{B}^{\mathcal{J}}\mid \mathcal{J}\in \procAlpha'\}$ where each $\mathcal{B}^{\mathcal{J}}$ is obtained with \autoref{const:B^J}, and $\mathcal{B}^S=\mathcal{A}'^S$ is obtained with \autoref{const:complete-dfa-A^J'}. 
    To gain more insight about the constructions, we refer to \autoref{ex:codet-complete-vra} below. 
    

    Let us prove that $\mathcal{B}$ is \codeterministic\ and \complete. For all $\langle c,r\rangle\in \calAlpha\times \retAlpha$, according to \eqref{eq:reg-lang-B^J}, the \regularLanguage s of all  DFAs $\mathcal{B^J}$, with $\mathcal{J}\in \procAlpha'^{\langle c,r\rangle}$, form a partition of $(\intAlpha\cup \procAlpha)^*$. By \autoref{prop:comp-codet-on-reg-lang}, since all automata of $\mathcal{B}$ are complete DFAs, it follows that $\mathcal{B}$ is \codeterministic\ and \complete. 

    We now prove that $\mathcal{B}$ accepts the same language as $\mathcal{A}$. We first prove the correctness of the \recursiveLanguage s $\Lang{\mathcal{B^J}}$ as exposed in \eqref{eq:lang-B^J}, which is a consequence of $\mathcal{B}$ being \codeterministic\ and \complete\ and of \autoref{pro:lang-B^J} that we now recall (see also \autoref{fig:venn-codet-complete}).

    \langBJ*

        \begin{proof}[of \autoref{pro:lang-B^J}] 
            The proof is by induction on the depth of $w\in\wm\pdwAlph$ that $w\in \Lang{\mathcal{A}^J} \Leftrightarrow \exists \mathcal{J}\in \procAlpha', \mathcal{J}\ni J: w \in  \Lang{\mathcal{B^J}}$.

            \begin{itemize}
                \item $depth(w)=0$ ($w\in\intAlpha^*$):
                By \autoref{pro:lang-A'^K}, we have $w\in \regLang{\mathcal{A}'^J} \Leftrightarrow w\in \regLang{\mathcal{A}^J}$.
                By \eqref{eq:reg-lang-B^J},
                since all $\regLang{\mathcal{B^J}}$, with $\mathcal{J}\ni J$, form a partition of $\regLang{\mathcal{A}'^J}$, we have $w\in \regLang{\mathcal{A}'^J} \Leftrightarrow \exists \mathcal{J}\ni J: w \in  \regLang{\mathcal{B^J}}$.
                As $w\in \intAlpha^*$, by \autoref{prop:vra-run-to-nfa-trans}, it follows that the required property holds. 
                \item $\depth{w} >0$ ($w=u_0c_1w_1r_1\dots c_nw_nr_nu_n$, with $n\in \mathbb{N}_0$, $u_i\in\intAlpha^*$, $c_i\in\calAlpha$, $r_i\in \retAlpha$ and $w_i\in \wm\pdwAlph$ for all $i$):
                \begin{itemize}
                    \item[$\Rightarrow$] Since $w\in \Lang{\mathcal{A}^J}$, by \autoref{prop:vra-run-to-nfa-trans}, we have $u_0{K}_1\dots {K}_n u_n\in\regLang{\mathcal{A}^J}$ for some ${K}_i\in \procAlpha^{\langle c_i,r_i\rangle}$ such that $w_i\in \Lang{\mathcal{A}^{{K}_i}}$, for all $i\in[1,n]$. By induction, for all $i\in [1,n]$, there exists $\mathcal{K}_i\in \procAlpha'^{\langle c_i,r_i\rangle}$ such that $K_i\in \mathcal{K}_i$ and $w_i \in \Lang{\mathcal{B}^{\mathcal{K}_i}}$. As $\mathcal{A}'^J$ is a complete DFA, there exists a \regularRun\ on $u_0\mathcal{K}_1\dots \mathcal{K}_n u_n$, which is accepting since $u_0{K}_1\dots {K}_n u_n\in \regLang{\mathcal{A}^J}$ and $ K_i \in \mathcal{K}_i$ (by \autoref{pro:lang-A'^K}). By \eqref{eq:reg-lang-B^J} and since $u_0\mathcal{K}_1\dots \mathcal{K}_n u_n\in \regLang{\mathcal{A}'^J}$, there exists a $\mathcal{J}\in \procAlpha'$ such that $u_0\mathcal{K}_1\dots \mathcal{K}_n u_n\in \regLang{\mathcal{B^J}}$ and $J\in \mathcal{J}$. Since  $\mathcal{K}_i\in \procAlpha'^{\langle c_i,r_i\rangle}$ and $w_i \in \Lang{\mathcal{B}^{\mathcal{K}_i}}$ for all $i$, by \autoref{prop:vra-run-to-nfa-trans}, $w\in \Lang{\mathcal{B^J}}$.
                    
                    \item [$\Leftarrow$] 
                    The other implication is proved similarly.
                    \hfill   $\lrcorner$ 
                \end{itemize}
            \end{itemize}
        \end{proof}

    Finally, to show that $\mathcal{A}$ and $\mathcal{B}$ are equivalent,   we must prove that for all $w \in \wm\pdwAlph$, $w \in \Lang{\mathcal{A}^{S}} \Leftrightarrow w \in \Lang{\mathcal{B}^S}$.
    Suppose that $w = u_0c_1w_1r_1 \ldots c_nw_nr_nu_n \in \wm\pdwAlph$ with $n\in \mathbb{N}$, $u_i\in \intAlpha^*$, $c_i\in \calAlpha$, $r_i\in \retAlpha$, $w_i \in \wm\pdwAlph$ for all~$i$: 
    \begin{itemize}
        \item[$\Rightarrow$]
        If $w \in \Lang{\mathcal{A}^S}$, by \autoref{prop:vra-run-to-nfa-trans}, we have $u_0{J}_1\dots {J}_n u_n\in\regLang{\mathcal{A}^S}$ for some ${J}_i\in \procAlpha^{\langle c_i,r_i\rangle}$ such that $w_i\in \Lang{\mathcal{A}^{{J}_i}}$, for all $i\in[1,n]$. By \autoref{pro:lang-B^J}, for all $i$, as $w_i \in \Lang{\mathcal{A}^{J_i}}$, there exists $\mathcal{J}_i \ni J_i$ such that $w_i \in \Lang{\mathcal{B}^{\mathcal{J}_i}}$. Then by \autoref{pro:lang-A'^K}, we deduce that $u_0\mathcal{J}_1 \ldots \mathcal{J}_nu_n \in \regLang{\mathcal{A}'^{S}}$. By \autoref{prop:vra-run-to-nfa-trans}, it follows that $w\in \Lang{\mathcal{A}'^{S}}= \Lang{\mathcal{B}^S}$.
        
        \item [$\Leftarrow$]
        The other implication is proved similarly.
    \end{itemize}

\medskip
To complete the proof, it remains to study the size of $\mathcal{B}$. The { number of states}  of $\mathcal{B}^S$ is equal to $2^{|Q^S|}$. 
For each $\langle c,r\rangle\in \calAlpha\times \retAlpha$, there are $2^{|\procAlpha^{\langle c,r\rangle}|}$ automata $\mathcal{B}^\mathcal{J}$, each with a {number of states}  $\prod_{J\in \procAlpha ^{\langle c,r\rangle}} 2 ^{|Q^J|} = 2^{\sum|Q^J|}$. Hence,
\begin{center}
    $|\vraStates{B}| = 2 ^{|Q^S|} + \sum_{\langle c,r\rangle \in \calAlpha \times \retAlpha}2^{|\procAlpha^{\langle c,r\rangle}|}\cdot 2^{\sum_{J \in \procAlpha^{\langle c,r\rangle}}|Q^J|} = 2^{\mathcal{O}({|\vraStates{A}|})}$.
\end{center} 
Since the number of transitions $|\vraTrans{B}|$ of $\mathcal{B}$ is in $\mathcal{O}(|\vraStates{B}|^2\cdot |\procAlpha'|)= 2^{\mathcal{O}({|\mathcal{A}|})}$ ($|\intAlpha|$ is supposed constant and $|\procAlpha'|\leq |\vraStates{B}|$), we conclude that $|\mathcal{B}|= 2^{\mathcal{O}(|\mathcal{A}|)}$.
\qed 
\end{proof}

 \begin{figure}[t]
        \centering
        \begin{subfigure}{\textwidth}
            \centering
            \begin{tikzpicture}[x=0.75pt,y=0.75pt,yscale=-1,xscale=1]

\draw    (56.13,40.44) -- (74.9,40.89) ;
\draw [shift={(76.9,40.93)}, rotate = 181.38] \arrow    ;
\draw    (104.9,40.93) -- (163.61,40.93) ;
\draw [shift={(165.61,40.93)}, rotate = 180] \arrow    ;
\draw    (90.9,54.93) -- (163.63,102.78) ;
\draw [shift={(165.3,103.88)}, rotate = 213.34] \arrow    ;
\draw    (179.3,54.83) -- (179.3,87.88) ;
\draw [shift={(179.3,89.88)}, rotate = 270] \arrow    ;
\draw    (187.41,114.84) .. controls (197.4,138.63) and (161.51,138.63) .. (171.56,118.1) ;
\draw [shift={(172.42,116.49)}, rotate = 120.01] \arrow    ;
\draw   (76.9,40.93) .. controls (76.9,33.2) and (83.17,26.93) .. (90.9,26.93) .. controls (98.63,26.93) and (104.9,33.2) .. (104.9,40.93) .. controls (104.9,48.67) and (98.63,54.93) .. (90.9,54.93) .. controls (83.17,54.93) and (76.9,48.67) .. (76.9,40.93) -- cycle ;
\draw   (165.3,40.83) .. controls (165.3,33.1) and (171.57,26.83) .. (179.3,26.83) .. controls (187.03,26.83) and (193.3,33.1) .. (193.3,40.83) .. controls (193.3,48.57) and (187.03,54.83) .. (179.3,54.83) .. controls (171.57,54.83) and (165.3,48.57) .. (165.3,40.83) -- cycle ;
\draw   (167.49,40.83) .. controls (167.49,34.31) and (172.78,29.02) .. (179.3,29.02) .. controls (185.82,29.02) and (191.11,34.31) .. (191.11,40.83) .. controls (191.11,47.36) and (185.82,52.65) .. (179.3,52.65) .. controls (172.78,52.65) and (167.49,47.36) .. (167.49,40.83) -- cycle ;
\draw   (165.3,103.88) .. controls (165.3,96.15) and (171.57,89.88) .. (179.3,89.88) .. controls (187.03,89.88) and (193.3,96.15) .. (193.3,103.88) .. controls (193.3,111.61) and (187.03,117.88) .. (179.3,117.88) .. controls (171.57,117.88) and (165.3,111.61) .. (165.3,103.88) -- cycle ;

\draw (29.83,29.67) node [anchor=north west][inner sep=0.75pt]   [align=left] {$\displaystyle \mathcal{B}^{S}$:};
\draw (102.47,23) node [anchor=north west][inner sep=0.75pt]   [align=left] {$\displaystyle \{R\} ,\{R,T\}$};
\draw (101.7,64.1) node [anchor=north west][inner sep=0.75pt]  [rotate=-33.33] [align=left] {$\displaystyle a,\{T\} ,\{\}$};
\draw (84,38) node [anchor=north west][inner sep=0.75pt]  [font=\small] [align=left] {$\displaystyle s_{0}$};
\draw (172.4,38) node [anchor=north west][inner sep=0.75pt]  [font=\small] [align=left] {$\displaystyle s_{1}$};
\draw (173,100) node [anchor=north west][inner sep=0.75pt]  [font=\small] [align=left] {$\displaystyle \perp $};
\draw (147,133.87) node [anchor=north west][inner sep=0.75pt]   [align=left] {$\displaystyle \Sigma _{int} \cup \Sigma _{proc}$};
\draw (178.6,54.87) node [anchor=north west][inner sep=0.75pt]   [align=left] {$\displaystyle  \begin{array}{l}
\Sigma _{int} \cup \\
\Sigma _{proc}
\end{array}$};

\end{tikzpicture}
            \vspace{-2em}
            \caption{Starting automaton $\mathcal{B}^S$ obtained with \autoref{const:complete-dfa-A^J'}. }
            \label{fig:codet-complete-rest}
        \end{subfigure}
        
        \begin{subfigure}{\textwidth}
            \centering
            \begin{tikzpicture}[x=0.75pt,y=0.75pt,yscale=-1,xscale=1]

\draw   (220,194.19) .. controls (220,186.73) and (235.11,180.69) .. (253.75,180.69) .. controls (272.39,180.69) and (287.5,186.73) .. (287.5,194.19) .. controls (287.5,201.64) and (272.39,207.69) .. (253.75,207.69) .. controls (235.11,207.69) and (220,201.64) .. (220,194.19) -- cycle ;
\draw   (91,128.19) .. controls (91,120.73) and (106.11,114.69) .. (124.75,114.69) .. controls (143.39,114.69) and (158.5,120.73) .. (158.5,128.19) .. controls (158.5,135.64) and (143.39,141.69) .. (124.75,141.69) .. controls (106.11,141.69) and (91,135.64) .. (91,128.19) -- cycle ;
\draw   (93,128.19) .. controls (93,121.84) and (107.21,116.69) .. (124.75,116.69) .. controls (142.29,116.69) and (156.5,121.84) .. (156.5,128.19) .. controls (156.5,134.54) and (142.29,139.69) .. (124.75,139.69) .. controls (107.21,139.69) and (93,134.54) .. (93,128.19) -- cycle ;
\draw   (91,193.69) .. controls (91,186.23) and (106.11,180.19) .. (124.75,180.19) .. controls (143.39,180.19) and (158.5,186.23) .. (158.5,193.69) .. controls (158.5,201.14) and (143.39,207.19) .. (124.75,207.19) .. controls (106.11,207.19) and (91,201.14) .. (91,193.69) -- cycle ;
\draw   (90.5,259.19) .. controls (90.5,251.73) and (105.61,245.69) .. (124.25,245.69) .. controls (142.89,245.69) and (158,251.73) .. (158,259.19) .. controls (158,266.64) and (142.89,272.69) .. (124.25,272.69) .. controls (105.61,272.69) and (90.5,266.64) .. (90.5,259.19) -- cycle ;
\draw   (92.5,259.19) .. controls (92.5,252.84) and (106.71,247.69) .. (124.25,247.69) .. controls (141.79,247.69) and (156,252.84) .. (156,259.19) .. controls (156,265.54) and (141.79,270.69) .. (124.25,270.69) .. controls (106.71,270.69) and (92.5,265.54) .. (92.5,259.19) -- cycle ;
\draw   (219.5,128.19) .. controls (219.5,120.73) and (234.61,114.69) .. (253.25,114.69) .. controls (271.89,114.69) and (287,120.73) .. (287,128.19) .. controls (287,135.64) and (271.89,141.69) .. (253.25,141.69) .. controls (234.61,141.69) and (219.5,135.64) .. (219.5,128.19) -- cycle ;
\draw   (348.5,128.19) .. controls (348.5,120.73) and (363.61,114.69) .. (382.25,114.69) .. controls (400.89,114.69) and (416,120.73) .. (416,128.19) .. controls (416,135.64) and (400.89,141.69) .. (382.25,141.69) .. controls (363.61,141.69) and (348.5,135.64) .. (348.5,128.19) -- cycle ;

\draw    (158.4,128.19) -- (217.5,128.19) ;
\draw [shift={(219.5,128.19)}, rotate = 180] \arrow;
\draw    (253.75,180.69) -- (253.28,143.69) ;
\draw [shift={(253.25,141.69)}, rotate = 89.27] \arrow;
\draw    (157,132.38) -- (231.35,182.56) ;
\draw [shift={(233,183.69)}, rotate = 214.18] \arrow;
\draw    (220,194.19) -- (161,194.19) ;
\draw [shift={(158.7,194.19)}, rotate = 360] \arrow;
\draw    (124.75,207.19) -- (124.75,244) ;
\draw [shift={(124.75,245.5)}, rotate = 270] \arrow;
\draw    (287.5,194.19) -- (380.5,142.66) ;
\draw [shift={(382.25,141.69)}, rotate = 151.01] \arrow;
\draw    (158,259.19) -- (382.5,259.19) -- (382.25,143.69) ;
\draw [shift={(382.25,141.69)}, rotate = 89.88] \arrow;
\draw    (227,186.19) -- (155.14,136.33) ;
\draw [shift={(153.5,135.19)}, rotate = 34.76] \arrow;
\draw    (155,199.69) -- (215,259.19) -- (382.5,259.19) -- (382.25,143.69) ;
\draw [shift={(382.25,141.69)}, rotate = 89.88] \arrow;
\draw    (67.02,127.91) -- (88.98,128.08) ;
\draw [shift={(90.98,128.09)}, rotate = 180.44] \arrow;
\draw    (124.75,141.69) -- (124.75,177.69) ;
\draw [shift={(124.75,180)}, rotate = 270] \arrow;

\draw    (132.25,272.51) .. controls (142.14,296.35) and (106.25,295.24) .. (116.38,274.72) ;
\draw [shift={(117.9,272.7)}, rotate = 120.26] \arrow;
\draw    (117.43,115.17) .. controls (107.57,91.33) and (143.45,92.46) .. (133.3,112.98) ;
\draw [shift={(132.43,114.8)}, rotate = 300.3] \arrow;
\draw    (144.4,117.19) .. controls (258.43,87.83) and (258.5,87.69) .. (380.4,114.28) ;
\draw [shift={(382.25,114.69)}, rotate = 192.31] \arrow;
\draw    (287,128.19) -- (346.7,127.99) ;
\draw [shift={(348.7,127.98)}, rotate = 179.81] \arrow;
\draw    (262.25,207.01) .. controls (272.14,230.85) and (236.25,229.74) .. (246.38,209.22) ;
\draw [shift={(247.26,207.6)}, rotate = 120.26] \arrow;
\draw    (409.22,120.4) .. controls (432.92,110.22) and (434.17,145.16) .. (413.61,135.24) ;
\draw [shift={(411.99,134.39)}, rotate = 29.54] \arrow;

\draw (221.9,187.17) node [anchor=north west][inner sep=0.75pt]  [font=\small] [align=left] {$\displaystyle \{r_{0} ,r_{1}\}\{t_{0}\}$};
\draw (100.4,121.17) node [anchor=north west][inner sep=0.75pt]  [font=\small] [align=left] {$\displaystyle \{r_{0}\}\{t_{0}\}$};
\draw (232.9,121.17) node [anchor=north west][inner sep=0.75pt]  [font=\small] [align=left] {$\displaystyle \{r_{1}\}\{\}$};
\draw (182.6,116.67) node [anchor=north west][inner sep=0.75pt]   [align=left] {$\displaystyle a$};
\draw (255.1,158.67) node [anchor=north west][inner sep=0.75pt]   [align=left] {$\displaystyle a$};
\draw (188.37,134.88) node [anchor=north west][inner sep=0.75pt]  [rotate=-34.06] [align=left] {$\displaystyle \{R,T\}$};
\draw (100.4,187.17) node [anchor=north west][inner sep=0.75pt]  [font=\small] [align=left] {$\displaystyle \{r_{1}\}\{t_{0}\}$};
\draw (235.1,224.67) node [anchor=north west][inner sep=0.75pt]   [align=left] {$\displaystyle \{R,T\}$};
\draw (100.26,148.68) node [anchor=north west][inner sep=0.75pt]  [rotate=-1.27] [align=left] {$\displaystyle \{R\}$};
\draw (106.4,251.67) node [anchor=north west][inner sep=0.75pt]  [font=\small] [align=left] {$\displaystyle \{\}\{t_{0}\}$};
\draw (65.89,208.8) node [anchor=north west][inner sep=0.75pt]  [rotate=-0.16] [align=left] {$\displaystyle  \begin{array}{{r}}
\{R\} ,\{T\} ,\\
\{R,T\}
\end{array}$};
\draw (29.5,118.67) node [anchor=north west][inner sep=0.75pt]   [align=left] {$\displaystyle \mathcal{B}^{\{T\}}$:};
\draw (367.9,120.17) node [anchor=north west][inner sep=0.75pt]  [font=\small] [align=left] {$\displaystyle \{\}\{\}$};
\draw (283,109) node [anchor=north west][inner sep=0.75pt]   [align=left] {$\displaystyle \intAlpha\cup\procAlpha'$};
\draw (71.6,291.17) node [anchor=north west][inner sep=0.75pt]   [align=left] {$\displaystyle \{R\} ,\{T\} ,\{R,T\}$};
\draw (308.54,163) node [anchor=north west][inner sep=0.75pt]  [rotate=-331.39] [align=left] {$\displaystyle \{\}$};
\draw (166.6,259.67) node [anchor=north west][inner sep=0.75pt]   [align=left] {$\displaystyle a,\{\}$};
\draw (254.91,78.97) node [anchor=north west][inner sep=0.75pt]  [rotate=-359.81] [align=left] {$\displaystyle \{\}$};
\draw (185.91,208.95) node [anchor=north west][inner sep=0.75pt]  [rotate=-43.31] [align=left] {$\displaystyle a,\{\}$};
\draw (113.15,82) node [anchor=north west][inner sep=0.75pt]   [align=left] {$\displaystyle \{T\}$};
\draw (178,155) node [anchor=north west][inner sep=0.75pt]  [rotate=-34.26] [align=left] {$\displaystyle \{T\}$};
\draw (182,180) node [anchor=north west][inner sep=0.75pt]   [align=left] {$\displaystyle \{R\}$};
\draw (406,102) node [anchor=north west][inner sep=0.75pt]   [align=left] {$\displaystyle \intAlpha\cup\procAlpha'$};

\end{tikzpicture}
            \vspace{-2em}
            \caption{Automaton $\mathcal{B}^{\{T\}}$ obtained with \autoref{const:B^J}.}
            \label{fig:codet-complete-const-b^cr}
        \end{subfigure}
        
        \caption{Illustration of \autoref{th:codet-complete-vra}.}
        \label{fig:ex-codet-complete-vra}
        \vspace{-1em}
    \end{figure}

\begin{example}\label{ex:codet-complete-vra}

    As an example of the constructions provided in \autoref{th:codet-complete-vra}, we construct a \codeterministic\ and \complete\ VRA accepting the same language as the VRA of \autoref{fig:ex-vra}. There is only one pair $(c,r)$, so $\procAlpha' = \{\{\},\{R\},\{T\},\{R,T\} \}$. The automaton $\mathcal{B}^{S}$ (in \autoref{fig:codet-complete-rest}), obtained as a copy of $\mathcal{A}'^S$, serves as an illustration of \autoref{const:complete-dfa-A^J'} (notice that no subset construction was here needed).
    In \autoref{fig:codet-complete-const-b^cr}, we show the automaton $\mathcal{B}^{\{T\}}$ obtained after applying \autoref{const:B^J}. The other automata $\mathcal{B}^{\mathcal{J}}$, $\mathcal{J} \in \{\{\},\{R\},\{R,T\}\}$, only differ from $\mathcal{B}^{\{T\}}$ by their final states. For instance, the unique final state of $\mathcal{B}^{\{\}}$ is $(\{\},\{\})$.
\end{example}

\section{\Codeterminism\ and Completeness Decision Problems}\phantomsection\label{ax:codet-comp-dp}

In this appendix, we study the complexity class of deciding whether a VRA is \codeterministic\ or \complete. Note that these complexities rely on closure properties and decision problems as discussed in \Cref{sec:ClosureProp}.

\subsection{\Codeterminism\ Decision Problem}

\begin{theorem}\label{th:codeterministic-dp}
    Deciding whether a VRA is \codeterministic\ is PTIME-complete.
\end{theorem}
\begin{proof} To check whether a VRA $\mathcal{A}$ is \codeterministic, we have to check whether for all distinct $J,J'\in \procAlpha$, $\Lang{\mathcal{A}^{J}}\cap \Lang{\mathcal{A}^{J'}} = \varnothing$. 
This is done in PTIME by using the intersection closure of VRAs (\autoref{th:closure}) and the emptiness decision problem for VRAs (\autoref{th:decision-problem}).


For the PTIME-hardness, we provide a logspace reduction from the emptiness decision problem. Intuitively, given a VRA $\mathcal{A}$, we construct a VRA $\mathcal{B}$ on a modified call alphabet in a way that checking whether $\mathcal{B}$ is \codeterministic\ reduces to checking the emptiness of only one intersection between $\mathcal{A}^S$ and a copy of it. 
    \begin{construction}\label{const:reduc-emptiness-to-codet}
        Let $\mathcal{A}=\langle\pdwAlph,\procAlpha,\AutomataSet,\mathcal{A}^S\rangle$ be a VRA with $\pdwAlph = \intAlpha\cup \calAlpha\cup \retAlpha$ its pushdown alphabet. We construct $\mathcal{B} = \langle\pdwAlph',\procAlpha',\AutomataSet',\mathcal{B}^{S}\rangle$ such that: 
        \begin{itemize}
            \item $\pdwAlph'= \intAlpha \cup \calAlpha' \cup \retAlpha'$ with $\calAlpha' =  \{c_J \mid  J \in \procAlpha 
            \}  \cup \{c'\}$ and $\retAlpha' = \retAlpha \cup  \{r'\}$, where all $c_J$, $J\in\procAlpha$, are pairwise distinct,  and $c'$ and $r'$ are distinct from the other call/return symbols.
            \item $\procAlpha' = \procAlpha\cup \{S_1,S_2\}$, with the \linkingFunction\ $f'$ defined as:  $f'(J)=\langle c_J,f_{\return}(J) \rangle $ for all $J\in \procAlpha$, and $f'(S_1)=f'(S_2)=\langle c',r' \rangle $.
            \item $\AutomataSet' = \AutomataSet\cup \{\mathcal{B}^{S_1},\mathcal{B}^{S_2}\}$, with $\mathcal{B}^{S_1}$ and $\mathcal{B}^{S_2}$ two copies of $\mathcal{A}^S$.
            \item $\mathcal{B}^{S}$ is the starting automaton, being any FA over $\intAlpha\cup \procAlpha'$.\footnote{Whatever the definition of $\mathcal{B}^{S}$, it does not impact the correctness of the reduction.}
        \end{itemize}
    \end{construction}
    
    By \autoref{def:codeterministic-vra} and by construction, $\mathcal{B}$ is \codeterministic\  iff $\Lang{\mathcal{B}^{S_1}}\cap \Lang{\mathcal{B}^{S_2}} = \varnothing$. Moreover, as $\mathcal{B}^{S_1}$ and $\mathcal{B}^{S_2}$ are copies of $\mathcal{A}^S$, $\mathcal{B}$ is \codeterministic\  iff $\Lang{\mathcal{B}^{S_1}} = \Lang{\mathcal{B}^{S_2}}= \varnothing$. Note that the construction of $\mathcal{B}$ can be done in logarithmic space. 
    It remains to show that $\Lang{\mathcal{A}} =  \varnothing$ iff $\mathcal{B}$ is \codeterministic, that is, iff $\Lang{\mathcal{B}^{S_1}} = \varnothing$. 

    \begin{property}\label{pro:codet-complete}
        For all $J \in \procAlpha$ (resp.\ for $J = S$), there exists $w \in \Lang{\mathcal{A}^J}$ iff there exists $w' \in \wm{\pdwAlph'}$ such that $w' \in \Lang{\mathcal{B}^J}$ (resp.\ $w' \in \Lang{\mathcal{B}^{S_1}}$).
    \end{property}
    \begin{proof}[of \autoref{pro:codet-complete}] We only prove the forward implication; for the converse, the argument is similar.
    We proceed by induction on depth of $w \in \wm\pdwAlph$.
    \begin{itemize}
        \item $depth(w) = 0$: As $\mathcal{B}^{J}$ is equal to $\mathcal{A}^{J}$ (resp.\ $\mathcal{B}^{S_1}$ is a copy of $\mathcal{A}^{S}$), we define $w' = w$ and get $w' \in \Lang{\mathcal{B}^{J}}$ (resp.\ $w' \in \Lang{\mathcal{B}^{S_1}}$).
        \item $depth(w) > 0$ ($w=u_0c_1w_1r_1\dots c_nw_nr_nu_n$, with $n\in \mathbb{N}_0$, $u_i\in\intAlpha^*$, $c_i\in\calAlpha$, $r_i\in \retAlpha$ and $w_i\in \wm\pdwAlph$): By \autoref{prop:vra-run-to-nfa-trans}, there exists a word $u_0K_1u_1 \dots K_nu_n \in \regLang{\mathcal{A}^J}=\regLang{\mathcal{B}^J}$ with $K_i \in \procAlpha^{\langle c_i,r_i\rangle}$ and $w_i\in \Lang{\mathcal{A}^{K_i}}$ for all $i$. 
        By induction hypothesis, for each $w_i$, there exists $w'_i \in \Lang{\mathcal{B}^{K_i}}$. As $f'(K_i) = (c_{K_i},r_i)$, by \autoref{prop:vra-run-to-nfa-trans}, we have $u_0c_{K_1}w'_1r_1 \dots c_{K_n}w'_nr_nu_n \in \Lang{\mathcal{B}^J}$ if $J \in \procAlpha$ (resp.\ $\Lang{\mathcal{B}^{S_1}}$ if $J = S$). \hfill   $\lrcorner$ 
\end{itemize}   
    \end{proof}
    With \autoref{pro:codet-complete}, we have that $\Lang{\mathcal{A}} = \Lang{\mathcal{A}^S} = \varnothing$ iff $\Lang{\mathcal{B}^{S_1}} =  \varnothing$, that is, iff $\mathcal{B}$ is \codeterministic.
    \qed
\end{proof}

\subsection{Completeness Decision Problem}
\begin{theorem}\label{th:complete-dp}
    Deciding whether a VRA is \complete\ is EXPTIME-complete.
\end{theorem}
\begin{proof}
    For EXPTIME-easiness, by \autoref{def:complete-vra}, we have to check whether each FA of $\mathcal{A}$ is complete (which can be done in polynomial time) and for all pairs $\langle c,r \rangle$, whether  $\bigcup_{J\in \procAlpha^{\langle c,r\rangle}} \Lang{\mathcal{A}^J} = \wm\pdwAlph$. The latter check can be done in exponential time, by using \Cref{th:closure,th:decision-problem}. 

    For EXPTIME-hardness, we provide a polynomial reduction from the universality decision problem for VRAs. Given a VRA $\mathcal{A}$, we must construct an automaton $\mathcal{B}$ which is \complete\ iff  $\mathcal{A}$ is universal. Intuitively, we modify $\mathcal{A}$ in order to achieve the sufficient and necessary condition: 
    \begin{itemize}  
        \item[$\Leftarrow$] First, we complete all FAs by adding a bin state and transitions to it. Then, for all pairs $\langle c,r\rangle$, we add to $\mathcal{B}$ an automaton linked to this pair which is a copy of $\mathcal{A}^S$. This ensures $\mathcal{B}$ to be \complete\ when $\mathcal{A}$ is universal.
        
        \item[$\Rightarrow$]
        We introduce a new call symbol $c'$. Then, we construct $\mathcal{B}$ such that, for every pair $\langle c', r \rangle$, there exists a unique automaton $\mathcal{B}^J$ with $f(J) = \langle c', r \rangle$. Moreover, this $\mathcal{B}^J$ is a copy of $\mathcal{A}^S$ with exactly the same recursive language. Hence, \(\mathcal{A}\) must be universal for $\mathcal{B}$ to be \complete.
    \end{itemize}


    \begin{construction}\label{const:reduction-universality-complete}
        Let $\mathcal{A}=\langle\pdwAlph,\procAlpha,\AutomataSet,\mathcal{A}^S\rangle$ be a VRA. We construct the VRA $\mathcal{B}=(\pdwAlph',\procAlpha',\AutomataSet',\mathcal{B}^{S})$ such that:
    \begin{itemize}
        \item $\pdwAlph' = \intAlpha \cup \calAlpha' \cup \retAlpha$, with $\calAlpha'=\calAlpha\cup\{c'\}$, and $c'$ distinct from symbols of $\calAlpha$.
        \item $\procAlpha'=\procAlpha\cup \mathcal{S}$, where $\mathcal{S}=\{S^{\langle c,r\rangle} \mid c\in\calAlpha', r\in \retAlpha\}$ and $f'(S^{\langle c,r\rangle})=\langle c,r\rangle$ for all $\langle c,r\rangle \in \calAlpha' \times \retAlpha$.
        \item $\AutomataSet' = \AutomataSet^{\procAlpha} \cup \AutomataSet^{\mathcal{S}}$ is the set of automata defined by: 
        \begin{itemize}
            \item $\AutomataSet^{\procAlpha} = \{\mathcal{B}^J \mid J \in \procAlpha\}$, where each $\mathcal{B}^J$ is a copy of $\mathcal{A}^J$, with an additional bin state and transitions in a way to get a complete FA.
            \item $\AutomataSet^{\mathcal{S}} = \{\mathcal{B}^{S^{\langle c,r\rangle}} \mid S^{\langle c,r\rangle} \in \mathcal{S}\}$, where each $\mathcal{B}^{S^{\langle c,r\rangle}}$ is a copy of $\mathcal{A}^S$, with an additional bin state and transitions in a way to get a complete FA.
        \end{itemize}
        \item $\mathcal{B}^{S}$ is the starting automaton, being any complete FA over $\intAlpha\cup \procAlpha'$.
    \end{itemize}
    \end{construction}
    
    \autoref{const:reduction-universality-complete} is done in polynomial time (as $\mathcal{O}(|\calAlpha|\times |\retAlpha|)$ copies of $\mathcal{A}^S$ are created).  By construction, it is clear that $\Lang{\mathcal{B}^J}=\Lang{\mathcal{A}^J}$ for all $J\in \procAlpha$. It is also straightforward to see that  $\Lang{\mathcal{B}^{S^{\langle c,r\rangle}}}=\Lang{\mathcal{A}^S}$ for all $S^{\langle c,r\rangle}\in\mathcal{S}$. 
    It remains to prove that $\mathcal{B}$ is \complete $\Leftrightarrow\Lang{\mathcal{A}}=\wm\pdwAlph$:
    \begin{itemize}
        \item[$\Leftarrow$] As each FA of $\mathcal{B}$ is complete by construction, the first condition of \autoref{def:complete-vra} is met. For the second condition, note that for all $\langle c,r\rangle\in\calAlpha'\times\retAlpha$, $\mathcal{B}^{S^{\langle c,r\rangle}}\in \bigcup_{J\in \procAlpha'^{\langle c,r\rangle}} \mathcal{B}^J$. Since $\Lang{\mathcal{B}^{S^{\langle c,r\rangle}}}=\Lang{\mathcal{A}^S}=\wm\pdwAlph$ (by hypothesis), the second condition is also met. 
        
        \item[$\Rightarrow$] By construction, for all $r\in \retAlpha$, we have $\procAlpha^{\langle c',r \rangle}=\{\mathcal{B}^{S^{\langle c',r \rangle}}\}$. By hypothesis, we have $\bigcup_{J\in \procAlpha^{\langle c',r \rangle}} \Lang{\mathcal{B}^J}=\Lang{\mathcal{B}^{S^{\langle c',r \rangle}}}=\wm\pdwAlph$. Since $\Lang{\mathcal{B}^{S^{\langle c',r \rangle}}}=\Lang{\mathcal{A}^S}$, we have that $\mathcal{A}$ is universal. \qed
    \end{itemize}
\end{proof}

\section{Closure Properties of VRAs}\phantomsection\label{ax:closure-properties}

\Cref{tab:vra-vpa} and \autoref{th:closure} present the state complexity to construct the operation closures of VRAs. In this section, we provide the proof and construction for each operation.

\closure*

We assume, without loss of generality, that the input alphabets of $\mathcal{A}_1$ and $\mathcal{A}_2$ are the same, and that their \proceduralAlphabet s are disjoint.



\subsection{Concatenation Closure}

    Intuitively, to construct $\mathcal{B}$ such that $\Lang{\mathcal{B}}=L_1\cdot L_2$, we proceed as follows: the FAs of $\mathcal{B}$ are copies of the ones of $\mathcal{A}_1$ and $\mathcal{A}_2$, and its  starting automaton accepts the concatenation of the \regularLanguage s $\regLang{\mathcal{A}_1^S}$ and $\regLang{\mathcal{A}_2^S}$. 
    
    \begin{construction}
    \label{const:concat}
        Let $\mathcal{A}_1=\langle\pdwAlph,\Sigma_{\proc 1},\AutomataSet_1,\mathcal{A}_1^S\rangle$ and $\mathcal{A}_2=\langle\pdwAlph,\Sigma_{\proc 2},\AutomataSet_2,\mathcal{A}_2^S\rangle$ be two VRAs. We construct the VRA $\mathcal{B}=\langle\pdwAlph,\Sigma_{\proc 1}\cup \Sigma_{\proc 2}, \AutomataSet_1 \cup \AutomataSet_2, \mathcal{B}^S\rangle$, with $\mathcal{B}^S$ an FA accepting the \regularLanguage\ $\regLang{\mathcal{B}^S}=\regLang{\mathcal{A}_1^S}\cdot \regLang{\mathcal{A}_2^S}$ \cite{sipser1996introduction}.
    \end{construction}

    \begin{property}\label{pro:concat-closure}
        $\mathcal{B}$ accepts $\Lang{\mathcal{B}}=\Lang{\mathcal{A}_1}\cdot \Lang{\mathcal{A}_2}$. 
    \end{property}
    \begin{proof}[of \autoref{pro:concat-closure}]
        Knowing that $\regLang{\mathcal{B}^S}=\regLang{\mathcal{A}_1^S}\cdot \regLang{\mathcal{A}_2^S}$, it is straightforward to prove that, for all $w\in \wm\pdwAlph$, $w\in \Lang{\mathcal{A}_1}\cdot \Lang{\mathcal{A}_2}  \Leftrightarrow w\in \Lang{\mathcal{B}}$. 
        \begin{itemize}
            \item[$\Rightarrow$] Using \autoref{prop:vra-run-to-nfa-trans}, since $w\in  \Lang{\mathcal{A}^S_1}\cdot \Lang{\mathcal{A}^S_2} $, there exists word in the \regularLanguage\ $\regLang{\mathcal{A}^S_1}\cdot \regLang{\mathcal{A}^S_2}=\regLang{\mathcal{B}^S}$ obtained by replacing each factor $cw'r$ (with $c\in \calAlpha$, $r\in \retAlpha$, $w'\in \wm\pdwAlph$) of $w$ by a procedural symbol of $\Sigma_{\proc 1}$ or $\Sigma_{\proc 2}$. Since all FAs of $\mathcal{B}$ are copies of those of $\mathcal{A}_1$ and $\mathcal{A}_2$, by \autoref{prop:vra-run-to-nfa-trans}, it follows that $w\in \Lang{\mathcal{B^S}}$.
            
            \item[$\Leftarrow$] The other implication is proved similarly. \hfill   $\lrcorner$
        \end{itemize}
        
            
    \end{proof}    


    Since $\mathcal{B}^S$ is constructed such that $\regLang{\mathcal{B}^S}=\regLang{\mathcal{A}_1^S}\cdot \regLang{\mathcal{A}_2^S}$, we have $|\mathcal{B}^S|= \mathcal{O}(|\mathcal{A}^S_1|+|\mathcal{A}^S_2|)$  \cite{sipser1996introduction}. Clearly, as all FAs of $\mathcal{B}$ are copies of those of $\mathcal{A}_1$ and $\mathcal{A}_2$, it follows $|\mathcal{B}|=\mathcal{O}(|\mathcal{A}_1|+|\mathcal{A}_2|)$.  

\subsection{Kleene-$*$ Closure}


    To obtain a VRA accepting the language $L_1^*$, we simply apply the Kleene-$*$ construction of FAs on the starting automaton of $\mathcal{A}_1$. 

    \begin{construction}\label{const:kleene}
        Let $\mathcal{A}_1=\langle\pdwAlph,\Sigma_{\proc 1},\AutomataSet_1,\mathcal{A}_1^S\rangle$ be a VRA. We construct the VRA $\mathcal{B}=\langle\pdwAlph,\Sigma_{\proc 1}, \AutomataSet_1, \mathcal{B}^S\rangle$, with $\mathcal{B}^S$ an FA accepting the \regularLanguage\ $\regLang{\mathcal{B}^S}=\regLang{\mathcal{A}_1^S}^*$ \cite{sipser1996introduction}.
    \end{construction}

    Since $\mathcal{B}^S$ is constructed such that $\regLang{\mathcal{B}^S}=\regLang{\mathcal{A}_1^S}^*$, we have $|\mathcal{B}^S|= \mathcal{O}(|\mathcal{A}^S_1|)$ \cite{sipser1996introduction}, and thus  $|\mathcal{B}|=\mathcal{O}(|\mathcal{A}_1|)$. It is straightforward to prove that $\Lang{\mathcal{B}}=\Lang{\mathcal{A}_1}^*$ with an argument similar to the proof of \autoref{pro:concat-closure}.

\subsection{Union Closure}

    To construct $\mathcal{B}$ such that $\Lang{\mathcal{B}}=L_1\cup L_2$, we proceed as for the concatenation. 
    
    \begin{construction}
        Let $\mathcal{A}_1=\langle\pdwAlph,\Sigma_{\proc 1},\AutomataSet_1,\mathcal{A}_1^S\rangle$ and $\mathcal{A}_2=\langle\pdwAlph,\Sigma_{\proc 2},\AutomataSet_2,\mathcal{A}_2^S\rangle$ be two VRAs. We construct the VRA $\mathcal{B}=\langle\pdwAlph,\Sigma_{\proc 1}\cup \Sigma_{\proc 2}, \AutomataSet_1 \cup \AutomataSet_2, \mathcal{B}^S\rangle$, with $\mathcal{B}^S$ an FA accepting the \regularLanguage\ $\regLang{\mathcal{B}^S}=\regLang{\mathcal{A}_1^S}\cup \regLang{\mathcal{A}_2^S}$ \cite{sipser1996introduction}.
    \end{construction}

    Since $\mathcal{B}^S$ is constructed such that $\regLang{\mathcal{B}^S}=\regLang{\mathcal{A}_1^S}\cup \regLang{\mathcal{A}_2^S}$, we have $|\mathcal{B}^S|= \mathcal{O}(|\mathcal{A}^S_1|+|\mathcal{A}^S_2|)$  \cite{sipser1996introduction}. It follows $|\mathcal{B}|=\mathcal{O}(|\mathcal{A}_1|+|\mathcal{A}_2|)$. It is straightforward to prove that $\Lang{\mathcal{B}}=\Lang{\mathcal{A}_1}\cup \Lang{\mathcal{A}_2}$ with an argument similar to the proof of \autoref{pro:concat-closure}.

\subsection{Intersection Closure}
\label{subsec:intersection}
    For the intersection closure, to construct correctly the \proceduralTransition s, we need to compute the intersection of each pair of \recursiveLanguage s $\Lang{\mathcal{A}^{J_1}_1}$ and $\Lang{\mathcal{A}^{J_2}_2}$ of $\mathcal{A}_1$ and $\mathcal{A}_2$.   
    Intuitively, we define the new \proceduralSymbol s $\langle J_1,J_2 \rangle\in \Sigma_{\proc 1}\times \Sigma_{\proc 2}$, replace all transitions $q_1\xrightarrow{J_1}p_1 \in \vraTrans{A_\text{1}}$ and $q_2\xrightarrow{J_2}p_2 \in \vraTrans{A_\text{2}}$, respectively by $q_1\xrightarrow{\langle J_1,J_2 \rangle}p_1$ and $q_2\xrightarrow{\langle J_1,J_2 \rangle}p_2$, and we finally  construct $\mathcal{B}^{\langle J_1,J_2 \rangle}$ equal to the Cartesian product of $\mathcal{A}^{J_1}_1$ and $\mathcal{A}^{J_2}_2$. This will ensure that $\Lang{\mathcal{B}^{\langle J_1,J_2 \rangle}}=\Lang{\mathcal{A}^{J_1}_1}\cap \Lang{\mathcal{A}^{J_2}_2}$. Let us detail the construction.

    \begin{construction}\label{const:intersection-vra}
        Let $\mathcal{A}_1 = \langle\pdwAlph,\Sigma_{\proc 1},\AutomataSet_1,\mathcal{A}_1^S\rangle$ and $\mathcal{A}_2=\langle\pdwAlph,\Sigma_{\proc 2},\AutomataSet_2,\mathcal{A}_2^S\rangle$ be two VRAs such that $\mathcal{A}^{S}_i=\langle\intAlpha \cup \Sigma_{\proc i},Q^S_i,I^S_i,F^S_i,\delta^S_i\rangle$, for $i=1,2$. We construct  $\mathcal{B}=\langle\pdwAlph, \procAlpha, \AutomataSet, \mathcal{B}^S\rangle$ 
        as follows:
        \begin{itemize}
        \item $\procAlpha = \bigcup_{\langle c,r\rangle\in \calAlpha\times \retAlpha} \procAlpha^{\langle c,r\rangle}$, with $\procAlpha^{\langle c,r\rangle}=\Sigma_{\proc 1}^{\langle c,r\rangle} \times \Sigma_{\proc 2}^{\langle c,r\rangle}$ for all $\langle c,r\rangle\in \calAlpha\times \retAlpha$.
        \item 
        $\mathcal{B}^{S}=\langle\intAlpha \cup \procAlpha,Q^S_1 \times Q^S_2,I^S_1\times I^S_2,F^S_1\times F^S_2,\delta^S\rangle$, where $\delta^S$ is defined by: 
        \begin{itemize}
            \item for all $a\in \intAlpha$: $\langle q_1,q_2 \rangle\xrightarrow{a}\langle p_1,p_2 \rangle \in \delta^S \Leftrightarrow (q_1\xrightarrow{a}p_1 \in \delta^S_1 \wedge q_2\xrightarrow{a}p_2 \in \delta^S_2)$,
            \item for all $\langle J_1,J_2 \rangle\in \procAlpha$: $\langle q_1,q_2 \rangle\xrightarrow{\langle J_1,J_2 \rangle}\langle p_1,p_2 \rangle \in \delta^S \Leftrightarrow (q_1\xrightarrow{J_1}p_1 \in \delta^S_1 \wedge q_2\xrightarrow{J_2}p_2 \in \delta^S_2)$.
        \end{itemize}
        \item $\AutomataSet = \{\mathcal{B}^{\langle J_1,J_2 \rangle}\mid \langle J_1,J_2 \rangle \in \procAlpha\}$, where each automaton $\mathcal{B}^{\langle J_1,J_2 \rangle}$ is built similarly to $\mathcal{B}^{S}$, but with $\mathcal{A}_1^{J_1}$ (resp.\ $\mathcal{A}_2^{J_2}$) instead of $\mathcal{A}^S_1$ (resp.\ $\mathcal{A}^S_2$).
    \end{itemize}
    \end{construction}

\begin{property}\label{pro:correctness-intersection}
    For all $w\in \wm\pdwAlph$, $\langle q_1,q_2 \rangle,\langle p_1,p_2 \rangle\in \vraStates{B}$: 
    \[\left\langle \langle q_1,q_2 \rangle,\varepsilon\right \rangle \xrightarrow{w}\left\langle\langle p_1,p_2 \rangle ,\varepsilon\right \rangle\in \runsA{\mathcal{B}} \iff \left\{\begin{array}{c}
             \langle q_1,\varepsilon \rangle \xrightarrow{w}\langle p_1 ,\varepsilon\rangle \in \runsA{\mathcal{A}_1} \\
              \langle q_2,\varepsilon \rangle \xrightarrow{w}\langle p_2 ,\varepsilon\rangle \in \runsA{\mathcal{A}_2}
        \end{array}\right. .\]
    \end{property}
    \begin{proof}[of \autoref{pro:correctness-intersection}]
         First, since $w\in \wm\pdwAlph$, by \autoref{lem:stack-equivalence-vra}, $\langle q_1,q_2 \rangle$ and $\langle p_1,p_2 \rangle$ belong to the same automaton $\mathcal{B}^{\langle J_1,J_2 \rangle}\in \AutomataSet$ (resp.\ $\mathcal{B}^S$), and so do $q_1,p_1\in Q^{J_1}_1$ (resp.\ $Q_1^S$) and $q_2,p_2\in Q_2^{J_2}$ (resp.\ $Q_2^S$). In what follows, we assume that the states belong to the automaton $\mathcal{B}^{\langle J_1,J_2 \rangle}$, but the proof holds even if they belong to $\mathcal{B}^S$. 
     We prove the property by structural induction of well-matched words. 
    
    \begin{itemize}
        \item $w\in \intAlpha^*$: The property holds since $\mathcal{B}^{\langle J_1,J_2 \rangle}$ is the Cartesian product of $\mathcal{A}^{J_1}_1$ and $\mathcal{A}^{J_2}_2$.
        \item $w=cw'r$ (with $c\in\calAlpha$, $r\in \retAlpha$ and $w'\in \wm\pdwAlph$):
        \begin{itemize}
            \item[$\Rightarrow$] Thanks to \autoref{prop:vra-run-to-nfa-trans}, there exists $\langle q_1,q_2 \rangle\xrightarrow{\langle K_1,K_2 \rangle}\langle p_1,p_2 \rangle\in \delta^{\langle J_1,J_2 \rangle}$ with $\langle K_1,K_2 \rangle\in \procAlpha^{\langle c,r\rangle}$ (thus, $K_1 \in \Sigma_{\proc 1}^{\langle c,r\rangle}$ and $K_2\in \Sigma_{\proc 2}^{\langle c,r\rangle}$) and $w'\in \Lang{\mathcal{B}^{\langle K_1,K_2 \rangle}}$. Since this transition exists, there exist $q_1\xrightarrow{K_1}p_1\in \delta^{J_1}$ and $q_2\xrightarrow{K_2}p_2\in \delta^{J_2}$. Moreover, by structural induction, we deduce from the accepting \recursiveRun\ on $w'$ in $\mathcal{B}^{\langle K_1,K_2 \rangle}$ the accepting \recursiveRun s on $w'$ in $\mathcal{A}_1^{K_1}$ and $\mathcal{A}_2^{K_2}$. Since $q_1\xrightarrow{K_1}p_1\in \delta^{J_1}$ (resp.\ $q_2\xrightarrow{K_2}p_2\in \delta^{J_2}$) and $w'\in \Lang{\mathcal{A}_1^{K_1}}$ (resp.\ $w'\in \Lang{\mathcal{A}_2^{K_2}}$), it follows by \autoref{prop:vra-run-to-nfa-trans} that $\langle q_1,\varepsilon \rangle \xrightarrow{cw'r}\langle p_1 ,\varepsilon\rangle \in \runsA{\mathcal{A}_1}$ (resp.\ $\langle q_2,\varepsilon \rangle \xrightarrow{cw'r}\langle p_2 ,\varepsilon\rangle \in \runsA{\mathcal{A}_2}$). 
            \item[$\Leftarrow$] The other implication is proved similarly.
        \end{itemize}
        \item $w=w_1 w_2$ (with $w_1,w_2\in \wm\pdwAlph$): This is trivial by  induction. \hfill   $\lrcorner$
    \end{itemize}
    \end{proof}
From \autoref{pro:correctness-intersection}, it is easy to see that the existence of an accepting \recursiveRun\ on a word in $\mathcal{B}^S$ implies the existence of the accepting \recursiveRun s on this word in $\mathcal{A}^S_1$ and $\mathcal{A}^S_2$, and conversely, thus  proving that $\Lang{\mathcal{B}}=\Lang{\mathcal{A}_1}\cap \Lang{\mathcal{A}_2}$. Finally, it is clear that the size of $\mathcal{B}$ is $|\mathcal{B}| = \mathcal{O}(|\mathcal{A}_1|\cdot |\mathcal{A}_2|) $.
\subsection{Complementation Closure}

Lastly, for the complementation closure, if $\mathcal{A}_1 = \langle\pdwAlph,  \Sigma_{\proc 1}, \AutomataSet_1, \mathcal{A}_1^S\rangle$ is \codeterministic\ and \complete\ with
all its automata being complete DFAs,  we construct $\mathcal{B}=\langle\pdwAlph,  \Sigma_{\proc 1}, \AutomataSet_1, \mathcal{B}^S\rangle$, with $\mathcal{B}^{S}$ accepting the \regularLanguage\ $\regLang{\mathcal{B}^S}=\overline{\regLang{\mathcal{A}_1^S}}$ (i.e., final states of $\mathcal{B}^S$ are the non final states of $\mathcal{A}_1^S$ \cite{HU79}). If not, we first construct a \codeterministic\ and \complete\ VRA accepting $L_1$ using \autoref{th:codet-complete-vra} (note that the construction leads to a VRA with all its automata being DFAs), and then apply the previous construction. Thus, $|\mathcal{B}|=2^{\mathcal{O}(|\mathcal{A}_1|)}$ (or $|\mathcal{B}|=|\mathcal{A}_1|$ if $\mathcal{A}_1$ was already \codeterministic\ and \complete).

Let us prove that the construction of the complementation is correct, that is, for all $w\in\wm\pdwAlph$: $w\in \Lang{\mathcal{A}} \Leftrightarrow w\notin \Lang{\mathcal{B}}$. Since $w\in \wm\pdwAlph$, it can be written as $w=u_0c_1w_1r_1\dots c_nw_nr_nu_n$, with $n\in \mathbb{N}$, $u_i\in \intAlpha^*$, $c_i\in \calAlpha$, $r_i\in\retAlpha$ and $w_i\in\wm\pdwAlph$ for all $i$. 
Let $w'=u_0 K_1 \dots K_n u_n$, for some $K_i \in \procAlpha^{\langle c_i,r_i\rangle}$, $i \in [1,n]$, such that $w_i\in \Lang{\mathcal{A}^{K_i}}=\Lang{\mathcal{B}^{K_i}}$. Note that, since $\mathcal{A}$ and $\mathcal{B}$ are \codeterministic\ and \complete, for all $i\in [1,n]$, there exists exactly one $K_i \in \procAlpha^{\langle c_i,r_i\rangle}$ such that $w_i\in \Lang{\mathcal{A}^{K_i}}=\Lang{\mathcal{B}^{K_i}}$, i.e., {the way $w'$ is defined from $w$ is unique}. Therefore, by \autoref{prop:vra-run-to-nfa-trans}, we have that $w\in \Lang{\mathcal{A}} \Leftrightarrow w'\in \regLang{\mathcal{A}^S} \Leftrightarrow w'\notin \regLang{\mathcal{B}^S} \Leftrightarrow w\notin \Lang{\mathcal{B}}$.

\subsection{Closure Properties  of Codeterministic and Complete VRAs}

\Codeterministic\ and \complete\ VRAs have interesting properties that we may wish to preserve in closure constructions for VRAs, much like determinism for FAs \cite{yu1994state} and VPAs \cite{complexityVPA}.
In this section, we revisit each language-theoretic operation when the given VRAs $\mathcal{A}_1$ and $\mathcal{A}_2$ are \codeterministic\ and \complete, with all their automata being DFAs, and we want to construct the resulting VRA $\mathit{B}$ satisfying the same constraints. 

We begin with the concatenation closure. The construction is based on both \Cref{const:concat,const:intersection-vra}. 
First, to obtain a \codeterministic\ and \complete\ VRA $\mathcal{B}$, that is, with its \recursiveLanguage s forming a partition of $\wm\pdwAlph$,  we compute the intersection of all pairs of automata from $\AutomataSet_1$ and $\AutomataSet_2$, (as done in \autoref{subsec:intersection}). In this way, any automaton $\mathcal{B}^{\langle J_1,J_2\rangle}$, with $J_1\in\Sigma_{\proc 1}^{\langle c,r\rangle}$ and $J_2\in\Sigma_{\proc 2}^{\langle c,r\rangle}$, accepts the \recursiveLanguage\ $\Lang{\mathcal{A}_1^{J_1}} \cap \Lang{ \mathcal{A}_2^{J_2}}$, and has a size in $\mathcal{O}(|\mathcal{A}_1^{J_1}|\cdot |\mathcal{A}_2^{J_2}|)$. Since $\mathcal{A}_1$ and $\mathcal{A}_2$ are \codeterministic\ and \complete, this ensures that $\mathcal{B}$ is \codeterministic\ and \complete\ too: Let $c \in \calAlpha$ and $r \in \retAlpha$,
\begin{itemize}
\item for all $J_1 \neq J_1'\in\Sigma_{\proc 1}^{\langle c,r\rangle}$ and $J_2 \neq J_2'\in\Sigma_{\proc 2}^{\langle c,r\rangle}$, we have
\[
\Lang{\mathcal{B}^{\langle J_1,J_2\rangle}}\cap \Lang{\mathcal{B}^{\langle J_1',J_2'\rangle}} = \Lang{\mathcal{A}_1^{J_1}} \cap \Lang{ \mathcal{A}_2^{J_2}} \cap \Lang{\mathcal{A}_1^{J_1'}} \cap \Lang{ \mathcal{A}_2^{J_2'}},
\]
which is empty since $\Lang{\mathcal{A}_1^{J_1}} \cap \Lang{\mathcal{A}_1^{J_1'}} = \varnothing$ (and $\Lang{\mathcal{A}_2^{J_2}} \cap \Lang{\mathcal{A}_2^{J_2'}} = \varnothing$), 

\item and the union of the \recursiveLanguage s is universal: 
    \[
    \begin{aligned}
        \bigcup_{\langle J_1 , J_2 \rangle\in \Sigma_{\proc 1}^{\langle c,r\rangle} \times \Sigma_{\proc 2}^{\langle c,r\rangle}} \Lang{\mathcal{B}^{\langle J_1, J_2 \rangle}} 
        &= \bigcup_{\langle J_1 , J_2 \rangle\in \Sigma_{\proc 1}^{\langle c,r\rangle} \times \Sigma_{\proc 2}^{\langle c,r\rangle}}  \Lang{\mathcal{A}^{J_1}}\cap \Lang{\mathcal{A}^{J_2}}\\
        & = \bigcup_{J_1 \in \Sigma_{\proc 1}^{\langle c,r\rangle} } \Lang{\mathcal{A}^{J_1}} \cap \bigcup_{J_2 \in \Sigma_{\proc 2}^{\langle c,r\rangle} } \Lang{\mathcal{A}^{J_2}}\\
        & = \wm\pdwAlph \cap \wm\pdwAlph = \wm\pdwAlph.
    \end{aligned}
    \]
\end{itemize}
Moreover, since the automata of $\mathcal{A}_1$ and $\mathcal{A}_2$ are complete DFAs, this is also the case for all the automata $\mathcal{B}^{\langle J_1,J_2\rangle}$ of $\mathcal{B}$.

Finally, to obtain $\mathcal{B}^S$, as done in \autoref{subsec:intersection}, we first replace the \proceduralTransition s of $\mathcal{A}^S_1$ and $\mathcal{A}^S_2$ by their corresponding ones on the new \proceduralAlphabet\ $\Sigma_{\proc 1} \times \Sigma_{\proc 1}$.  We then compute the DFA $\mathcal{B}^S$ such that it accepts the \regularLanguage\ $\regLang{\mathcal{A}^{S}_1} \cdot \regLang{\mathcal{A}^S_2}$. Since $\mathcal{B}^S$ is a DFA, it can have a size exponential in the size of $\mathcal{A}^{S}_2$: $|\mathcal{B}^S|=\mathcal{O}(|\mathcal{A}_1^S|)\cdot 2^{\mathcal{O}(|\mathcal{A}_2^S|)}$ \cite{yu1994state}. As for \autoref{const:concat}, we get this way a VRA $\mathcal{B}$ accepting $\Lang{\mathcal{A}_1} \cdot \Lang{\mathcal{A}_2}$, and whose size is in $\mathcal{O}(|\mathcal{A}_1|)\cdot 2^{\mathcal{O}(|\mathcal{A}_2|)}$.


\medskip Let us now consider the Kleene-$*$ closure. Suppose that $\mathcal{A}_1$ is \codeterministic\ and \complete, with all its automata being DFAs. 
Following \autoref{const:kleene}, $\mathcal{B}$ stays \codeterministic\ and \complete, except maybe $\mathcal{B}^S$, which is not a DFA. For $\mathcal{B}^S$ being a DFA accepting $\regLang{\mathcal{A}_1^S}^*$,  it has been shown that $|\mathcal{B}^S|=2^{\mathcal{O}(|\mathcal{A}_1^S|)}$ \cite{yu1994state}. Therefore, $\mathcal{B}$ has size $|\mathcal{B}|=2^{\mathcal{O}(|\mathcal{A}_1|)}$.



\medskip 
For both closures of union and intersection, we proceed as for the concatenation closure. We compute the intersection of all pairs of automata from $\AutomataSet_1$ and $\AutomataSet_2$. Moreover, $\mathcal{B}^S$ is obtained with the Cartesian product of $\mathcal{A}_1^S$ and $\mathcal{A}_2^S $ and a proper definition of its set of final states to accept either the union or the intersection. Hence, the resulting VRA $\mathcal{B}$ has size $|\mathcal{B}|=\mathcal{O}(|\mathcal{A}_1|\cdot |\mathcal{A}_2|)$~\cite{yu1994state}.


\medskip
Finally, the construction remains unchanged for the complementation closure. Indeed, it requires first to transform the given VRA into a \codeterministic\ and \complete\ one, with all its automata being DFAs. Since $\mathcal{A}_1$ already satisfies these constraints, then this transformation is not mandatory, and $|\mathcal{B}|=\mathcal{O}(|\mathcal{A}_1|)$.

\section{Proof of \autoref{th:decision-problem}}\phantomsection\label{ax:emptiness-dp}

\Cref{tab:vra-vpa} and \autoref{th:decision-problem} present the complexity of several decision problems on VRAs. In this section, we provide the proof for each problem.

\emptinessdp*

The complexity classes of these problems for VRAs match those for VPAs \cite{alur2004vpl,lange2011p}. Indeed,  \autoref{th:vpa-vra} states that VRAs and VPAs are equivalent under a logspace reduction (see \autoref{ax:vra-vpa} for more details about the complexity of the constructions).
We prove the upper bound time complexities given in \autoref{th:decision-problem} by solving each decision problem individually. 

\subsection{Emptiness Decision Problem}

    Let $\mathcal{A}=(\pdwAlph, \procAlpha, \AutomataSet, \mathcal{A}^S)$ with $\vraInit{A}=\bigcup_{J\in\procAlpha\cup\{S\}} I^J$. Recall the pseudo algorithm to solve the emptiness decision problem:
        \begin{itemize}
        \item \textbf{Initialization:} $\mathit{Reach}_0 = \vraInit{A}$;
        \item \textbf{Main loop:}  Let $\mathcal{J}_i = \{J \in \procAlpha \mid F^J \cap \mathit{Reach}_i \neq \varnothing \}$: $\mathit{Reach}_{i+1} = \mathit{Reach}_i \cup \{p \in Q\mid \exists q\in \mathit{Reach}_i,a\in \intAlpha \cup \mathcal{J}_i : q\xrightarrow{a}p\in \delta \}$; 
        \item \textbf{Output:} When $\textit{Reach}_{i+1}=\textit{Reach}_i$, $\Lang{\mathcal{A}} = \varnothing$ iff $F^S\cap \textit{Reach}_i = \varnothing$. 
    \end{itemize}
\medskip

We note $\mathit{Reach}_*$ (resp.\ $\mathcal{J}_*$) the value of $\mathit{Reach}_i$ (resp.\ $\mathcal{J}_i$) when the algorithm converges, that is, when $\mathit{Reach}_{i+1}=\mathit{Reach}_i$. Since, for all $i\in \mathbb{N}$, $\mathit{Reach}_i \subseteq \vraStates{A}$ and $\mathit{Reach}_i \subseteq \mathit{Reach}_{i+1}$, the convergence to $\mathit{Reach}_*$ is guaranteed when $i=|\vraStates{A}|$. The correctness of the algorithm relies on the following properties.

\begin{restatable}{property}{correctnessEmptiness}\label{pro:correctness-emptiness}
    For all $J\in \procAlpha\cup \{S\}$: $\Lang{\mathcal{A}^J} \neq \varnothing \Leftrightarrow F^J \cap \mathit{Reach}_* \neq \varnothing$.
    \end{restatable}
    
\begin{proof}[]
    Since we use a reachability algorithm on $\intAlpha \cup \mathcal{J}_*$, it is clear that $\textit{Reach}_*$ contains all states $p\in Q$ such that there exists a \regularRun\ on a word $w'\in (\intAlpha\cup \mathcal{J}_*)^*$ from an initial state $q\in \vraInit{A}$ to the state $p$. With that, we prove \autoref{pro:correctness-emptiness}.
    \begin{itemize}
    \item[$\Rightarrow$]
    By contradiction, assume that there exist some $J \in \procAlpha \cup \{S\}$ and a word $w \in \wm\pdwAlph$ such that $w \in \Lang{\mathcal{A}^J}$ and $F^J \cap \mathit{Reach}_* = \varnothing$. We choose such a word $w=u_0c_1w_1r_1\dots c_nw_nr_nu_n\in \Lang{\mathcal{A}^J}$ (with $n\in \mathbb{N}$, $u_i\in\intAlpha^*$, $c_i\in\calAlpha$, $r_i\in\retAlpha$ and $w_i\in\wm\pdwAlph$ for all $i$) of minimal depth. 
    {Suppose that $n=0$, i.e., $w\in \intAlpha^*$, as $w \in \regLang{\mathcal{A}^J}$, it follows that $F^J \cap \mathit{Reach}_* \neq \varnothing$, a contradiction. Therefore $n>0$.} 
    By \autoref{prop:vra-run-to-nfa-trans}, $w'=u_0K_1\dots K_nu_n\in \regLang{\mathcal{A}^J}$, with $K_i \in \procAlpha^{\langle c_i,r_i\rangle}$ and $w_i\in \Lang{\mathcal{A}^{K_i}}$ for all $i\in[1,n]$. Since $\depth{w_i}<\depth{w}$ and $\Lang{\mathcal{A}^{K_i}}\neq \varnothing$
    , by minimality of depth of $w$, we have that $F^{K_i} \cap \mathit{Reach}_* \neq \varnothing$, and thus   $K_i\in \mathcal{J}_*$. This ensures the final state reached at the end of the accepting \regularRun\ on $w'$ to be in $\textit{Reach}_*$. Hence, $F^J \cap \mathit{Reach}_* \neq \varnothing$, which is a contradiction.

    \item[$\Leftarrow$] 


    We first prove by induction on $i\in \mathbb{N}$ that, for all $p\in \mathit{Reach}_i$, there exist $q\in \vraInit{A}$ and  $w\in \wm\pdwAlph$ such that $\langle q,\varepsilon \rangle\xrightarrow{w}\langle p,\varepsilon \rangle\in \runsA{\mathcal{A}}$. 
    \begin{itemize}
    \item Initially, when $i=0$, the property holds with $w=\varepsilon$, as $\mathit{Reach}_0=\vraInit{A}$. 
    \item When $i>0$, by construction of $\mathit{Reach}_i$, there exists an initial state $q\in \vraInit{A}$ and a word $w'=u_0K_1\dots K_n u_n \in (\intAlpha \cup \mathcal{J}_{i-1} )^*$ (with $n\in \mathbb{N}$, $u_j \in \intAlpha^*$ and $K_j\in \mathcal{J}_{i-1}$ for all $j$) such that there exists a \regularRun\ on $w'$ from $q$ to $p$.
    For all $j\in [1,n]$, since $K_j\in \mathcal{J}_{i-1}$, we know that $\mathit{Reach}_{i-1}$ contains a final state $p_j\in F^{K_j}$. By induction, there exist $q_j\in \vraInit{A}$ and $w_j\in \wm\pdwAlph$ such that $\langle q_j,\varepsilon \rangle\xrightarrow{w_j}\langle p_j,\varepsilon \rangle\in \runsA{\mathcal{A}}$. By \autoref{lem:stack-equivalence-vra}, $q_j$ belongs to the same automaton as $p_j$, and we deduce that $q_j\in I^{K_j}$ and $w_j\in \Lang{\mathcal{A}^{K_j}}$. From the regular run on $w'$, by \autoref{prop:vra-run-to-nfa-trans}, we have $\langle q,\varepsilon \rangle \xrightarrow{u_0 c_1w_1r_1 \dots  c_nr_n u_n} \langle p,\varepsilon \rangle \in \runsA{\mathcal{A}}$ with $\langle c_j,r_j\rangle=f(K_j)$ for all $j\in [1,n]$.
    \end{itemize}
    
    Finally, if $F^J\cap \mathit{Reach}_* \neq \varnothing$, we deduce from the previous property the existence of a word accepted by $\mathcal{A}^J$, thus $\Lang{\mathcal{A}^J} \neq \varnothing$. \hfill $\lrcorner$
\end{itemize}
\end{proof}

\begin{algorithm}[]
\caption{Emptiness decision problem for VRAs}\label{alg:emptiness-vra}
$\mathit{Done} = \{\} ;\hspace{1em}\mathcal{J} = \{\};\hspace{1em}\mathit{Reach} \gets \vraInit{A};\hspace{1em} \mathit{Later} \gets \{\}$\;
\While{$\mathit{Reach}  \neq \varnothing$}{
    $q \gets \mathit{Reach} .pop();\hspace{1em}\mathit{Done}.add(q)$\;
    $\mathit{Reach} \gets \mathit{Reach} \cup \{p \in \vraStates{A} \hspace{-.2em} \setminus \hspace{-.2em} \mathit{Done} \mid \exists a \in \intAlpha\cup\mathcal{J}, (q,a,p)\in\vraTrans{A}\}$\;
    $\mathit{Later} \gets \mathit{Later} \cup \{(q, J, p)\in\vraTrans{A} \mid \exists J \in \procAlpha \hspace{-.2em} \setminus \hspace{-.2em} \mathcal{J}\}$\; 
    \If{$q \in F^J \wedge J \in \procAlpha \hspace{-.2em} \setminus \hspace{-.2em} \mathcal{J}$}{
      $\mathit{Reach} \gets \mathit{Reach} \cup \{p \in \vraStates{A} \hspace{-.2em} \setminus \hspace{-.2em} \mathit{Done}   \mid  \exists (q',J,p) \in \mathit{Later}\}$\;
      $\mathcal{J}.add(J)$\;
    }
}
\textbf{return} $\mathit{true}$ \textbf{if} $F^S\cap \mathit{Done} = \varnothing$ \textbf{else} $\mathit{false}$.
\end{algorithm}


It remains to explain how we  compute $\mathit{Reach}_*$ with a time complexity in $\mathcal{O}(|\mathcal{A}|)$. We suppose that, given $q\in Q$ and $a\in \intAlpha\cup \procAlpha$, we can access to all transitions of the form $(q,a,p)\in \vraTrans{A}$ in $\mathcal{O}(1)$, for example 
with a matrix $Q\times (\intAlpha\cup \procAlpha)$ giving the list of those states $p$. \autoref{alg:emptiness-vra} works as follows.
We process each state $q\in \mathit{Reach}$ exactly once. When we process $q$, we add to $\mathit{Reach}$ all states $p$ such that $p$ has not been processed yet and there  exists a transition $(q,a,p)\in \vraTrans{A}$ with $a\in \intAlpha\cup \mathcal{J}$. Additionally, we add to $\mathit{Later}$ all transitions $(q,J,p)\in \vraTrans{A}$ such that $J\notin \mathcal{J}$. Then, if $q$ is a final state of some $\mathcal{A}^J$ with $J\in \procAlpha \setminus  \mathcal{J}$, we add $J$ to $\mathcal{J}$ and we process all transitions of the form $(q',J,p)$ in $\mathit{Later}$ (we suppose that such transitions can be accessed in $\mathcal{O}(1)$, for instance, by defining $\mathit{Later}$ as an hashmap, with $J$ the key and a list of delayed transitions over $J$ as values). When $\mathit{Reach}$ becomes empty, the algorithm has converged, and we know that $\Lang{\mathcal{A}}= \varnothing$ iff no final state of $\mathcal{A}^S$ has been processed. Using this algorithm, each state and internal transition are processed at most once, and each \proceduralTransition\ is processed at most twice. This leads to a time complexity in $\mathcal{O}(|\vraStates{A}|+|\vraTrans{A}|) = \mathcal{O}(|\mathcal{A}|)$.


Note that this complexity is better than the time complexity to solve the emptiness decision problem for VPAs. To solve the latter problem, we need to compute a context-free grammar which accepts the same language as the given VPA, and then apply an emptiness algorithm on it \cite{alur2004vpl}. The computation of the grammar is done in time complexity $\mathcal{O}(|Q|^3 + |\delta|^2)=\mathcal{O}(|\mathcal{A}|^3)$ \cite{sipser1996introduction}
, where $Q$ and $\delta$ are respectively the set of states and the set of transitions of the VPA, and checking the emptiness of the grammar is done in $\mathcal{O}(n)$ \cite{hopcroft2007introduction}, where $n$ is the number of productions of the grammar. The overall algorithm is thus in $\mathcal{O}(|\mathcal{A}|^3)$. 

\subsection{Universality, Inclusion, and Equivalence Decision Problems}


The algorithms for the universality, inclusion, and equivalence decision problems for VRAs are classical. We present them one after the other.
    
Without transforming the VRA $\mathcal{A}$ into a VPA, we describe a simple way to decide whether $\Lang{\mathcal{A}}
    =\wm\pdwAlph$. We first construct a VRA $\mathcal{B}$ such that $\Lang{\mathcal{B}}= \overline{\Lang{\mathcal{A}}}$  (see \autoref{th:closure}). We then decide whether $\Lang{\mathcal{B}}
    =\varnothing$ using the emptiness decision problem (see \autoref{th:decision-problem}), with a total time complexity in $2^{\mathcal{O}(|\mathcal{A}|)}$.

Deciding whether $\Lang{\mathcal{A}_1} \subseteq \Lang{\mathcal{A}_2}$ amounts to deciding whether $\Lang{\mathcal{A}_1}\setminus \Lang{\mathcal{A}_2}=\Lang{\mathcal{A}_1}\cap  \overline{\Lang{\mathcal{A}_2}}=\varnothing$. Using the intersection and complementation constructions (see \autoref{th:closure}), and the emptiness decision algorithm (see \autoref{th:decision-problem}), we obtain a time complexity in $\mathcal{O}(|\mathcal{A}_1|)\cdot 2^{\mathcal{O}(|\mathcal{A}_2|)}$.


Solving the equivalence decision problem is equivalent to solving two inclusion decision problems, that is, whether $\Lang{\mathcal{A}_1} \subseteq \Lang{\mathcal{A}_2}$ and $\Lang{\mathcal{A}_2} \subseteq \Lang{\mathcal{A}_1}$. By \autoref{th:decision-problem}, we get a total complexity in $2^{\mathcal{O}(|\mathcal{A}_1|+|\mathcal{A}_2|)}$. 

\end{document}